\documentclass{CSML}

\def\dOi{12(1:6)2016}
\lmcsheading%
{\dOi}
{1--46}
{}
{}
{Jun.~17, 2015}
{Mar.~14, 2016}
{}

\ACMCCS{[{\bf Theory of computation}]: Models of computation;
  Semantics and reasoning---Program reasoning---Program 
  verification}


\usepackage{amssymb}
\usepackage{amsthm}

\usepackage{ifthen}
\usepackage{centernot}
\usepackage{mathtools}

\newlength{\minlen}
\newlength{\minwlen}
\newlength{\arrowlen}
\newlength{\inputlen}
\newcommand{\Xrightarrow}[5][]{
  \ifthenelse { \equal {#1} {} }
    { }   
    { \settowidth{\minlen}{$\xrightarrow{#1}$} }  
  \settowidth{\inputlen}{$\xrightarrow{#2}$}
  \ifthenelse{\lengthtest{\the\minlen>\the\inputlen}}
           {\setlength{\arrowlen}{\minlen}}
           {\setlength{\arrowlen}{\inputlen}}
    \mathrel{
        \xrightarrow[#5]{\mathmakebox[\arrowlen]{#2}}
      \!\!{}^{\scriptstyle{#3}}_{\scriptstyle{#4}}
    }
  }
\newcommand{\XRightarrow}[5][]{
  \ifthenelse { \equal {#1} {} }
    { }   
    { \settowidth{\minwlen}{$\xRightarrow{#1}$} }  
  \settowidth{\inputlen}{$\xRightarrow{#2}$}
  \ifthenelse{\lengthtest{\the\minwlen>\the\inputlen}}
           {\setlength{\arrowlen}{\minwlen}}
           {\setlength{\arrowlen}{\inputlen}}
    \mathrel{
        \xRightarrow[#5]{\mathmakebox[\arrowlen]{#2}}
      \!\!{}^{\scriptstyle{#3}}_{\scriptstyle{#4}}
    }
  }

\newcommand{\notXrightarrow}[5][]{
  \ifthenelse { \equal {#1} {} }
    { }   
    { \settowidth{\minlen}{$\xRightarrow{#1}$} }  
  \settowidth{\inputlen}{$\xrightarrow{#2}$}
  \ifthenelse{\lengthtest{\the\minlen>\the\inputlen}}
           {\setlength{\arrowlen}{\minlen}}
           {\setlength{\arrowlen}{\inputlen}}
    \mathrel{
        \centernot{
              \xrightarrow[#5]{\mathmakebox[\arrowlen]{#2}}
        }
        \!\!{}^{\scriptstyle{#3}}_{\scriptstyle{#4}}
    }
  }
\newcommand{\notXRightarrow}[5][]{
  \ifthenelse { \equal {#1} {} }
    { }   
    { \settowidth{\minwlen}{$\xRightarrow{#1}$} }  
  \settowidth{\inputlen}{$\xRightarrow{#2}$}
  \ifthenelse{\lengthtest{\the\minwlen>\the\inputlen}}
           {\setlength{\arrowlen}{\minwlen}}
           {\setlength{\arrowlen}{\inputlen}}
    \mathrel{
        \centernot{
          \underset{#5}{
            \xRightarrow{\mathmakebox[\arrowlen]{#2}}
          } 
        }
        \!\!{}^{\scriptstyle{#3}}_{\scriptstyle{#4}}
    }
  }

\newcommand{\Step}[4][]{\Xrightarrow[#1]{#2}{#3}{#4}{}}
\newcommand{\notStep}[4][]{\notXrightarrow[#1]{#2}{#3}{#4}{}}

\newcommand{\step}[2][]{\Step[#1]{#2}{}{}}

\newcommand{\WStep}[4][]{\XRightarrow[#1]{#2}{#3}{#4}{}}

\newcommand{\wstep}[2][]{\WStep[#1]{#2}{}{}}

\usepackage[disable]{todonotes}
\usepackage{xcolor}

\newcommand{\MORE}[1]{{}\todo[inline,backgroundcolor=blue!5,size=\small]{#1}}
\renewcommand{\MORE}[1]{}

\usepackage{varioref}
\usepackage{hyperref}
\usepackage[capitalise]{cleveref}

\newtheorem{theorem}{Theorem}

\newtheorem{lemma}[theorem]{Lemma}
\newtheorem{proposition}[theorem]{Proposition}
\newtheorem{corollary}[theorem]{Corollary}

\theoremstyle{definition}
\newtheorem{definition}[theorem]{Definition}
\newtheorem{example}[theorem]{Example}

\theoremstyle{remark}
\newtheorem{remark}{Remark}
\newtheorem{claim}{Claim}

\crefname{claim}{claim}{claims}
\Crefname{claim}{Claim}{Claims}

\newcommand{\MCOL}[1]{#1} 
\newcommand{\MCOLS}[1]{#1} 
\newcommand{\ABRV}[1]{\text{#1}}

\newcommand{\net}{{\mathcal N}}
\newcommand{\nett}{{\mathcal M}}

\newcommand{\PSPACE}{{PSPACE}}
\newcommand{\EXPSPACE}{{EXPSPACE}}

\newcommand{\R}{Spoiler}
\newcommand{\V}{Duplicator}
\newcommand{\sgame}{\MCOL{simulation game}}  
\newcommand{\sgames}{\MCOL{simulation games}}  
\newcommand{\wsgame}{\MCOL{weak simulation game}}  %
\newcommand{\agame}{\MCOL{approximant game}}  
\newcommand{\slgame}{\MCOL{slope game}}  
\newcommand{\slgames}{\MCOL{slope games}}  
\newcommand{\cycl}{\text{\sc cycle}}
\newcommand{\pref}{\text{\sc prefix}}
\newcommand{\qq}{K}
\newcommand{\lesssteep}{\ll}
\newcommand{\moresteep}{\gg}

\newcommand\slope{\text{\sc slope}}
\newcommand\belt{\text{\sc belt}}

\newcommand\is{L_0}
\newcommand\per{\text{\sc periodic}}
\newcommand \squa[3]{\text{\sc rect}(#1, #2, #3)}
\newcommand\aper{\text{\sc aperiodic}}
\newcommand{\deffect}[1]{\effect{}'(#1)}
\newcommand{\seffect}[1]{\effect{}(#1)}
\newcommand{\dguard}[1]{\guard{}'(#1)}
\newcommand{\sguard}[1]{\guard{}(#1)}

\newcommand{\dmax}{d_\text{max}}

\newcommand{\Cacyc}{\MCOLS{C}}

\newcommand{\effect}[1]{\Delta_{#1}}

\newcommand{\guard}[1]{\Gamma_{#1}}

\newcommand{\prefix}[2]{{}^{#1}#2}
\newcommand{\suff}[1]{{\it suf}({#1})}
\newcommand{\N}{\mathbb{N}}

\newcommand{\Z}{\mathbb{Z}}

\newcommand{\x}{\times}
\newcommand{\card}[1]{|#1|}
\newcommand{\Ord}{\mathit{Ord}}
\renewcommand{\O}{\mathcal{O}}
\newcommand{\sseq}{\sqsubseteq}

\newcommand{\SIMSYMBOL}{\preceq}
\newcommand{\SIMBYSYMBOL}{\succeq}

\newcommand{\SIM}[2]{\ensuremath{\mathrel{\SIMSYMBOL_{#2}^{#1}}}}
\newcommand{\notSIM}[2]{\ensuremath{\mathrel{\not\SIMSYMBOL^{#1}_{#2}}}}
\newcommand{\ssim}{\SIMSYMBOL}

\newcommand{\notSIMBY}[2]{\ensuremath{\mathrel{{^{#1}_{#2}\!\!\not\SIMBYSYMBOL}}}}

\newcommand{\textSSIM}{\MCOL{simulation}}
\newcommand{\TextSSIM}{\MCOL{Simulation}}

\newcommand{\WEAKSIMSYMBOL}{{\raisebox{+0.3ex}{$\preceq$}\kern -.68em
            \raisebox{-0.5ex}{$\rule{.59em}{.4pt}$}}}
\newcommand{\WSIM}[2]{\ensuremath{\mathrel{\WEAKSIMSYMBOL^{#1}_{#2}}}}
\newcommand{\notWSIM}[2]{\ensuremath{\mathrel{\centernot{\WEAKSIMSYMBOL}\kern -.05em{}^{#1}_{#2}}}}
\newcommand{\wsim}{\WEAKSIMSYMBOL}
\newcommand{\textWSIM}{\MCOL{weak simulation}}

\newcommand{\textTINC}{\MCOL{trace inclusion}}

\newcommand{\Act}{\text{Act}}

\newcommand{\Net}[1]{{\ensuremath{\mathcal #1}}}

\newcommand{\AN}[1]{{\ensuremath{\mathcal #1}}}
\newcommand{\NN}[1]{{\ensuremath{\mathcal #1}}}

\newcommand{\NS}[1]{{\ensuremath{\mathcal #1}}}
\newcommand{\ND}[1]{{\ensuremath{\mathcal #1}'}}
\newcommand{\source}[1]{\mathit{source}(#1)}
\newcommand{\target}[1]{\mathit{target}(#1)}

\newcommand{\VASS}{\MCOL{\ABRV{VASS}}}

\newcommand{\textVASSs}{\MCOL{vector addition systems with states}}

\newcommand{\textOCSs}{\MCOL{one-counter systems}}

\newcommand{\OCA}{\MCOL{\ABRV{OCA}}}
\newcommand{\OCAs}{\MCOL{\ABRV{OCAs}}}
\newcommand{\textOCA}{\MCOL{one-counter automaton}}

\newcommand{\TextOCAs}{\MCOL{One-counter automata}}

\newcommand{\OCN}{\MCOL{\ABRV{OCN}}}
\newcommand{\OCNs}{\MCOL{\ABRV{OCNs}}}
\newcommand{\textOCN}{\MCOL{one-counter net}}
\newcommand{\textOCNs}{\MCOL{one-counter nets}}

\newcommand{\TextOCNs}{\MCOL{One-counter nets}}

\newcommand{\PDA}{\MCOL{\ABRV{PDA}}}

\newcommand{\LTS}{\MCOL{\ABRV{LTS}}}

\newcommand{\case}[2][Case]{\MCOL{\emph{#1} #2.}}

\newcommand{\ignore}[1]{}
\newcommand{\ol}[1]{\overline{#1}}

\title[Simulation Problems Over One-Counter Nets]
      {Simulation Problems Over One-Counter Nets}

 \author[P.~Hofman]{Piotr Hofman\rsuper a}
 \address{{\lsuper a}
         LSV, CNRS \& ENS de Cachan,
         61 avenue du Pr\'{e}sident Wilson,
         94235 CACHAN Cedex, France 
 }
 \urladdr{http://www.mimuw.edu.pl/\~{}ph209519/}
 \thanks{{\lsuper a}Supported by Labex Digicosme, Univ. Paris-Saclay, project VERICONISS}

 \author[S.~Lasota]{S{\l}awomir Lasota\rsuper b}
 \address{{\lsuper b}University of Warsaw, 
     Faculty of Mathematics, Informatics and Mechanics,
     Banacha 2,
     02-097 Warszawa,
     Poland
 }
 \urladdr{https://www.mimuw.edu.pl/\~{}sl/}
 \thanks{{\lsuper b}Supported by NCN grant 2013/09/B/ST6/01575.}

 \author[R.~Mayr]{Richard Mayr\rsuper c}
 \address{{\lsuper c}University of Edinburgh, School of Informatics,
10 Crichton Street, Edinburgh EH89AB, UK}
 \urladdr{http://www.inf.ed.ac.uk/people/staff/Richard\_Mayr.html}
 \thanks{{\lsuper c}Supported by EPSRC grant EP/M027651/1.}

 \author[P.~Totzke]{Patrick Totzke\rsuper d}
 \address{{\lsuper d}University of Warwick,
     Department of Computer Science,
     Coventry CV47AL, UK}
 \urladdr{http://www.dcs.warwick.ac.uk/\~{}totzke/}
 \thanks{{\lsuper d}Supported by EPSRC grant EP/M011801/1.}


\begin{document}

\maketitle

\begin{abstract}
\noindent
One-counter nets (\OCN) are
finite automata equipped with a counter that can store non-negative integer
values, and that cannot be tested for zero.
Equivalently, these are exactly $1$-dimensional \textVASSs.
We show that both strong and weak simulation preorder on \OCN\ are \PSPACE-complete.
\end{abstract}

\section{Introduction}\label{sec:introduction}



\TextOCAs\ (\OCAs) are Minsky counter automata with only one counter,
and they can also be seen as a subclass of pushdown automata with just one
stack symbol (plus a bottom symbol).
\TextOCNs\ (\OCNs) are a subclass of \OCAs\ where the counter
cannot be fully tested for zero, because steps enabled at counter value
zero are also enabled at nonzero values.
\OCNs\ correspond to
$1$-dimensional vector addition systems with states,
and are arguably the simplest model of discrete infinite-state systems,
except for those that do not have a global finite control.

Notions of behavioral semantic equivalences have been classified in 
Van Glabbeek's linear time - branching time spectrum \cite{Gla2001}.
The most common ones are, in order from finer to coarser, 
bisimulation, simulation and trace equivalence.
Each of these have their standard (called strong) variant, and a weak variant
that abstracts from arbitrarily long sequences of internal actions.

For \OCAs/\OCNs, strong bisimulation 
is \PSPACE-complete \cite{BGJ2010, DBLP:journals/jcss/BohmGJ14}, 
while weak bisimulation is undecidable \cite{May2003}.
Strong trace inclusion is undecidable for \OCAs\ \cite{Valiant1973},
and even for \OCNs\ \cite{HMT2013}, and this trivially carries over to weak trace
inclusion.

The picture is more complicated for simulation preorders.
While strong and weak simulation are undecidable for \OCAs\ \cite{JMS1999},
they are decidable for \OCNs.
Decidability of strong simulation on \OCNs\ was first proven in \cite{AC1998},
by establishing that the simulation relation follows a certain regular pattern.
This idea was made more graphically explicit in later proofs \cite{JM1999,JKM2000},
which established the so-called {\em Belt Theorem}, that states that the simulation
preorder relation on \OCNs\ can be described by finitely many partitions of the
grid $\N\x\N$, each induced by two parallel lines.
In particular, this implies that the simulation relation is semilinear.
However, the proofs in \cite{AC1998,JM1999,JKM2000} did not yield any upper
complexity bounds; in particular, the first proof was based on two semi-decision procedures
and the later proof of the Belt Theorem was non-constructive. 
A \PSPACE\ lower bound for strong simulation on \OCNs\ follows from \cite{Srb2009}.

Decidability of weak simulation on \OCNs\ was shown in \cite{HMT2013}, using 
a converging series of semilinear approximants. This proof used the 
decidability of strong simulation on \OCNs\ as an oracle, and thus did not immediately yield any 
upper complexity bound.

\subsubsection*{Our contribution.}
First, we provide a new constructive proof of the Belt Theorem and derive a \PSPACE\ algorithm
for checking strong simulation preorder on \OCNs.
Together with the lower bound from \cite{Srb2009}, this shows \PSPACE-completeness of the problem.

Second, via a technical adaption of the algorithm for weak simulation in \cite{HMT2013},
and the new \PSPACE\ algorithm for strong simulation, we also obtain a \PSPACE\ algorithm
for weak simulation preorder on \OCNs. Thus even weak simulation preorder on
\OCNs\ is 
\PSPACE-complete.

The decidability and complexity status of the most relevant semantic equivalences and preorders 
for \OCAs/\OCNs\ is summarized in the table below (`$\times$' stands for undecidable).
Our PSPACE-completeness results close the last remaining important open problem concerning the complexity of 
equivalence/preorder checking for \textOCSs.

\begin{table}[h]
\newcommand{\undec}{$\times$}  
\begin{center}
  \begin{tabular}{ | l | c | c| c | c  | c | }
    \hline
    	 	& simulation 		& bisimulation 			& weak sim.	& weak bisim.		 & trace inclusion \\ \hline
    OCN 	& \textbf{\PSPACE}	& \PSPACE\ \cite{BGJ2010} 	&\textbf{\PSPACE}	& \undec\ \cite{May2003}	 &\undec\ \cite{HMT2013}\\ \hline
    OCA 	& \undec\ \cite{JMS1999}	&\PSPACE\ \cite{BGJ2010}	 & \undec\ \cite{JMS1999}	 & \undec\ \cite{May2003} 	&\undec\ \cite{Valiant1973}\\
    \hline
  \end{tabular}
\end{center}
\end{table}

This paper is a revised and extended version of material previously presented
in \cite{HMT2013,HLMT2013,Totzke:PhD2014}, and is 
organized as follows. In \cref{sec:problem} we state the
simulation problems and our main result, and give an outline of the ideas used in the proof.
In Section~\ref{sec:prelim} we fix basic terms and notation, and show how to
transform the problem into a more convenient normal form.
The proof of \PSPACE-completeness for strong simulation preorder, 
as well as an analysis of the combinatorial structure of this relation,
is presented in \Cref{sec:strongSim}.
We then apply and extend this result in Section~\ref{sec:wsim} to show
\PSPACE-completeness even for weak simulation preorder.
Finally, in Section~\ref{sec:conclusion}, we summarize our results and
mention some open problems.

\section{Statement of the Result} \label{sec:problem}
A labeled transition system (LTS) over a finite alphabet $\Act$ of actions
consists of a set $\mathcal{S}$ of configurations (also called \emph{processes}) and, for every action $a \in \Act$, a binary relation
$\step{a} \subseteq \mathcal{S}^2$ between configurations.
For $(s, s') \in \, \step{a}$ we also write $(s, a, s')$ or $s \step{a} s'$,
and call it an $a$-labeled \emph{step} from $s$ to $s'$.

\begin{definition}\label{def:strongsim}
Given two labeled transition systems $S$ and $S'$, a relation $R$ between the configurations
of $S$ and $S'$ is a \emph{strong simulation} if for every 
pair of configurations $(c, c') \in R$ and every step $c \step{a} d$
there exists a step $c' \step{a} d'$ such that $(d,d') \in R$.
\end{definition}

As usual, w.l.o.g.~one may assume $S=S'$, since one can consider disjoint union of two LTSs.
Strong simulations are closed under union, so there exists a unique maximal strong simulation.
This maximal strong simulation is a preorder, called {\em strong simulation preorder},
and denoted by $\ssim$. If $c \ssim c'$ then one says that $c'$ \emph{strongly simulates} $c$.


Simulation preorder can also be characterized
as an interactive, two-player game played between \emph{\R},
who wants to establish non-simulation
and \emph{\V}, who wants to frustrate this.

\begin{definition}
    \label{def:simgame}
    A \emph{simulation game} is played in rounds between the two players
    \R\ and \V, where the latter tries to stepwise match the moves of the former.

    A \emph{play} is a finite or infinite sequence of
    game \emph{positions}, which are pairs of processes.
    If a finite play $(\alpha_0,\alpha'_0),(\alpha_1,\alpha'_1),\dots,(\alpha_i,\alpha'_i)$
    is not already winning for one of the players,
    the next pair $(\alpha_{i+1},\alpha'_{i+1})$
    is determined by a round of choices:
    \begin{enumerate}
      \item \R\ chooses a step $\alpha_i\step{a}\alpha_{i+1}$ where $a$ is any element of $\Act$.
      \item \V\ responds by picking an equally labeled step
            $\alpha'_i\step{a}\alpha'_{i+1}$.
    \end{enumerate}
    If one of the players cannot
    move then the other wins, and \V\ wins every infinite play.

    A \emph{strategy} is a set of rules that tells a player how to move.
    More precisely, a strategy for \R\
    is a function $\sigma:PP \to (\step{})$,
    where $PP$ denotes the set of partial plays (non-empty sequences of game
    positions), and $\step{}$ is the step-relation in the transition system.
    Similarly, a strategy for \V\ is a function $\sigma':PP\x(\step{}) \to
    (\step{})$,
    assigning each partial play and \R\ move a response.
    A player plays according to a strategy if all his moves obey the rules of the strategy.
    A strategy is \emph{winning} from $(\alpha,\alpha')$ if every play that starts in $(\alpha,\alpha')$
    and which is played according to that strategy is winning.
    Finally, we say that a player wins the simulation game from $(\alpha,\alpha')$
    if there is some winning strategy for this player from position $(\alpha,\alpha')$.

\end{definition}
Due to the type of winning condition 
(a simulation game is essentially a turn-based reachability game where
Spoiler
wins a play if it reaches a game configuration where Duplicator is stuck) 
positional (i.e., memoryless) strategies are sufficient. Thus one can
restrict to strategies that map the current game configuration to a step,
i.e. $\sigma:(\mathcal{S}\x \mathcal{S}) \to (\step{})$. 
Correspondingly, a strategy for \V\ is a partial function
$\sigma':(\mathcal{S}\x \mathcal{S}\x \step{}) \to (\step{})$,
that prescribes a response for the current position and \R's move.


We see that one round of the \sgame\ directly corresponds to the simulation condition of
\cref{def:strongsim}.
\R\ can stepwise demonstrate
that the condition is not an invariant if the initial pair of processes is indeed not in
simulation.
Conversely, any simulation that contains the initial pair of processes
prescribes a winning strategy for \V\ in the \sgame.
\begin{proposition}
    \label{prop:sgames}
    For any two processes $\alpha,\alpha'\in \mathcal{S}$,
    \V\ has a winning strategy in the simulation game from position
    $(\alpha,\alpha')$ if and only if
    $\alpha\ssim \alpha'$.
\end{proposition}

\medskip
A natural extension of simulation is  \emph{weak simulation}, that abstracts from internal steps.
For a labeled transition system with a special action $\tau\in\Act$, define \emph{weak} step relations by
\[
\wstep{\tau} \ \ = \ \ (\step{\tau})^* \qquad \text{ and } \qquad \wstep{a} \ \ = \ \ (\step{\tau})^*\step{a}(\step{\tau})^* \quad \text{ for } a \neq \tau.
\]
%
Weak simulation is defined similar to strong simulation 
in \cref{def:strongsim}, except that the weak simulation condition requires
that some weak step exists. Formally:
\begin{definition}\label{def:weaksim}
Given two labeled transition systems $S$ and $S'$, a relation $R$ between the configurations
of $S$ and $S'$ is a \emph{weak simulation} if for every 
pair of configurations $(c, c') \in R$ and every step $c \step{a} d$
there exists a weak step $c' \wstep{a} d'$ such that $(d,d') \in R$.
\end{definition}
Weak simulation preorder can also be characterized using a variant of the
simulation game described above, in which \V\ moves along weak steps.
This game is called the \emph{\wsgame}.
Yet another variant of this game,
in which also \R\ moves along weak steps, induces the same notion of weak
simulation preorder.
We will use the ``asymmetric'' game define above in this paper.

For systems without $\tau$-labeled steps, $\step{a} \ = \ \wstep{a}$
holds for every action $a$,
and therefore strong and weak simulation coincide.
In general however, weak simulation is coarser than strong simulation: $c \ssim c'$ implies $c \ \wsim \ c'$.


\begin{definition}[One-Counter Nets]
    A \emph{\textOCN} (\OCN) is a triple $\net=(Q,\Act,\delta)$
    consisting of finite sets of control states $Q$, action labels $\Act$
    and transitions $\delta\subseteq Q\x \Act\x\{-1,0,1\}\x Q$.
    Each transition $t=(p,a,d,q)\in\delta$
    defines a relation $\step{t}\subseteq Q\x\N\x Q\x\N$
    where for all control states $p',q'\in Q$ and integers $m,n\in\N$
    $$(p',m)\step{t}(q',n)\quad\text{if}\quad p'=p, q'=q\text{ and }n=m+d\ge0.$$
    The labeled transition system \emph{induced} by the \OCN\
    has the same action alphabet $\Act$ 
    and the set of configurations $\mathcal{S} = Q\x\N$.
    Its step relations $\step{a}$ are defined as follows.
    We have $(p,m)\step{a}(q,n)$ iff $\exists t=(p,a,d,q)\in\delta.\, (p,m)\step{t}(q,n)$. 
\end{definition}

    In the sequel we use both the relations $\step{t}$ labeled by transitions $t$, and the relations $\step{a}$ labeled by actions $a\in\Act$.
    For convenience, we will assume that $Q\cap\N=\emptyset$ and
    write configurations $(p,m)$ simply as $pm$.
    On the formal level, steps should not be confused with transitions:
    there is a \emph{step} $pm\step{a}qn$ iff there is a \emph{transition} $(p,a,d,q)\in\delta \text{ and }n=m+d\ge0$.

    We will sometimes simply write \emph{\OCN\ process} for a configuration in
    the \LTS\ induced by some \OCN.

\begin{example}
    Let $\net=(\{p\},\{a,\tau\},\{(p,a,-1,p),(p,\tau,+1,p)\})$
    be the \OCN\ consisting of a single state with two self-looping transitions:
    One is labeled by $a$ and is counter decreasing,
    and the other is labeled by $\tau$ and increases the counter.
    In this system,
    $pn$ is simulated by $pm$ ($pn\ssim pm$) if, and only if $n\le m$.
    However, $pn\ \wsim\ pm$ holds for all $n,m\in\N$ because of the
    weak steps $pm\wstep{a}pm'$ for every $m'\ge (m-1)$.
\end{example}

\noindent We study the computational complexity of the following decision problem.
      
\begin{table}[h]
\begin{tabular}{ll}
  \multicolumn{2}{l}{\bf Weak Simulation Problem for \OCNs}\\
  \hline
  \sc Input:  & Two \OCNs\ $\net$ and $\net'$ together with configurations $qn$ and $q'n'$\\
              & of $\net$ and $\net'$, respectively, where $n$ and $n'$ are given in binary.\\
  \sc Question: & $qn \ \wsim \ q'n'$ ?
\end{tabular}
    \end{table}

\noindent The main result of this paper is the following upper bound.

\begin{theorem}\label{thm:strongsim-pspace}
    The weak simulation problem for \OCNs\ is in \PSPACE.
\end{theorem}

\begin{remark}
The upper bound applies also to strong simulation, since for systems without
$\tau$-labeled steps, strong and weak simulation coincide.
Combined with the \PSPACE-hardness result for strong simulation
by~\cite{Srb2009} (which holds even if all numbers are represented in unary),
this yields \PSPACE-com\-ple\-te\-ness of 
both strong and weak simulation problems.
\end{remark}

\begin{remark}\label{rem:semilinearsize}
Our construction can also be used to compute the simulation relation as a semilinear
set, but its description requires exponential space. 
However, checking a point instance $qn \ \wsim \ q'n'$ of the simulation problem
can be done in 
polynomial space by stepwise guessing and verifying only a polynomially bounded part of the relation.
\end{remark}

\subparagraph*{Outline of the proof.}
In LTSs induced by \OCNs, the step relation is monotone w.r.t.\ the counter value.
Thus, the strong and weak \sgames\ are also monotone for both players:
If \V\ wins from a position $(qn, q'n')$ then he also wins from
$(qn, q'm)$ for all $m > n'$. Similarly, if \R\ wins from $(qn, q'n')$ then she also wins
from $(qm, q'n')$ for all $m > n$.
It follows that, for every fixed pair $(q,q')$ of control states, 
the winning regions of the two players partition the grid $\N\x\N$ into two \emph{connected} subsets.
For strong simulation, it is known \cite{JM1999,JKM2000} that the \emph{frontier} between these two subsets
is contained in a \emph{belt}, i.e., it lies between two parallel lines with a
rational slope. This property is also known as the Belt Theorem.
However, previous proofs of this theorem \cite{JM1999,JKM2000} used
non-constructive arguments and did not yield precise bounds on the width of
the belt and on the rational coefficients of the slope.

We provide a new constructive proof of the Belt Theorem that yields tight
bounds on the width and slopes of the belts, which makes it possible to obtain
a \PSPACE\ algorithm for checking strong simulation preorder.
Our proof is based on the analysis of symbolic \emph{\slgames}. This new
game is similar to the \sgame, but necessarily ends after a polynomial number of
rounds. We show that, for sufficiently high counter-values, both
players can re-use winning strategies from the \slgame\ also in the \sgame.
As a by-product of this characterization, we obtain polynomial bounds on the widths and
slopes of the belts. Once the belt-coefficients are known, one can compute the
frontiers between the winning sets of the opposing players exactly,
because every frontier necessarily adheres to a regular pattern.

In the second part of the paper (\cref{sec:wsim}) we prove the decidability of weak simulation
preorder by showing that it is the limit of a finitely converging series of
effectively constructible semilinear relations that over-approximate it.
A careful analysis of the size of the representations of these approximants,
combined with the previously established \PSPACE\ algorithm for strong simulation
preorder, then yields a \PSPACE\ algorithm for checking weak simulation
preorder.\newpage

\section{Preliminaries}\label{sec:prelim}
\newcommand{\Path}[1]{\Step{#1}{}{}}
\newcommand{\PPath}[1]{\Step{#1}{}{+}}

\subsection{Paths and Loops}
    \label{def:oca-paths}
    Let $\net=(Q,\Act,\delta)$ be a \OCN.
    For a transition $t=(p,a,d,p')\in \delta$    we write $\source{t}=p$ and $\target{t}=p'$
    for the source and target states,    $\lambda(t)=a$ for its label
    and $\Delta(t)=d$ for its effect on the counter.
    
    A \emph{path} (of \emph{length} $k$) in $\AN{N}$ is a sequence
    $\pi=p_0t_1p_1t_2p_2\dots p_{k-1}t_kp_k$
    where all $p_i\in Q$ and $t_i\in\delta$ and for every $1\le i\le k$,
    $p_{i-1}=\source{t_i}$ and $\target{t_i}=p_{i}$.
    The source and target of $\pi$ are
    $p_0$ and $p_k$, respectively.
    Its label is $\lambda(\pi)=\lambda(t_1)\lambda(t_2)\dots\lambda(t_k)\in\Act^*$
    and its \emph{effect}
    is the cumulative effect of its transitions:
    \begin{equation}
        \effect{}(\pi) = \sum_{i=1}^k \effect{}(t_i)
    \end{equation}
    A path $\pi$ as above     is a \emph{cycle} if $p_0=p_k$ and a
    \emph{simple cycle} if it is a cycle and moreover, no proper subpath is itself a cycle.

    We say a path $\pi$ is \emph{enabled} in configuration $pm$ if
    it prescribes a valid path from configuration $pm$ in the labeled transition system of $\AN{N}$, i.e.,
    if there exist non-negative integers $m_0,m_1,\dots,m_k$ such that $p_0m_0=pm$ and
    $p_{i-1}m_{i-1}\step{t_i}p_{i}m_{i}$ for all $1\le i\le k$.
    In this case we write $p_0m_0\step{\pi}p_{k}m_{k}$ and say $\pi$ is a \emph{run}
    or path \emph{of} \AN{N} from $p_0m_0$ to $p_km_k$.
    Note that $m_k=m_0+\effect{}(\pi)$.

There is a minimal sufficient counter value $\guard{}(\pi)$ that enables it.
This \emph{guard} of $\pi$ can be defined as 
the minimal $m\in\N$ such that no
prefix of $\pi$ has an effect less than $-m$.
Writing $\prefix{i}{\pi}$ for the prefix of path $\pi$ of length $i$,
the guard of $\pi$ is given as
\begin{equation}
    \guard{}(\pi) = - \min\{\effect{}(\prefix{i}{\pi})\ |\ 0\le i\le k\}.\label{eq:oca:guards}
\end{equation}
Note that there are different paths of length $0$ because the initial
state forms part of a path.
Any zero-length path $\pi$
has effect and guard $\effect{}(\pi)=\guard{}(\pi)=0$.
Surely, both the effect and the guard of any path are bounded by its length.

\subsection{Monotonicity}

\TextOCNs\ enjoy the following important monotonicity property
which is crucial in our argument and which immediately follows from the definition.

A step $pm\Step{a}{}{}qn$ in a \OCN\ $\Net{N}=(Q,\Act,\delta)$
is due to some transition $(p,a,d,q)\in\delta$ with $d=n-m$.
The same transition then justifies a step $p(m+l)\step{a}q(n+l)$
for any number $l\in\N$. We thus observe that for all \OCN\
processes $pm$ and $l\in\N$,
\begin{equation}
    pm \ssim p(m+l)\label{eq:ocn:cc}
\end{equation}
because \V\ can mimic
the behavior of \R's process to win the \sgame.
Seen as a function, this ``copycat'' strategy is simply the identity.
Seen as a tree, it has the property that
every node is of the form $[qn,q(n+l)]$, where $q\in Q$ and $n\in\N$.

\Cref{eq:ocn:cc} implies that on \OCNs, all preorders that are coarser than $\ssim$,
the maximal strong simulation,
are monotonic in the following sense.

\begin{lemma}[Monotonicity]
    \label{prop:ocn-monotonicity}
    \label{lem:monotonicity}
    Let $pm$ be a \OCN\ process,
    $s$ an arbitrary process
    and $\sseq$ be any transitive relation
    that subsumes strong simulation $\ssim$.
    Then, for every $m \le n$,
    \begin{enumerate}
      \item $pm\not\sseq s$ implies $pn\not\sseq s$, and
      \item $s\sseq pm$ implies $s\sseq pn$.
    \end{enumerate}
\end{lemma}
\begin{proof}
    By \cref{eq:ocn:cc} we have $pm\ssim pn$ and thus $pm\sseq pn$.
    The claim directly follows from this observation and the transitivity of $\sseq$.
\end{proof}
  The above monotonicity property holds in particular for $\sseq$ being strong or \textWSIM,
  \textTINC\ or any approximating relation $\sseq_\alpha$ defined later
  in this paper. 

The following is a direct consequence of \cref{lem:monotonicity} that we state here only
because we are particularly interested in simulation games played on \OCN s.
\begin{corollary}\label{lem:ocn-ocn-monotonicity}
    Let $pm$ and $p'm'$ be two \OCN\ processes and $\sseq$ be any transitive
    relation that subsumes strong simulation. Then
    $pm\sqsubseteq p'm'$ implies $pn\sqsubseteq p'n'$ for all $n\le m$ and $m'\le n'$.
\end{corollary}

\subsection{Product Graphs}

When we consider simulation games played on \LTS\ induced by \OCN,
it is convenient to identify individual plays
with paths in the \emph{synchronous product}
of the two given \OCNs. In later constructions we will in particular
be interested in the effects of cyclic paths in this product.

    The \emph{product graph}
    of two \OCNs\ $\net=(Q,\Act,\delta)$ and    $\net'=(Q',\Act,\delta')$
    is the finite, edge-labeled graph with nodes $V=Q\x Q'$ and
    edges
    \begin{equation*}
    E \ = \ \{ (t, t') \in  \delta \times \delta' \ : \ \lambda(t)=\lambda(t')\}.
    \end{equation*}
    A \emph{path} in the product is a sequence
    $\xi=v_0T_1v_1T_2v_2\dots v_{k-1}T_kv_{k}$.
    As $\xi$ is a sequence of pairs (each $v_i$ is a pair of states in $Q\x Q'$
    and each $T_i\in E \subseteq \delta\x\delta'$
    is a pair of transitions)
    we can naturally speak of its two projections, $\pi$ and $\pi'$,
    which are paths in $\net$ and $\net'$, respectively.
    The path $\xi$ is \emph{enabled in $(pm,p'm')$}
    if both $\pi$ is enabled in $pm$ and $\pi'$ is enabled in $p'm'$.
    In this case we write $(pm,p'm')\step{\xi}(qn,q'n')$ to mean that both
    $pm\step{\pi}qn$ and $p'm'\step{\pi'}q'n'$.
    
    We write $T\in\xi$ if $T=T_i$ for some index $1\le i\le k$. 
    
    The $source$, and $target$ of paths in \OCN\ are lifted to paths in products in a natural way:
    We define $\source{\xi}=(\source{\pi},\source{\pi'})$,
    $\target{\xi}=(\target{\pi},\target{\pi'})$.
    We write $\seffect{\xi} = \effect{}(\pi)$
    and $\sguard{\xi}=\guard{}(\pi)$
    as well as
    $\deffect{\xi}=\effect{}(\pi')$
    and $\dguard{\xi}=\guard{}(\pi')$.

    A nonempty path $\xi$ is a \emph{cycle} if $\source{T_1}=\target{T_k}$.
    It is a \emph{simple cycle} or \emph{loop}
    if it is a cycle but none of its proper subpaths is a cycle.
    
    A \emph{lasso} is a path that contains a cycle while none of its strict prefixes does.
    \label{def:lasso}
    That is, a path $\xi$ as above is a lasso if there exists
    $l\le k$ such that $\target{T_k}=\source{T_l}$ and for all $1\le i\le j< k$,
    $\target{T_j}\neq \source{T_i}$. 
    A lasso $\xi$ naturally splits into
    $\pref(\xi)=v_0T_1v_1T_2\dots T_{l-1}v_l$
    and $\cycl(\xi) = v_lT_lv_{l+1}T_{l+1} \dots T_kv_k$.

\subsection{Normal Form}

We prove a simple normal-form theorem (\cref{thm:oca.nf}) for simulation games on \OCNs,
that essentially states that \R\ can only win if she forces \V\ to empty his counter.

\begin{definition}
    \label{def:complete}
    \label{def:normal-form}
    A \OCN\ $\AN{N}=(Q,\Act,\delta)$ is \emph{complete} if for every state $p\in Q$ and every action
    $a\in\Act$, there exists at least one transition $(p,a,d,q)\in\delta$.
    It is \emph{non-blocking} if none of its processes is a deadlock,
    i.e., if for every state $p\in Q$ there is some transition $(p,a,d,q)\in\delta$ with $d\in\{0,1\}$.
    
    A pair $\net, \net'$ of \OCNs\ is in \emph{normal form} if $\net$ 
    is non-blocking
    and $\net'$ is complete.
\end{definition}

\begin{lemma}
    \label{thm:oca.nf}
    \label{lem:normal-form}
    For any two \OCNs\ $\net=(Q,\Act,\delta)$    and $\net'=(Q',\Act',\delta')$,
    one can compute in logarithmic space a    pair $\nett,\nett'$ of \OCNs\ in normal form
    with sets of control states $Q$ and $S\supseteq Q'$, respectively, 
    such that for all $(q, n, q',n') \in (Q\x\N\x Q'\x\N)$
    and for every $\sqsubseteq{}\in\{\ssim,\wsim\}$
    it holds that
    \begin{equation}
        qn \sqsubseteq q'n' \text{ w.r.t.\ } \net,\net' \iff qn \sqsubseteq q'n' \text{ w.r.t.\ } \nett,\nett'.
    \end{equation}
\end{lemma}
\begin{proof}
    We pick a new action label
    $\$\not\in \Act$
    and turn $\net$ into a non-blocking net $\nett$ by adding 
    $\$$-labeled cycles with effect $0$ to all states:
    $\nett = (Q,\Act\cup\{\$\}, \ol{\delta})$ with
    $\ol{\delta} = \delta \cup\{(s,\$,0,s) \mid s\in Q\}$.
    To compensate for this, we add $\$$-cycles to all states of $\net'$ in the same way.
    To complete the second net, add a sink state $L$ (for ``losing''),
    which has counter-decreasing cycles for all actions, including $\$$ action,
    and connect all states without outgoing
    $a$-transitions to $L$ by $a$-labeled transitions.

    Assume \R, playing on $\net$, wins the (weak) \sgame\ against \V\
    playing on $\net'$.
    In the game on $\nett$ and $\nett'$, \R\ can move according to
    a winning strategy in the original game and thus force
    a play ending in a position $(pm,p'm')$
    that is immediately winning in the game on $\net$ and $\net'$,
    i.e., $pm\step{a}rl$ but $p'm'\notStep{a}{}{}$ for some action $a$.
    Thus
    the game on $\nett,\nett'$ continues in the position $(rl, Lm')$, which is clearly winning for
    \R\ because she can exhaust her opponent's counter and win using finitely many $\$$-moves.

    Conversely, if \V\ wins the (weak) \sgame\ on $\net$ and $\net'$ this means 
    that each play is either infinite or ends in a position $(pm,p'm')$ where $pm\notStep{a}{}{}$ for all
    actions $a$. In
    the game on $\nett$ and $\nett'$, the latter case means that \R\ has no choice but to make
    $\$$-moves indefinitely, which is losing for her.
\end{proof}

\Cref{thm:oca.nf} allows us to focus w.l.o.g.~on instances of the (weak) \textSSIM\ problems where the given systems are normalized.
In particular,
\R\ cannot get stuck and only loses infinite plays, and
\V\ can only be stuck (and lose the game) if his counter equals zero.
Therefore, every branch in any winning strategy for \R\ ends in a position where \V\ has counter value $0$.

\section{Strong Simulation}\label{sec:strongSim}
In this section 
we consider strong simulation $\ssim$ only
and therefore write shortly `simulation preorder', `simulation game', etc.~instead of strong simulation preorder/game.

Let us fix two \OCN\ $\net$ and $\net'$, with sets of control states $Q$ and
$Q'$, respectively. Following \cite{JKM2000,JM1999},
we interpret $\SIM{}{}$ as a 2-coloring of $\qq=|Q\x Q'|$ Euclidean planes, one for each
pair of control states $(q,q')\in Q\x Q'$. As proposed by Jan\v{c}ar and
Moller \cite{JM1999}, every
pair of configurations $(qn, q'n')$ is represented by the unique point
$(n,n')$ on the plane for the pair of control states
$(q,q')$. If $qn \SIM{}{} q'n'$ then the point is colored with color-$\SIM{}{}$ and 
otherwise with color-$\notSIM{}{}$. This graphical perspective on
the simulation relation is very helpful in many parts of the proof.

The main combinatorial insight of \cite{JKM2000}
(this was also present in \cite{AC1998}, albeit less explicitly)
is the so-called \emph{Belt Theorem}, that states that
each such plane can be cut into segments by two parallel lines such that
the coloring of $\SIM{}{}$ in the outer two segments is constant; see Figure~\ref{fig:belt}.
We provide a new constructive proof of this theorem, stated as
Theorem~\ref{thm:belt-theorem-bounds} below, that allows us to derive polynomial bounds on
the coefficients of all belts.

\begin{figure}[h]
  \begin{center}
      \includegraphics{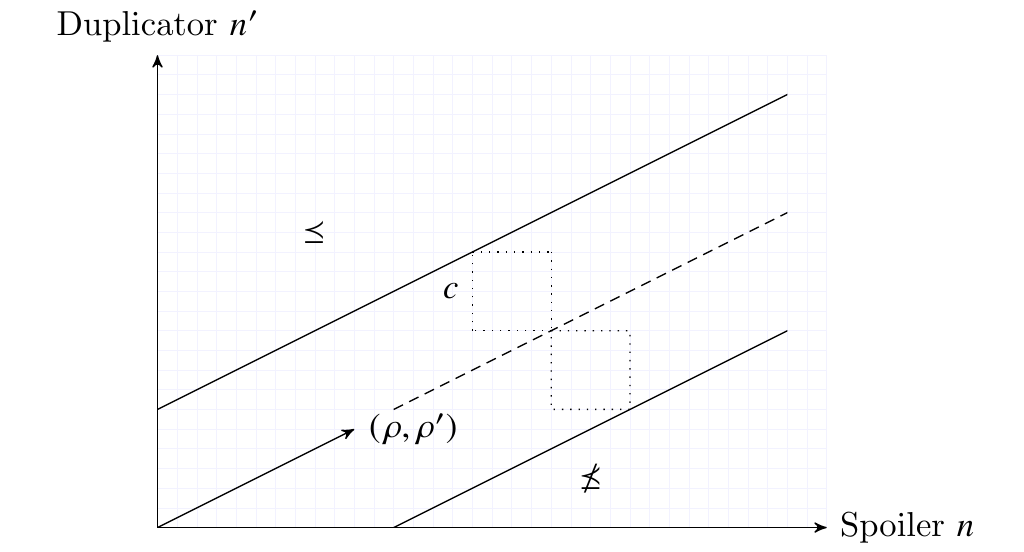}
  \end{center}
  \caption{A belt with slope $\rho/\rho'$. 
      The dashed half-line indicates the direction of the vector
      $(\rho,\rho')$.}
  \label{fig:belt}
\end{figure}

\begin{definition}
    A vector $(\rho,\rho')\in \mathbb{R}\x\mathbb{R}$ 
    is called \emph{positive} if
    $\rho\ge 0$, $\rho'\ge 0$ and $(\rho,\rho')\neq (0,0)$.
    Its \emph{direction} is the set
    $\mathbb{R}^+\cdot(\rho,\rho') = \{ (r \cdot \rho, r \cdot \rho') : r
        \in\mathbb{R}^+\}$
    of points that lie on the half-line defined by $(\rho,\rho')$ from the
    origin\footnote{$\mathbb{R}$ and $\mathbb{R}^+=\mathbb{R}\setminus\{0\}$ denote
    the sets of real numbers and non-negative real numbers, respectively.}.

    For a positive vector $(\rho,\rho')$ 
    and a number $c\in \mathbb{R}$ we say that the point $(n,n')\in\Z\x\Z$ is
    \emph{$c$-above} $(\rho,\rho')$ if
    there exists some point $(r,r') \in \mathbb{R}^+\cdot(\rho,\rho')$ in the
    direction of $(\rho,\rho')$
    such that
    \begin{equation}
      n < r - c \qquad \text{and} \qquad n' > r' + c.
    \end{equation}
    Symmetrically, $(n,n')$ is \emph{$c$-below} $(\rho,\rho')$ if there is a
    point $(r,r')\in \mathbb{R}^+\cdot(\rho,\rho')$ with
    \begin{equation}
      n > r + c \qquad \text{and} \qquad n' < r' - c.
    \end{equation}
\end{definition}
\noindent
When $c = 0$ we omit it and write simply `below' or `above'.

\begin{theorem}[Belt Theorem]\label{thm:belt-theorem-bounds}
    Let $\Net{N}$ and $\Net{N}'$ be two \OCNs\
    in normal form, with sets of states $Q$
    and $Q'$ respectively
    and let 
    $\Cacyc\le |Q\x Q'| \in \N$ be $1$ plus
    the maximal length of an acyclic path in the product graph
    of $\Net{N}$ and $\Net{N}'$.
    Then for every pair $(q,q')\in Q\x Q'$ of states
    there is a positive vector $(\rho, \rho')\in\N^2$ such that
    \begin{enumerate}
        \item if $(n,n')$ is $\Cacyc$-above $(\rho,\rho')$ then $qn\SIM{}{}q'n'$,
        \item if $(n,n')$ is $\Cacyc$-below $(\rho,\rho')$ then $qn\notSIM{}{}q'n'$,
        \item $\rho,\rho'\le \Cacyc$.
    \end{enumerate}
\end{theorem}


Notice that a point $(n,n')\in\N^2$ is $c$-below the positive vector $(0,1)$
iff $n>c$ and that no point in $\N^2$ is $c$-above this vector.
In the particular case of a pair of states $(p,p')$ with $pm\not\ssim p'm'$
for all $m,m'\in\N$, the vertical vector $(\rho,\rho')=(0,1)$
satisfies the claim of the Belt Theorem.

Our proof of Theorem~\ref{thm:belt-theorem-bounds} can be found in
Section~\ref{sec:belts/proof}. It is based on a finite abstraction of the
\sgame, that we will analyze in the next section.

\subsection{Slope Games}


By \cref{thm:oca.nf} (page~\pageref{thm:oca.nf}), we can assume without loss of
generality that the pair of \OCNs\
$\NN{N},\NN{N'}$ are in normal form (\cref{def:normal-form}).
Intuitively, this means that in a \sgame,
it is \R's objective to exhaust her opponent's counter.
Consequently, her local goal is to maximize the ratio $n/n'$
between the counter values along a play.

Consider the product graph of $\Net{N}$ and $\Net{N'}$   
and let $\qq=|Q\x Q'|$ be the number of states in this product.
If we ignore the actual counter values, any play of the \sgame\
starting in two processes of $\Net{N}$ and $\Net{N'}$ respectively,
describes a path in this product graph.
Moreover, after at most $\qq$ rounds, a pair of control states is revisited,
which means the corresponding path in the product is a lasso.

The effects of cycles in the product will play a central role in our further construction.
The intuition is that
if a play of a \sgame\ describes a lasso then both players
``agree'' on the chosen cycle. Repeating this cycle will change the ratio of the
counter values towards its effect.

To formalize this intuition, we define a finitary \slgame\ which
proceeds in phases.
In each phase, the players alternatingly move on the control graphs of their original
nets, ignoring the counter, and thereby determine the next lasso that occurs.
After such a phase, a winning condition is evaluated that compares the effect of the
chosen lasso's cycle with that of previous phases.
Now either one player immediately wins or the
effect of the last cycle was strictly smaller than all previous ones
and the next phase starts.
The number of different effects of simple cycles therefore bounds the maximal 
number of phases played. Since
each phase describes a lasso path in the product
this implies a bound on the total length of any play.

\begin{definition}
  \label{def:steeper}
  \label{def:behind}
    {
        \renewcommand{\R}{\mathbb{R}}
    Let $(\rho,\rho')$ and $(\alpha,\alpha')$
    be two vectors in $\R\x\R$
    }
    and consider the clockwise oriented angle from $(\rho,\rho')$ to $(\alpha,\alpha')$
    with respect to the origin $(0,0)$. We say that $(\alpha,\alpha')$ is \emph{behind}
    $(\rho,\rho')$ if this oriented angle is strictly between $0^\circ$ and $180^\circ$.
    See \cref{fig:behind} for an illustration.

    Positive vectors may be naturally ordered: We will call $(\rho,\rho')$
    \emph{steeper} than $(\alpha,\alpha')$,
    written $(\alpha,\alpha') \lesssteep (\rho,\rho')$, if $(\alpha,\alpha')$ is behind
    $(\rho,\rho')$.
\end{definition}
\begin{figure}[h]
  \begin{minipage}[t]{.45\textwidth}
      \centering
      \includegraphics[width=\textwidth]{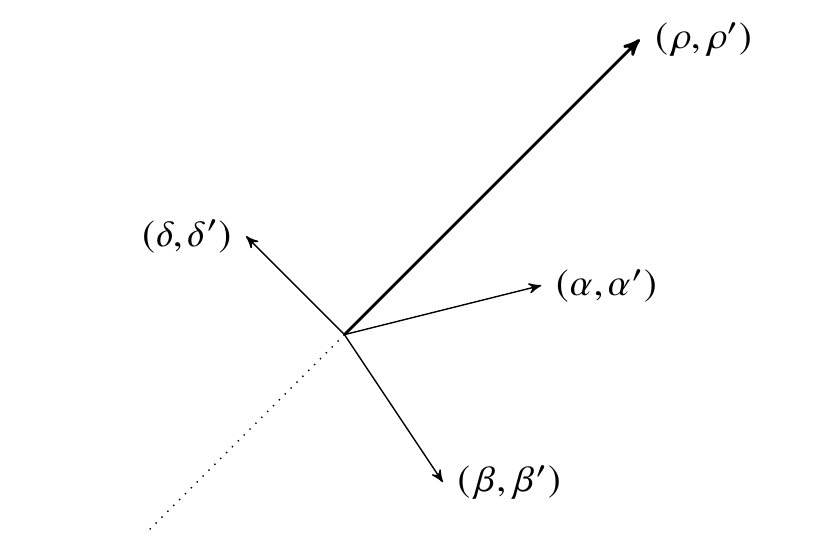}
      \caption{Vectors $(\alpha,\alpha')$ and $(\beta,\beta')$ are behind $(\rho,\rho')$, but
          $(\delta,\delta')$ is not.
      }
  \label{fig:behind}
    \end{minipage}
    \qquad
  \begin{minipage}[t]{.45\textwidth}
  \centering
    \includegraphics[width=\textwidth]{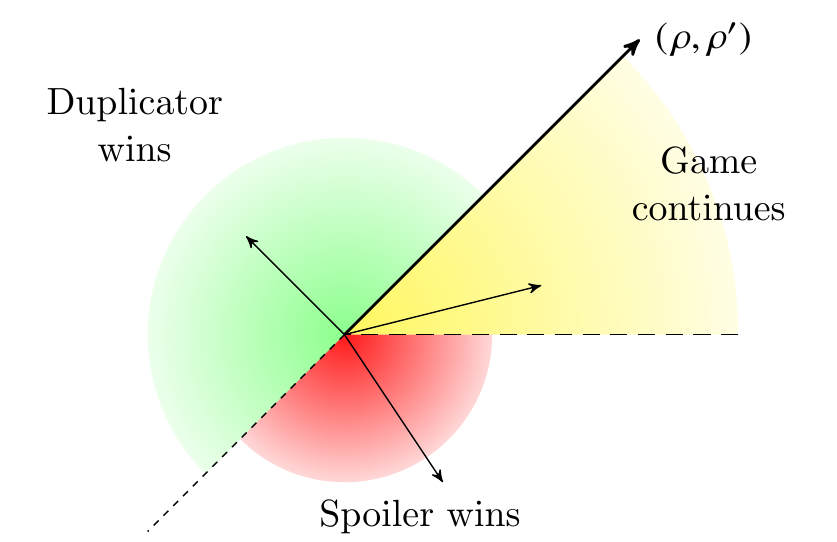}
  \caption{Evaluating the winning condition in position $(\pi,(\rho,\rho'))$
      after a phase of the \slgame.}
  \label{fig:slope-game}
  \end{minipage}
\end{figure}

\noindent Notice that the property of one vector being behind another only depends on their directions.
Also note that ``steeper'' (the relation $\lesssteep$) is only defined for positive vectors.
The following simple lemma will be useful in the sequel.

\begin{lemma}\label{lem:preserve-above}
Let $(\rho,\rho')$ be a positive vector and $c,m,n\in \N$.
\begin{enumerate}
    \item If $(n,n')$ is $c$-below $(\rho,\rho')$
        then $(n,n') + (\alpha,\alpha')$ is $c$-below $(\rho,\rho')$
        for any vector $(\alpha,\alpha')$ which is behind $(\rho,\rho')$.
    \item If $(n,n')$ is $c$-above $(\rho,\rho')$
        then $(n,n') + (\alpha,\alpha')$ is $c$-above $(\rho,\rho')$
        for any vector $(\alpha,\alpha')$ which is not behind $(\rho,\rho')$.
\end{enumerate}
\end{lemma}

\begin{definition}
    \label{def:slope_game}
  A \emph{\slgame} is a strictly alternating two player game played on a pair
  $\Net{N},\Net{N}'$ of \textOCNs\ in normal form. The game positions are pairs
  $(\pi,(\rho,\rho'))$, where $\pi$ is an acyclic path in the product graph
  of $\Net{N}$ and $\Net{N}'$, and $(\rho,\rho')$ is a positive vector called
  \emph{slope}.

The game is divided into \emph{phases}, each starting with a path
$\pi=(q_0,q'_0)$ of length $0$.
Until a phase ends, the game proceeds in rounds like a \sgame, but the players pick transitions instead of steps:
in a position $(\pi,(\rho,\rho'))$ where $\pi$ ends in states $(q,q')$, \R\ chooses a transition 
$t=(q \step{a, d} p)$,
then \V\ responds with a transition $t'=(q' \step{a, d'} p')$.
If the extended path $\bar\pi=\pi (t,t')(p,p')$ is still not a lasso,
the next round continues from the updated position $(\bar\pi,(\rho,\rho'))$; otherwise the phase ends
with \emph{outcome} $(\bar\pi,(\rho,\rho'))$.
The slope $(\rho,\rho')$ does not restrict the possible moves of either player, nor changes during a phase. 
We thus speak of \emph{the slope of a phase}.

If a round ends in position $(\pi, (\rho, \rho'))$ where $\pi$ is a lasso,
then the winning condition is evaluated.
We distinguish three non-intersecting cases
depending on how the effect 
\begin{align}  \label{eq:effects}
(\Delta(\cycl(\pi)), \Delta'(\cycl(\pi))) \ = \ (\alpha,\alpha')
\end{align} of the lasso's cycle relates to
$(\rho,\rho')$: 

\begin{enumerate}
  \item If $(\alpha,\alpha')$ is not behind $(\rho, \rho')$, \V\ wins immediately.
  \item If $(\alpha,\alpha')$ is behind $(\rho, \rho')$ but not positive, \R\ wins immediately.
  \item If $(\alpha,\alpha')$ is behind $(\rho, \rho')$ and positive, the game continues with a new
      phase from position $(\bar\pi,(\alpha,\alpha'))$, where $\bar\pi = \target{\pi}$
      is the path of length $0$
      consisting of the pair of ending states of $\pi$.
\end{enumerate}


\noindent\Vref{fig:slope-game} illustrates the winning condition.
Note that if there is no immediate winner it is guaranteed that $(\alpha, \alpha')$ is a positive
vector that is behind the slope $(\rho,\rho')$ of the last phase.
The number of different positive vectors that derive from
the effects of simple cycles thus bounds the maximal number of phases in the
game.
\end{definition}

The connection between the slope and \sgames\
is that the outcome of a \slgame\ from initial position $((q, q'), (\rho, \rho'))$ determines how
the initial slope $(\rho, \rho')$ relates to the belt in the plane for $(q,q')$
in the simulation relation.
Roughly speaking, if $(\rho,\rho')$ is less steep than the belt then \R\ wins the \slgame; if $(\rho,\rho')$
is steeper then \V\ wins.

Consider a \sgame\ in which the ratio $n/n'$ 
of the counter values of \R\ and \V\ is the same as the ratio
$\rho/\rho'$, i.e.,~suppose $(n,n')$ is contained in the direction of $(\rho,\rho')$.
Suppose also that the values $(n,n')$ are sufficiently large.
By monotonicity, we know that the steeper the slope $(\rho,\rho')$, the better for \V.
Hence if the effect $(\alpha,\alpha')$ of some cycle is behind $(\rho,\rho')$ and positive, 
then it is beneficial for \R\ to repeat this cycle. With more and more repetitions, 
the ratio of the counter values will get arbitrarily close to $(\alpha,\alpha')$.
On the other hand, if $(\alpha,\alpha')$ is behind $(\rho,\rho')$ but not positive then \R\ wins
by repeating the cycle until the \V's counter decreases to $0$.
Finally, if the effect of the cycle is not behind $(\rho,\rho')$ then repeating this cycle leads to \V's win. 

The next lemma follows from the observation that in \slgames,
the slope of a phase must be strictly less steep than those of all previous
phases.

\begin{lemma}\label{lem:slopegame-bounds}
    For a fixed pair $\Net{N},\Net{N}'$ of \OCNs\ in normal form,
    \begin{enumerate}
      \item any \slgame\ ends after at most $(\qq+1)^2$ phases, and
      \item \slgames\ are effectively solvable in \PSPACE.
    \end{enumerate}
\end{lemma}
\begin{proof}
  After every phase, the slope $(\rho, \rho')$ is equal to the effect of a simple
  cycle, which must be a positive vector.  Thus the absolute values of both numbers
  $\rho$ and $\rho'$ are bounded by $\qq=\card{Q\x Q'}$. It follows that the total
  number of different possible values for $(\rho,\rho')$, and therefore the maximal
  number of phases played, is at most $(\qq + 1)^2$.
  Point 2 is a direct consequence
  as one can find and verify winning strategies by an exhaustive search; polynomial space suffices 
  as the depth of the search is polynomial.
\end{proof}


The outcome of a \slgame\ depends only on the effects of simple
cycles 
that are behind the current slope $(\rho, \rho')$,
and not the actual values $\rho,\rho'$.
This motivates the following definition.

\begin{definition}
    \label{def:vector-equivalence}
Consider all the non-zero effects $(\alpha,\alpha')$ of all simple cycles 
and denote the set of all these vectors by $V$. We say that a positive vector $(\sigma,\sigma')$
\emph{subsumes} a positive vector $(\rho,\rho')$ when
for all $(\alpha,\alpha') \in V$,
\begin{equation}
    (\alpha,\alpha') \text{ is behind } (\rho,\rho') \quad  \implies \quad
     (\alpha,\alpha') \text{ is behind } (\sigma,\sigma').
\end{equation}
Call $(\rho,\rho')$ and $(\sigma,\sigma')$ \emph{equivalent} if they subsume each other.
\end{definition}

\begin{remark}\label{rem:vector-equivalence-epsilon}
Notice that for every positive vector $(\rho,\rho')\in V$
there exist $\epsilon>0$ such that
$(\rho-\epsilon,\rho'+\epsilon)$ subsumes $(\rho,\rho')$.
\end{remark}

In particular, all positive vectors lying in the open angle between any two angle-wise neighbors from $V^- = V \cup -V$
(where $-V = \{(-\alpha,-\alpha') \ : \ (\alpha, \alpha') \in V\}$) are equivalent.
We claim that equivalent slopes have the same winner in the \slgame.

\begin{lemma} \label{lem:constant-winner}
  If \R\ wins the \slgame\ from $((q,q'),(\rho,\rho'))$ and $(\sigma,\sigma')$ subsumes $(\rho,\rho')$ then 
  \R\ also wins the \slgame\ from $((q,q'),(\sigma,\sigma'))$.
  In consequence, when $(\rho,\rho')$ and $(\sigma,\sigma')$ are equivalent then the same player wins the \slgame\
  from $((q,q'),(\rho,\rho'))$ and $((q,q'),(\sigma,\sigma'))$.
\end{lemma}
\begin{proof}
  A winning strategy in the \slgame\ from $((q,q'),(\rho,\rho'))$ may be literally
  used in the \slgame\ from $((q,q'),(\sigma,\sigma'))$.
  This holds because the assumption that $(\sigma,\sigma')$ subsumes $(\rho,\rho')$
   implies that all possible outcomes of the initial phase of the
   \slgame\ are evaluated equally.
\end{proof}

\subsection{Proof of the Belt Theorem}\label{sec:belts/proof}
Consider one phase of a \slgame, starting from a position $(\pi, (\rho, \rho'))$.
The phase ends with a lasso whose cycle effect $(\alpha,\alpha')$ satisfies 
exactly one of three conditions, as examined by the evaluating function.
Accordingly, depending on its initial position, 
every phase falls into exactly one of three disjoint cases:

\begin{enumerate}
\item\label{final1}
\R\ has a strategy to win the \slgame\ immediately,
\item\label{final2}
\V\ has a strategy to win the \slgame\ immediately or
\item \label{nonfinal}
neither \R\ nor \V\ have a strategy to win immediately. 
\end{enumerate}

\noindent
In case (\ref{final1}) or (\ref{final2}) we call the phase \emph{final}, and in
case (\ref{nonfinal}) we call it \emph{non-final}.
The non-final phases are the most interesting ones
as there, both \R\ and \V\ have a 
strategy to either win immediately or continue the \slgame,
i.e., to avoid an immediate loss.

Both in final and non-final phases, a strategy for \R\ or \V\ is a tree as described below. 
For the definition of strategy trees we need to consider
not only \R's positions $(\pi, (\rho, \rho'))$ but also \V's positions, the
intermediate positions within a single round. These intermediate positions may be modeled
as triples $(\pi, (\rho,\rho'), t)$ where $t$ is a transition in $\Net{N}$ from the last state of $\pi$.
Observe that the bipartite directed graph, with positions of a phase as vertexes and edges determined by the single-move relation, 
is actually a tree, call it $T$. 
Thus a \R-strategy, i.e.~a subgraph of $T$ containing exactly one successor of every \R's position and all successors
of every \V's position, is a tree as well; and so is any strategy for \V.

Such a strategy (tree) in the \slgame\ naturally splits into \emph{segments}, each segment being a strategy (tree) in one phase.
The segments themselves are also arranged into a tree, which we call
a \emph{segment tree}.
Regardless of which player wins a \slgame, according to the above observations,
this player's winning strategy contains segments of two kinds:

\begin{itemize}
  \item non-leaf segments are strategies to either win immediately or continue the Slope
        Game (these are strategies for non-final phases); 
  \item leaf segments are strategies to win the \slgame\ immediately (these are strategies in final phases).
\end{itemize}

\noindent By the \emph{segment depth} of a strategy we mean the depth of its segment tree.
By point~1 of
\cref{lem:slopegame-bounds} (page~\pageref{lem:slopegame-bounds}),
we know that a \slgame\ ends after at most
$\dmax = (\qq + 1)^2$ phases.
Consequently, the segment depths of strategies are at most $\dmax$ as well.

Recall the value $\Cacyc$ defined as the maximal length of a simple cycle in the product graph,
i.e., the maximal length of any acyclic path plus $1$.
The claim of \cref{thm:belt-theorem-bounds} will easily follow from
the following two \cref{lem:ssim:spoiler-transfer,lem:ssim:duplicator-transfer}; they
state that if a player wins the \slgame,
an excess of counter value of $\Cacyc$ is sufficient to be able to safely ``replay'' a
winning strategy in the \sgame.

\begin{lemma}\label{lem:ssim:spoiler-transfer}
  If \R\ wins the \slgame\ from
  position $((p,p'), (\rho, \rho'))$ then \R\ wins the \sgame\ from every
  position $(p m, p' m')$ which is $\Cacyc$-below $(\rho,\rho')$.
\end{lemma}

\begin{proof}


\emph{(Informally)}
a position in the \slgame\ contains a positive vector $(\rho,\rho')$,
while a position in the \sgame\ contains a pair $(m,m') \in \N\x\N$ 
of counter values, that can also be interpreted as a positive vector.
%
%
%
The crucial idea of the proof is to consider the segments of the supposed winning strategy in the
\slgame\ separately.
Each such segment is a strategy for one phase and as such, describes how to move in
the \sgame\ until the next lasso is observed. Afterwards, \R\ can
choose to continue playing according to the next lower segment, or ``roll back'' the
cycle and continue playing according to the current segment.
By the rules of the \slgame\ we observe that after sufficiently many such rollbacks
the difference between the ratio $m/m'$ of the actual counters and the slope of the next
lower segment is negligible, i.e., these vectors are equivalent in the sense
of \fullref{def:vector-equivalence}. 
Then, \R\ can safely continue to play according to the next lower segment.

To safely play such a strategy in the \sgame, \R\ needs to ensure that her own
counter does not decrease too much as that could restrict her ability to move.
We observe however, that any partial play that ``stays in some segment''
can be decomposed into a single acyclic prefix plus a number of cycles.
Such a play therefore preserves the invariant that all visited points are
below the slope of the phase. In particular, this means that \R's
counter is always $\ge 0$.

\vspace{0.5cm}

\emph{(Formally)} the proof of \cref{lem:ssim:spoiler-transfer} proceeds by induction on the segment depth $d$ of the
assumed winning strategy in the \slgame.

\vspace{\baselineskip}
\case{$d = 1$} 
This means that \R\ has a strategy to win the \slgame\ in the first phase,
and hence to enforce that the effect of all cycles is behind $(\rho,\rho')$ but not
positive. Denote this strategy by $\sigma$.
In the \sgame\, \R\ will re-use this strategy as we describe below.
At every position $(qn,q'n')$ in the \sgame\, 
\R\ keeps a record of the \emph{corresponding position} $(\pi, (\rho,\rho'))$ in the \slgame\, enforcing the
invariant that $(q,q')$ are the ending states of the path $\pi$.

From the initial position  $(pm,p'm')$ with corresponding position $((p,p'), (\rho,\rho'))$,
\R\ starts playing the \sgame\ according to $\sigma$,
until the path in the corresponding position of the \slgame\, say $\pi_1$,
describes a lasso (this must happen after at most $\Cacyc$ rounds).
Thus $\pi_1$ splits into:
\begin{equation}
  \pi_1 = \alpha_1\beta_1
\end{equation}
where $\beta_1$ is a cycle.
Let $(a_1, a'_1) = (\seffect{\alpha_1},\deffect{\alpha_1})$ 
and $(b_1, b'_1) = (\seffect{\beta},\deffect{\beta_1})$ 
be the effects of $\alpha_1$ and $\beta_1$, respectively. 
The current values of counters are clearly
\begin{equation}
 m + a_1 + b_1 \qquad \text{and }\quad
 m' + a'_1 + b'_1
\end{equation}
assuming that the play did not end by now with \R's win.
As the length of path $\pi_1$ is at most $\Cacyc$ and
$(m,m')$ is assumed to be $\Cacyc$-below $(\rho,\rho')$, we know that
all positions visited by now in the \sgame\ were below $(\rho,\rho')$.
In particular, \R's counter value was surely non-negative by now.

Now \R\ ``rolls back'' the cycle $\beta_1$, namely changes the corresponding
position in the \slgame\ from $(\pi_1, (\rho,\rho'))$ to $(\alpha_1,(\rho,\rho'))$ 
and continues playing according to $\sigma$. 
The play continues until \R\ wins or the path in the corresponding position of the
\slgame\, say $\pi_2$, is a lasso again. Again, we split the path into an acyclic
prefix and a cycle:
\begin{equation}
  \pi_2 = \alpha_2 \beta_2.
\end{equation}
Denote the respective effects by $(a_2,a'_2)$
and $(b_2,b'_2)$.
A crucial but simple observation is that, assuming that the play did not end by now with \R's win,
the current values of counters are now
\begin{equation}
    m + a_2 + b_1 + b_2 \qquad \text{and }\quad
    m' + a'_2 + b'_1 + b'_2,
\end{equation}
i.e.~the effect $(a_1, a'_1)$ of the prefix
$\alpha_1$ of the previous lasso does not contribute any more.
As $(b_1, b'_1)$ is behind $(\rho,\rho')$ we may
apply \cref{lem:preserve-above}
(page~\pageref{lem:preserve-above}) to $(b_1, b'_1)$ with $c = 0$
in order to deduce, similarly as before, that all positions by now were below $(\rho,\rho')$.
Now \R\ rolls back $\beta_2$ by establishing $(\alpha_2, (\rho,\rho'))$ as
the new corresponding position in the \slgame.
Continuing in this way, after $k$ rollbacks the counter values are:
\begin{align}
\begin{aligned}
 n&=\;m\; + a_k + (b_1 + b_2 + \ldots + b_{k-1}) + b_k
 \qquad \text{and}\\
 n'&=\;m' + a'_k + (b'_1 + b'_2 + \ldots + b'_{k-1}) + b'_k,
 \end{aligned}
\end{align}
assuming that \R\ did not win earlier. 
All the effect-vectors $(b_i, b'_i)$ and thus also the sum
\begin{equation}\label{eq:vectorsum}
  (b_1 + b_2 + \ldots \ + b_{k-1}, b'_1 + b'_2 + \ldots + b'_{k-1})
\end{equation}
are behind $(\rho,\rho')$, hence similarly as before all positions by now have been below $(\rho,\rho')$,
by \cref{lem:preserve-above} applied to the vector~\eqref{eq:vectorsum} above.
This in particular means that \R's counter remains non-negative.
However, as by assumption all observed cycles come from a final segment in her \slgame\ strategy,
the vector~\eqref{eq:vectorsum} cannot be positive for any $k$. Thus, every rollback strictly
decreases \V's counter value.
We conclude that after sufficiently many rollbacks, \V's counter will reach $0$
and the game will end in a position immediately winning for \R.

\medskip
\case{$d > 1$} 
By assumption, \R\ has a strategy with segment depth $d$ to win the \slgame.
As before, we prescribe a strategy for her in the \sgame\ that
will re-use her \slgame\ strategy using rollbacks.

\R\ plays according to the initial segment of this strategy, that allows her to win or at least
guarantee that the effect 
of the first observed lasso's cycle is less steep than $(\rho,\rho')$.
After some rollbacks, the counter values will be of the form:
\begin{align} \label{eq:countervalues}
\begin{aligned}
  n&=\;m\; + a \; + (b_1  + \ldots + b_{l})  + (c_1  + \ldots + c_{k})\quad\text{and}\\
  n'&=\;m' + a' + (b'_1 + \ldots + b'_{l}) + (c'_1 + \ldots + c'_{k}),
 \end{aligned}
\end{align}
where the absolute values of $a$ and $a'$ are at most $\Cacyc$,
the vectors $(c_i,c'_i)$ are behind $(\rho,\rho')$ and positive,
and the vectors $(b_i,b'_i)$ are behind $(\rho,\rho')$ and non-positive.
We apply \cref{lem:preserve-above} and obtain
that all the positions so far have been below $(\rho,\rho')$.

In general \R\ has no power to choose whether the effect of the cycle at the
next rollback is positive or not.
However, if from some point on all effects are non-positive then \V's counter
eventually drops below $0$ and \R\ wins.
Thus w.l.o.g.\ we focus on positions in the \sgame\ immediately after a rollback of a
cycle with positive effect.
Using the notation from~\eqref{eq:countervalues}, suppose $(c_k,c'_k)$ 
is the effect of the last rolled back cycle.
In order to apply the induction assumption we need the following claim. The intuition is that
after sufficiently many rollbacks the vector $(n, n')$ will fall arbitrarily close to being $\Cacyc$-below some
vector $(c_i, c'_i)$.
Recall the relation of subsumption between positive vectors introduced in~\cref{def:vector-equivalence} 
(page~\pageref{def:vector-equivalence}).
\begin{claim}
  After sufficiently many rollbacks the vector $(n, n')$ of counter values
  in the \sgame\ is $\Cacyc$-below some vector $(\gamma,\gamma')$
  which subsumes
  the positive effect $(c_k,c'_k)$ of the last rolled back cycle.
\end{claim}
\begin{proof}
%
Simple geometric reasoning.
Let $(\rho_0,\rho_0')$ be the current slope of the phase in the slope game
and let
$
(\rho_1,\rho_1')\moresteep (\rho_2,\rho_2') \moresteep \dots\moresteep (\rho_k,\rho_k')
$
be the possible outcomes of the phase if \R\ plays according to the assumed
strategy.
Since the strategy is winning in the slope game, $(\rho_0,\rho_0')$ is steeper than all of them:
$(\rho_0,\rho_0')\moresteep  (\rho_1,\rho_1')$.

As mentioned in \cref{rem:vector-equivalence-epsilon},
for every $(\rho_{i},\rho_{i}')$ there exists a value $\epsilon_i>0$
such that $(\rho_{i}-\epsilon_i,\rho_{i}'+\epsilon_i)$ subsumes it.
Since $(\rho_{1}-\epsilon_1,\rho_{1}'+\epsilon_1) \moresteep (\rho_{1},\rho_{1}')$,
after a finite number of rollbacks the pair of counter values
in the \sgame\ must describe a positive vector that is
$\Cacyc$-below $(\gamma, \gamma') = (\rho_{1}-\epsilon_1,\rho_{1}'+\epsilon_1)$.
Since the effects of all possible outcomes of the phase are behind this vector,
\cref{lem:preserve-above} implies that from now on,
the counter-values after a rollback are $\Cacyc$-below $(\gamma, \gamma')$.
Now we consider two cases. If eventually a cycle with effect
$(\rho_1,\rho_1')$ is rolled back, the claim holds
since  $(\gamma, \gamma')$ subsumes it.
Otherwise, no cycle with effect $(\rho_1,\rho_1')$ is ever rolled back again.
In this case the whole above argument can be repeated for the next positive vector $(\rho_2,\rho_2')$, and so on.
An induction on the number $k$ of possible outcomes then shows the claim.
%
%
%
%
\end{proof}

Let $( q  n,  q'  n')$ be a position of the \sgame\ satisfying the claim.
We know that \R\ has a winning strategy in the \slgame\ from $(( q,  q'),
(c_k,c'_k))$, of segment depth at most $d-1$.
Because $(\gamma,\gamma')$ subsumes
$(c_k,c'_k)$, 
we apply \cref{lem:constant-winner} (page~\pageref{lem:constant-winner}) to know 
that the same strategy is winning in the
\slgame\ from $((q,q'), (\gamma, \gamma'))$.
By the induction assumption we conclude that \R\ wins the \sgame\
from $(qn,q'n')$, which completes the proof of
\cref{lem:ssim:spoiler-transfer}.\qedhere
\end{proof}

\begin{lemma}\label{lem:ssim:duplicator-transfer}
    If \V\ wins the \slgame\ from a
    position $((p,p'), (\rho, \rho'))$ 
    then \V\ wins the \sgame\ from every
    position $(p m, p' m')$ which is $\Cacyc$-above $(\rho,\rho')$.
\end{lemma}
\begin{proof}


Building again on the concept of rollbacks, we prescribe a winning strategy for \V\ in the \sgame\ that
is based on the assumed winning strategy $\sigma$ in the \slgame.
Intuitively, \V's strategy in the \sgame\ consists of two parts: first he plays according to $\sigma$ until a leaf segment
is reached, and then continues to play according to this segment using
rollbacks.
Since $\dmax= (\qq+1)^2$ bounds the maximal number of segments in $\sigma$
and every path in a segment is no longer than $\Cacyc$,
we know that an offset of $\dmax\cdot \Cacyc$ is sufficient to 
ensure that some position in a leaf segment can be reached.

We can accelerate this strategy, allowing
\emph{forward jumps}:
\V\ starts to play according to the initial segment of $\sigma$
at height $d$.
At any given position in a segment at height $h$,
\V\ first checks if the same pair
of control states appears in a segment at a lower height $h'<h$.
If such a position exists, \V\ continues to play from there,
otherwise he plays as prescribed by the current position.
See \cref{fig:dup_tiles} below for an illustration.
\begin{figure}[h]
  \begin{center}
      \includegraphics[scale=1.0]{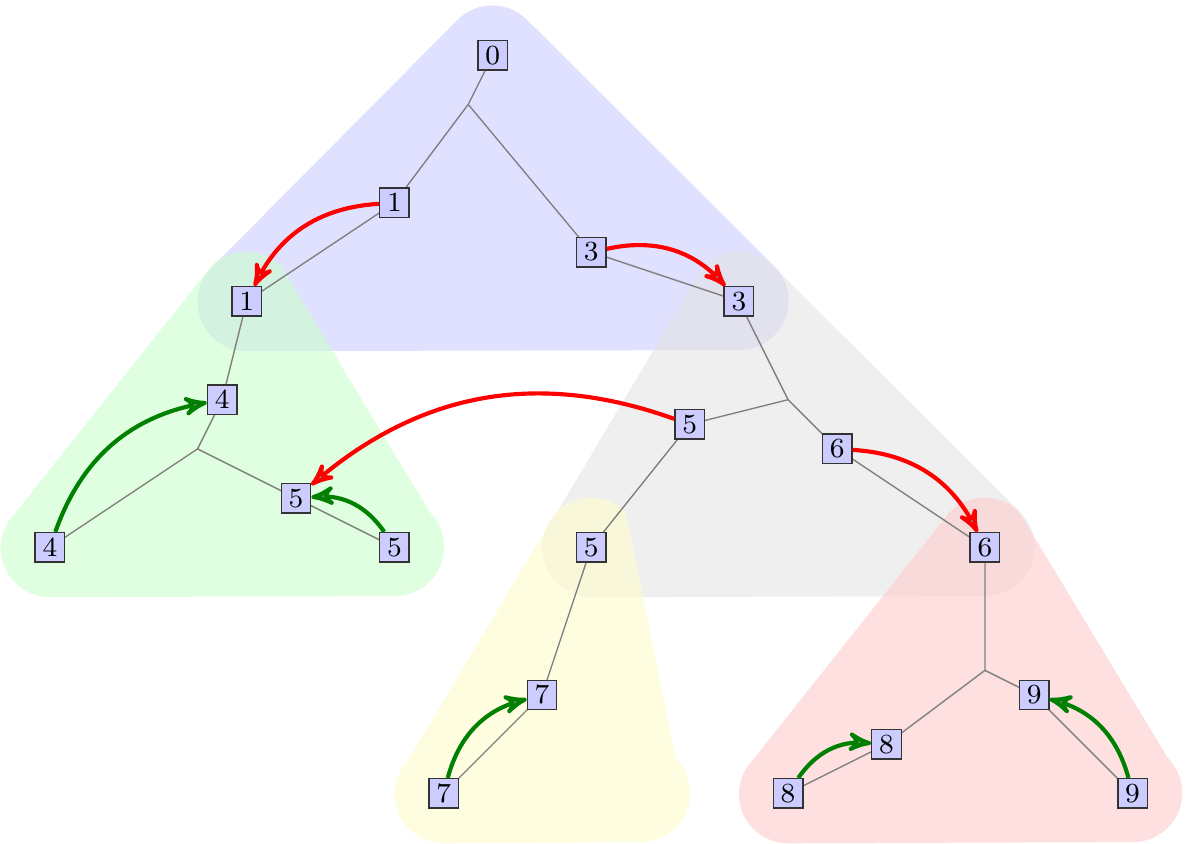}
  \end{center}
  \caption{A strategy for \V\ in the \slgame\ is turned into a strategy in the \sgame\
      by inserting forward jumps (red) and rollbacks (green).
      The green, yellow and red segments have height 1, the gray and blue segments
      have height 1 and 2 respectively.
      Nodes with the same labeling indicate positions with the same pair of control states.
  }
  \label{fig:dup_tiles}
\end{figure}

\noindent If \V\ plays as described above, he guarantees that no control states are
repeated unless he is already in a leaf segment $T$.
Moreover, as $T$ itself is a winning strategy
in the \slgame\ for some slope $(\varphi,\varphi')$,
he can enforce that the effects $(\seffect{\beta},\deffect{\beta})$
of the cycles $\beta$ of all observed lassos
are not behind $(\varphi,\varphi')$.
Let $\pi$ be an arbitrary play of the \sgame, in which \V\ plays as above using forward jumps, and then uses
rollbacks according to the segment $T$.
The effects of $\pi$ can be decomposed as
\begin{align}
    \label{eq:lem_transfer_d_effect}
\begin{aligned}
    \seffect{\pi} &= \seffect{\alpha}\; +\seffect{\beta_1\beta_2\dots \beta_k}\\
    \deffect{\pi} &= \deffect{\alpha} +\deffect{\beta_1\beta_2\dots \beta_k}
 \end{aligned}
\end{align}
where $\alpha$ is prefix and $\beta_i$ are simple cycles
with effect-vector not behind $(\varphi,\varphi')$.
Because \V\ uses forward jumps as soon as possible,
we know that no pair of states visited before entering $T$
can be contained in $T$, and thus $\alpha$ is acyclic.
The initial pair of counter values $(m,m')$ is $\Cacyc$-above $(\rho,\rho')$
and thus also $\Cacyc$-above $(\varphi,\varphi')$ because $(\varphi,\varphi')\lesssteep (\rho,\rho')$.
As $|\alpha|\le \Cacyc$, we know that $(m,m')+(\seffect{\alpha},\deffect{\alpha})$ is above $(\varphi,\varphi')$.
Moreover, as all the effects of all $\beta_i$ are not behind $(\varphi,\varphi')$,
their sum $\sum_{i\le k}(\seffect{\beta_i},\deffect{\beta_i})$ is also not behind $(\varphi,\varphi')$.
Using part~2 of \cref{lem:preserve-above}
we get that $(n,n') = (m,m') + (\seffect{\pi},\deffect{\pi})$ is still above $(\varphi,\varphi')$.
This in particular means that \V's counter value $n'$ remains non-negative.
Since $\pi$ was arbitrary, this shows that
\V\ can prevent his counter from ever decreasing below $0$
and thus enforce an infinite play and win.
\end{proof}
Assume a pair $\NN{N},\NN{N}'$ of \OCNs\ in normal form.
For two states $q\in Q$ and $q'\in Q'$
we will determine the ratio $(\rho,\rho')$ that, together with
$\Cacyc$, characterizes the belt of the plane $(q,q')$.
First observe the following monotonicity property of the \slgame.

\begin{lemma}
    \label{lem:slopegame-monotonicity}
  If \R\ wins the \slgame\ from a position $((q,q'), (\rho,\rho'))$ and
  $(\alpha,\alpha')\lesssteep(\rho,\rho')$
  then \R\ also wins the \slgame\
  from $((q,q'), (\alpha,\alpha'))$.
\end{lemma}
\begin{proof}
  Assume that \R\ wins from the position $((q,q'), (\rho,\rho'))$ while
  \V\ wins from $((q,q'), (\alpha,\alpha'))$, for some slope
  $(\alpha,\alpha')\lesssteep (\rho,\rho')$.
  This means that 
  a point $(n,n') \in \N\x\N$ exists which is both $\Cacyc$-above $(\alpha,\alpha')$
  and $\Cacyc$-below $(\rho,\rho')$.
  Applying both \cref{lem:ssim:spoiler-transfer,lem:ssim:duplicator-transfer}
  immediately yields a contradiction.
\end{proof}
Equivalently, if \V\ wins the \slgame\ from $((q,q'), (\rho,\rho'))$ and
$(\alpha,\alpha')$ is steeper than $(\rho,\rho')$ then he also wins
from $((q,q'), (\alpha,\alpha'))$. We conclude that for every pair $(q,q')$ of
states, there is a \emph{boundary slope}
$(\beta,\beta')\in\mathbb{R}\x\mathbb{R}$ such that

\begin{enumerate}
  \item \label{obs:steeper} \R\ wins the \slgame\ from $((q,q'), (\alpha,\alpha'))$
      for every $(\alpha,\alpha')\lesssteep(\beta,\beta')$;
  \item \label{obs:less-steep} \V\ wins the \slgame\ from $((q,q'), (\alpha,\alpha'))$
      for every $(\alpha,\alpha') \moresteep (\beta,\beta')$.
\end{enumerate}

Note that we claim nothing about the winner from the position $((q,q'), (\beta,\beta'))$ itself.
Applying \cref{lem:ssim:spoiler-transfer,lem:ssim:duplicator-transfer} we see that this boundary slope
$(\beta,\beta')$ satisfies the claims 1 and 2 of \cref{thm:belt-theorem-bounds}.
Indeed, consider a pair $(n,n')\in\N\x\N$ of counter values.
If $(n,n')$ is $\Cacyc$-below $(\beta,\beta')$, then there is certainly a
vector $(\alpha,\alpha')$ less steep than $(\beta,\beta')$ such that $(n,n')$ is $\Cacyc$-below $(\alpha,\alpha')$.
By point~\ref{obs:steeper} above, \R\ wins the \slgame\ from
$((q,q'), (\alpha,\alpha'))$. By \cref{lem:ssim:spoiler-transfer}, \R\ wins the
\sgame\ from $(qn,q'n')$.
Analogously, one can use point~\ref{obs:less-steep} above together with
\cref{lem:ssim:duplicator-transfer}
to show the second condition of \cref{thm:belt-theorem-bounds}.
This concludes the proof of the Belt Theorem.\qed

Recall the equivalence of positive vectors introduced in~\cref{def:vector-equivalence} 
(page~\pageref{def:vector-equivalence}), based on the set $V$ of ratios of simple cycles.  
Two vectors are equivalent if the same vectors from $V$ are behind both of them.
\Cref{lem:constant-winner} states that the outcome of a \slgame\ from
a fixed pair of states is the same for equivalent initial slopes.

By \cref{lem:constant-winner},
a boundary slope $(\beta,\beta')$ as used in the proof above must 
correspond to a slope contained in $V^- = V \cup -V$.
Indeed, otherwise $(\beta,\beta')$ must be between
two vectors from $V^-$   
and thus there are two equivalent vectors $(\gamma, \gamma')$ and $(\alpha,\alpha')$ satisfying
$
(\gamma,\gamma')
\lesssteep
(\beta,\beta')
\lesssteep
(\alpha,\alpha')
$.
By \cref{lem:constant-winner}, the outcome of a \slgame\ for
$(\gamma,\gamma')$ or $(\alpha,\alpha')$ is the same, contradicting
that $(\beta,\beta')$ is a boundary.

We conclude that the slope $(\beta,\beta')$ of any belt must be the effect
of a simple cycle of the product graph. Such paths are no longer than $\Cacyc$ and
because along a path of length $\Cacyc$ the counter values cannot change by more than
$\Cacyc$, we get that $\beta,\beta'\le \Cacyc$ as well.

\subsection{Locality}\label{sec:locality}
\TextSSIM\ preorder enjoys a certain
locality property due to the simulation condition.
Intuitively, the outcomes of all possible successor positions after one round of the \sgame\
determine the outcome of the game.
For \OCNs,
this can be stated as a precise geometric property.
Whether or not one process simulates another 
is completely determined by their control states and the coloring of its
surrounding pairs.

\newcommand{\NH}[2]{\mathit{NH}_{#1}^{(#2)}}
\begin{definition}
    \label{def:neightbourhood}
    Let $R\subseteq Q\x\N\x Q'\x\N$ be some relation on the configurations
    of two \OCN\ with sets of states $Q$ and $Q'$ respectively.
%
%
    The $R$-neighborhood of $(m,m')\in\N^2$
    is the function $\NH{R}{m,m'}:Q\x Q'\x\{-1,0,1\}\x\{-1,0,1\}\to\{0,1,\bot\}$ with

    \begin{equation}
        \NH{R}{m,m'}(q,q',l,l')  =
    \begin{cases}
        1, &\mbox{if }  (qm+l,q'm'+l')\in R\\
        0, &\mbox{if }  (qm+l,q'm'+l')\in (Q\x \N\x Q'\x \N)\setminus R\\
        \bot, &\mbox{if }  (qm+l,q'm'+l')\notin (Q\x \N\x Q'\x \N)\\
    \end{cases}
    \end{equation}
\end{definition}

\noindent The $R$-neighborhood of $(m,m')$ determines the coloring of $R$
on all points surrounding $(m,m')$.
Observe that there are at most $3^{|Q\x Q'|\cdot 3\cdot3}$ different
neighborhoods.
The $\bot$-values ensure that if two points $(m,m')$ and $(n,n')$ in $\N^2$
have the same neighborhood, then they have the same relative position
to the axes, i.e., $m=0\iff n=0$ and $m'=0\iff n'=0$.

We can now precisely state what we mean with the \emph{locality}
of \textSSIM\ on \OCA.
\begin{lemma}[Locality]
    \label{lem:ssim:locality}
    Consider a pair $(p,p')\in (Q\x Q')$ of states and naturals
    $m,m',n,n'\in \N$.
    If the $\SIMSYMBOL$-neighborhoods of $(m,m')$ and $(n,n')$
    agree on every $(q,q',l,l')\neq(p,p',0,0)$,
    then they also agree on $(p,p',0,0)$,
    i.e., $pm\ssim p'm' \iff pn\ssim p'n'$.
\end{lemma}
\begin{proof}
Suppose that \V\ wins the \sgame\ from $(pm, p'm')$.
For every move $pm\step{a}qm+l$ in the game from $(pm,p'm')$, \V\ has a response
$p'm'\step{a}q'm'+l'$ such that $qm+l\ssim q'm'+l'$.
Due to the assumption that $\SIMSYMBOL$-neighborhoods of $(m,m')$ and $(n,n')$
agree on every $(q,q',l,l')\neq(p,p',0,0)$, we learn that
for every move $pn\step{a}qn+l$ in the game from $(pn,p'n')$, \V\ has a response
$p'n'\step{a}q'n'+l'$ such that either $qn+l\ssim q'n'+l'$, or $(q,q',l,l')=(p,p',0,0)$.
This proves that \V\ wins the \sgame\ from $(pn, p'n')$, as required.
%
%
%
\end{proof}

Since the simulation condition for a pair of processes depends only on their
neighborhood, we can
locally verify that some finite coloring is not self-contradictory.
Moreover, if a relation on the configurations of two \OCN\ 
is not a simulation, then this is witnessed
locally by some inconsistent neighborhood.

\begin{lemma}
    \label{lem:ssim:locality:simcondition}
    A relation $R\subseteq (Q\x\N\x Q'\x\N)$
    is a simulation if
    for every $(pm,p'm')\in R$ there exists $(n,n')\in\N^2$
    with $\NH{R}{m,m'}\;=\; \NH{\SIMSYMBOL}{n,n'}$.
\end{lemma}
\MORE{Not an iff here: consider the systems without any steps.
    Then any relation is a simulation, also the $R=\;\ssim\setminus \{(p,5,p',5)\}$
    that artificially removes a single point. Still,
    $\NH{R}{5,5}\neq\NH{\SIMSYMBOL}{n,n'}$ for every $n,n'\in\N$.}
\begin{proof}
  The condition implies that $R$ satisfies the
  simulation condition:
  Pick any $(pm,p'm')\in R$ and let $n,n'\in\N$ such that
  $\NH{R}{m,m'}\;=\; \NH{\SIMSYMBOL}{n,n'}$
  and consider a \R-move $pm \step{a} q(m+d)$.
  We have $\NH{\ssim}{n,n'}(p,p',0,0) = \NH{R}{m,m'}(p,p',0,0) = 1$
  and therefore that $pn\ssim p'n'$.
  So there is a valid \V's response $pn'\step{a} q'(n'+d') \SIMBYSYMBOL q(n+d)$.
  But then also $q(m+d)~R~q'(m'+d')$ as
  $\NH{R}{m,m'}(q,q',d,d') = \NH{\ssim}{n,n'}(q,q',d,d') = 1$.
\end{proof}

\subsection{Characterizing Strong Simulation Preorder}\label{sec:max strong sim}
We follow here the approach presented in \cite{JKM2000} to turn the
Belt Theorem into a working algorithm.
The idea is to guess and verify a description of $\ssim$ in terms of belts and local
colorings on-the-fly. Due to the polynomial bounds on the width of belts
stated in \cref{thm:belt-theorem-bounds}, such a procedure requires polynomial
space.

Consider two \OCN\ $\Net{N}$ and $\Net{N}'$
in normal form, with sets of
control states $Q$ and $Q'$, respectively and let $\Cacyc\le |Q\x Q'|$
be the maximal length of an acyclic path in their product plus $1$,
as used in \cref{thm:belt-theorem-bounds}.

For
convenience, we will write
$(pm,p'm') + k\cdot (n,n')$ to mean $(p(m+k\cdot n),p'(m'+k\cdot n'))$
for any $(p,p')\in(Q\x Q')$ and $m,m',n,n',k\in\N$.
Similarly,
for a relation $R\subseteq(Q\x\N\x Q'\x\N)$ we write
$R+ k\cdot(n,n') = \{(pm,p'm')+k\cdot (n,n') \mid (pm,p'm')\in R\}$.

\begin{definition}
    The \emph{slope} of a pair $(p,p')\in Q \x Q'$ of control states,
    is the positive vector $\slope(p,p')=(\rho,\rho')$
    satisfying the claim of the Belt Theorem.
    The \emph{belt} with slope $(\rho,\rho')$
    is the set of points $(n,n')\in\N^2$
    which are neither $\Cacyc$-above nor $\Cacyc$-below $(\rho,\rho')$.
   %
    The \emph{extended belt} is the relation 
    $\belt(p,p')\subseteq (Q\x\N\x Q'\x\N)$
    that contains $(qn,q'n')$ iff
    $(n,n')$ is in the belt with slope $\slope(p,p')$.
\end{definition}

Recall that simulation preorder on the configurations with
control states $p$ and $p'$ is trivially outside of $\belt(p,p')$:
it contains all pairs $(p m, p' m')$ such that\ $(m,m')$ is $\Cacyc$-above $\slope(p,p')$,
and contains no pairs $(p m, p' m')$ where $(m,m')$ is $\Cacyc$-below $\slope(p,p')$.
We show (\cref{lem:ssim:uperiodic}) 
that the non-trivial part
    \begin{align*}
        \SIM{}{p,p'}\;=\;\ssim\:\cap\;\belt(p,p')
    \end{align*}
is repetitive in the sense defined in \cref{def:ssim:uperiodic} below.
Essentially, one can cut through the belt at two levels $n_1,n_2\in\N$
such that the coloring of $\belt(p,p')$ above level $n_2$
repeats the (finite) coloring between $n_1$ and $n_2$ indefinitely.
This implies that $\SIM{}{p,p'}$ and hence also $\ssim$ are semilinear,
and each $\SIM{}{p,p'}$ can be represented by the finite
coloring up to level $n_2$.
This is already enough to decide strong \textSSIM,
and to compute a representation of the maximal \textSSIM,
since one can
enumerate candidate relations $R\subseteq(Q\x\N\x Q'\x\N)$
that are represented in this way and check that they
satisfy the simulation condition.

Due to the polynomial bounds on the width and the slopes of belts
provided by \cref{thm:belt-theorem-bounds},
we can further bound the cut-levels $n_1,n_2$ and thus the representation
of periodic candidate relations, exponentially in the size of the input nets.
The crucial idea for deciding strong simulation in polynomial space
is that
one can stepwise guess and locally verify the coloring
of a (extended) belt by shifting a polynomially bounded window
along the belt.

By \cref{thm:belt-theorem-bounds}, we know that
coefficients $\rho$ and $\rho'$ of any slope $\slope(p,p')=(\rho,\rho')$
are bounded by $\Cacyc$.
Consequently, there are at most $\Cacyc^2$ different slopes
and belts and
apart from vertical and horizontal slopes
(those with $\rho=0$ or $\rho'=0$ respectively),
the maximally and minimally steep (cf.~\fullref{def:steeper}) possible slopes
are $(1,\Cacyc)$ and $(\Cacyc,1)$ respectively.
We can therefore find polynomially bounded $l_0,l_0'\in\N$
such that belts are pairwise disjoint outside
the \emph{initial rectangle} $\is$ between corners $(0,0)$ and $(l_0,l_0')$.
For technical convenience we assume w.l.o.g.~that only horizontal belts
(those with $\slope(p,p') = (n,0)$ for some $n$)
cross the vertical border of $\is$.
This can always be achieved by extending $\is$, if necessary.

By our definition of belts,
shifting a point along the vector $\slope(p,p')$ 
preserves membership in $\belt(p,p')$,
i.e., for every $(qn,q'n')\in (Q\x\N\x Q'\x\N)$,
\begin{equation}
    (qn,q'n') \in \belt(p,p') \iff (qn,q'n')+k\cdot\slope(p,p') \in \belt(p,p').
\end{equation}
This is why we restrict our focus to multiples of vectors $\slope(p,p')$.

\begin{definition}
    \label{def:ssim:uperiodic}
Fix a pair $(p,p') \in Q\times Q'$ and $j,k \in \N$
and let $l_0,l'_0\in\N$ define the initial rectangle $\is$ discussed above.
We write $\squa{p}{p'}{j}$ for the rectangle between corners $(0,0)$ and
$(l_0,l_0') +  j\cdot \slope(p,p')$.
A subset $R \subseteq \belt(p,p')$  is called \emph{$(j,k)$-ultimately-periodic} if
for all $(n,n') \in \N^2 \setminus \squa{p}{p'}{j}$ and every $(q,q')\in (Q\x Q')$,
\begin{align}
\begin{aligned}
    (qn,q'n') \in R \iff (qn,q'n')+k\cdot\slope(p,p') \in R.
\end{aligned}
\end{align}
\end{definition}

\noindent One can represent
a $(j,k)$-ultimately-periodic set $R$
by the two numbers
$n_1'=l_0'+j\cdot\rho'$
and
$n_2'=n_1'+k\cdot\rho'$
and two finite sets
\begin{equation}
    \label{eq:ssim:up-representation}
    \{(qn,q'n')\in R\ |\ n' < n_1'\}
    \qquad\text{and}\qquad
    \{(qn,q'n')\in R\ |\ n_1'\le n'<n_2'\}.
\end{equation}
This in particular means that $R$ is semilinear,
where the left subset above forms the bases,
and the only period is always $\slope(p, p')$.  
We continue to show that the non-trivial part $\SIM{}{p,p'}$
of the coloring of \textSSIM\ is such a $(j,k)$-ultimately periodic set
for every pair $(p,p')$ of states.

\begin{figure}[h]
    \centering
  \includegraphics{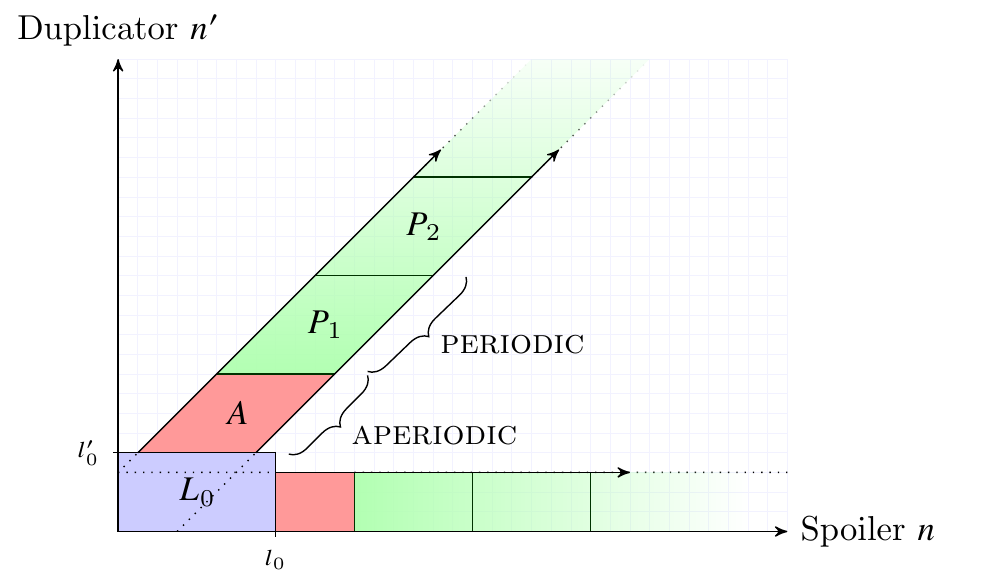}
\caption{The initial rectangle $\is$ (blue) and two belts.
    Outside $\is$, the coloring of a belt
    consists of some exponentially bounded block (red), and another exponentially bounded non-trivial block (green)
    which repeats ad infinitum along the rest of the belt.
}\label{fig:ssim:belt-decomp}
\end{figure}

\begin{lemma} \label{lem:ssim:uperiodic}
For every pair $(p, p') \in Q\times Q'$, the set
$\SIM{}{p,p'}$
is $(j,k)$-ultimately periodic for some $j,k \in \N$ exponentially bounded in
$\Cacyc$.
\end{lemma}
\begin{proof}


Fix states $p, p'$ and let $(\rho, \rho') = \slope(p,p')$.
W.l.o.g.~suppose that $\slope(p,p')$ is positive and $\belt(p,p')$ therefore intersects
the horizontal border of $\is$ 
(if the belt is horizontal and intersects the vertical border of $\is$ the proof is analogous).

By a \emph{cross-section} at level $n'$ we mean the
set of all points in $\belt(p,p')$ on a horizontal line at that level,
i.e., $\{(qn,qn')\in\belt(p,p')\ |\ n\in\N\}$.
We say that two cross-sections $s_1$ and $s_2$ are \emph{equal} if
one of them is obtained by a shift of the other by a multiple of
$\slope(p,p')=(\rho, \rho')$
and moreover, the $\SIMSYMBOL$-neighborhoods of any two corresponding points
are the same.
Formally, we require that for some $k \in \N$,
\begin{enumerate}
  \item $s_2 = s_1 + k \cdot (\rho,\rho')$ \label{eq:equalcs}
  \item $\NH{\SIMSYMBOL}{m,m'}=\NH{\SIMSYMBOL}{n,n'}$
        for any $(qn,q'n')\in s_1$ and
        $(m,m')=(n,n')+k\cdot\slope(p,p')$.
\end{enumerate}
Notice that there are at most $2^{\qq\cdot W}$ pairwise different colorings for
any cross-section, where $\qq=|Q\x Q'|$ and $W$ is the maximal width of a belt.
By our definition of
neighborhoods, two cross-sections are equal
only if their coloring agrees
and the same is true for the
(pairs of)
cross-sections
directly above and below.
This means that in total, there are no more than
$2^{\qq\cdot W\cdot3}$
pairwise different cross-sections for a given belt.

We choose two equal cross-sections 
at levels $n'_1$ and $n'_2$ respectively,
such that
$n_1'=l_0'+j\cdot\rho'$
and 
$n_2'=n_1'+k\cdot\rho'$ for some $j,k\in\N$.
That is, we demand
that $l_0'<n_1'<n_2'$ and the respective offsets
are divisible by the vertical offset $\rho'$ of $\slope(p,p')$.
By our observation above it is safe to assume that
both $j$ and $k$ are bounded exponentially in $\Cacyc$.

Based on $n_1'$ and $n_2'$, we decompose $\SIM{}{p,p'}$
into finite segments.
To this end, first extend $n'_1$ and $n'_2$ to an infinite progression
$n'_1, n'_2, n'_3, \ldots$
where $n'_{i+1} = n'_i + k \cdot \rho'$ for $i \geq 1$.
Now let $A$ be the restriction of $\ssim_{p,p'}$ to the area below $n_1'$
and for any $i\ge 1$,
let $P_i$ be the restriction of $\ssim_{p,p'}$ to the area between $n_i'$ and
$n_{i+1}$ (see \cref{fig:ssim:belt-decomp}):
\begin{align*}
    A \ & = \  \{ (qn,qn') \in\;\SIM{}{p,p'} \ : \ n' < n'_1 \} &
    P_i \ & = \ \{ (qn,qn') \in\;\SIM{}{p,p'} \ : \ n'_i \leq n' < n'_{i+1} \}.
\end{align*}
We now show that
\begin{equation}
    \label{eq:ssim:up-A-P}
    \SIM{}{p,p'} \ = \ A \ \cup \ P_1^*, \quad \text{ where } P_1^* \ = \
    \bigcup_{i \in \N} (P_1 + i \cdot k \cdot (\rho,\rho')).
\end{equation}
That is, apart from the initial fragment $A$,
the coloring of $\SIM{}{p,p'}$
is actually an infinite repetition of a finite
coloring $P_1$ along the belt: $P_{i+1} = P_i + k\cdot\slope(p,p')$.
This implies the claim of the lemma, since
$A\cup P_1^*$ is clearly $(j,k)$-ultimately periodic.
The proof of \cref{eq:ssim:up-A-P} strongly relies on the locality of the simulation condition
(\fullref{lem:ssim:locality:simcondition}).

For the first inclusion ($A \ \cup \ P_1^* \subseteq \ \ssim_{p,p'}$)
we show  that the relation
\begin{equation}
    R \;=\;(\ssim\;\setminus \;\SIM{}{p,p'}) \ \cup (\ A \cup \ P_1^*)
\end{equation}
obtained from $\ssim$ by replacing $\ssim_{p,p'}$ with $A \ \cup \ P_1^*$,
is a simulation.
Recall that $n_1'$ and $n_2'$ were chosen sufficiently high (above the initial
rectangle $\is$) such that any two different belts are disjoint.
This means that the $R$-neighborhood of any point in $P_i$ for some $i>1$
is the same as the $R$-neighborhood and hence also the
$\SIMSYMBOL$-neighborhood of the corresponding point in $P_1$.
By \cref{lem:ssim:locality:simcondition}, this means that $R$ is a simulation
and since $\ssim$ is the largest simulation, the claimed inclusion follows.

It remains to show the other inclusion ($A \ \cup \ P_1^* \supseteq\ \SIM{}{p,p'}$).
Assume the contrary. We already know that $A \ \cup \ P_1^* \subseteq \ \SIM{}{p,p'}$,
so we must have $P_1+ i\cdot k\cdot\slope(p,p') \subsetneq P_{i}$ for some $i>1$.
Since $P_i\subseteq\;\SIM{}{p,p'}$ is part of the coloring of \textSSIM\
$\ssim$, it is clearly locally consistent.
This means if we replace $P_1$ with the coloring according to $P_i$,
we again derive a consistent coloring.
Formally, we let $P = P_i + i\cdot (-k)\cdot\slope(p,p')$
and replace $\SIM{}{p,p'}$ with $A\cup P^*$ in the coloring of $\ssim$.
Similar to the first case, the resulting relation
\begin{equation}
    (\ssim\;\setminus \;\SIM{}{p,p'}) \ \cup \ ( A \cup P^*)
\end{equation}
is a simulation due to the locality of the simulation condition.
This implies that $P_1\subsetneq P\subseteq\;\SIM{}{p,p'}$, which means
that there exists some point $(qn,q'n')\in\;\SIM{}{p,p'}\setminus\ P_1$ with 
$n_1'\le n' < n_2'$.
This contradicts the definition of $P_1$
as the set of points $(qn,q'n')$ in $\SIM{}{p,p'}$ with $n'_1\le n'<n'_2$.
\end{proof}

\Cref{lem:ssim:uperiodic} implies that the largest strong simulation $\ssim$
is not only semilinear, but also its nontrivial part $\bigcup_{p, p'} \SIM{}{p,p'}$ is the finite union of
$(j,k)$-ultimately periodic sets, for exponentially bounded $j,k$.
It therefore admits an \EXPSPACE\ representation that consists, 
for every pair of states $(p,p')$, of:

\begin{itemize}
\item a polynomially bounded vector $(\rho,\rho') = \slope(p,p')$
\item exponentially bounded natural numbers $n_1',n_2'\in\N$
\item two exponentially bounded relations: 
\begin{align*}
  \aper &\;=\;\{(qn,q'n')\in\;\SIM{}{p,p'} |\ n'\le n_1'\}\\
  \per  &\;=\;\{(qn,q'n')\in\;\SIM{}{p,p'} |\ n_1'\le n'<n_2'\}
\end{align*}
\end{itemize}

\noindent Assume w.l.o.g.~that in descriptions of the above form, the coefficients
$n_1'$ and $n_2'$
are the same for all pairs $(p,p')$ with the same
$\slope(p,p')$. This is a safe assumption as the least common multiples of the
respective values are still exponentially bounded.

The above characterization immediately leads to a na\"ive exponential-space
algorithm for checking strong simulation for pairs of \OCN s
in normal form:
Guess the description of a candidate relation $R$ for the simulation relation, verify
that it is a simulation and check if it contains the input pair of configurations.

Checking whether the input pair is in the (semilinear) relation $R$ is trivial.
To verify that the relation $R$ is a simulation, one needs to check
the simulation condition for every pair of configurations $(qn,q'n')$ in
$R$. But due to the particular periodic structure of the candidate relation
and the locality of \textSSIM\ (\fullref{lem:ssim:locality}),
it suffices to locally verify the finite initial and periodic parts
for every pair of control states.
%

\medskip
\noindent
\textbf{A \PSPACE\ procedure.}
The na\"ive algorithm outlined above may easily be turned into a \PSPACE\ algorithm by
a window shifting trick.
Instead of guessing the complete exponential-size description upfront, we start by
guessing the polynomially bounded relation inside $\is$ and verifying it locally.
Next, the procedure stepwise guesses parts of the relations $\aper$ and later $\per$,
inside a polynomially bounded rectangle window through the belt
and shifts this window along the belt, checking the simulation condition for all
contained points along the way. Since the simulation condition is local, everything
outside this window may be forgotten, save for the
first repetitive window that is used as a certificate for successfully having guessed
a consistent periodic set, once it repeats.
By \cref{lem:ssim:uperiodic},
this repetition needs to occur after an exponentially bounded number of shifts.
Therefore, polynomial space is sufficient to store a binary counter
that counts the number of shifts and allows to terminate unsuccessfully once the
limit is reached.

We summarize our findings as the theorem below.

\begin{theorem}\label{thm:ssim-pspace}
  Checking strong simulation preorder between two \OCNs\
  is in \PSPACE.
  Moreover, the maximal simulation relation is semilinear
  and can be represented in space exponential in the number of states
  of the input nets.
\end{theorem}

\section{Weak Simulation}\label{sec:wsim}
%
We now turn to the problem of checking if \textWSIM\ holds between two
\OCN-processes.
This problem was shown to be decidable in \cite{HMT2013} and later \PSPACE-complete \cite{HLMT2013}.
We provide here a unified presentation of the
argument for its decidability and the
subsequent improvement to \PSPACE.

The main obstacle is that, with respect to weak steps, \V's system is
infinitely-branching.
This implies that non-simulation does not necessarily manifest itself locally,
i.e., the weak simulation condition is not local
in the sense discussed in \cref{sec:locality}.
Our approach is based on a generalization of simulation \emph{approximants},
which we will recall below.

\begin{definition}[Approximants]\label{def:normal_approximants}
    Take two labeled transition systems with sets of configurations $S$ and
    $S'$, respectively.
    Strong \emph{simulation approximants} $\SIM{}{\alpha}$ with respect to
    $S,S'$ are inductively defined for all ordinals $\alpha\in\Ord$:
    \begin{enumerate}
       \item $\SIM{}{0}\;=\;S\x S'$ is the full relation.
       \item $s \SIM{}{\alpha+1} s'$
           holds if for all $s\step{a}t$ there is a step $s'\step{a}t'$ such
           that $t\SIM{}{\alpha}t'$.
       \item For limits $\lambda$, let $\SIM{}{\lambda} \;=\; \bigcap_{\alpha<\lambda}\SIM{}{\alpha}$.
    \end{enumerate}
    Weak simulation approximants $\WSIM{}{\alpha}$ are defined as
    above, where we replace (2) by the weak simulation condition:
       $s \WSIM{}{\alpha+1} s'$
           iff for all $s\step{a}t$ there is a step $s'\wstep{a}t'$ with $t\WSIM{}{\alpha}t'$.
\end{definition}

One can show (see e.g.~\cite[Chapter 10.4]{Mil1989} for an argument for
bisimulation approximants)
that regardless of the given LTSs $S$ and $S'$ it holds that
\begin{equation}
\SIM{}{}\;=\bigcap_{\alpha\in\Ord}\SIM{}{\alpha}
\end{equation}
In particular this means that for fixed $S,S'$ there exists
some \emph{convergence} ordinal $\gamma$ with $\SIM{}{}\;=\;\SIM{}{\gamma}$.
Moreover, if $S'$ is a finitely branching LTS (each configuration has finitely many
successors), then convergence happens at most at the first limit ordinal.
In this case, if $s\not\SIM{}{} s'$ then already $s\not\SIM{}{k}s'$ at some 
finite level $k\in\N$.
It is this finite convergence property that fails in the case
of weak simulation for LTS defined by \textOCNs, as the example below
demonstrates.

\todo{PT: my changes of the weaksim intro end here}

\begin{example}
    \label{ex:weaksim:nonconvergence}
    Consider the simple process
{$\begin{tikzpicture}[baseline]%
\node[anchor=base] (A) {A};%
\draw[->] (A) edge[loop right] node {$a$} (A);
\end{tikzpicture}$},
    that can only
    loop on action $a$,
    and
    the \OCN\ depicted below.
    \begin{center}
      \includegraphics{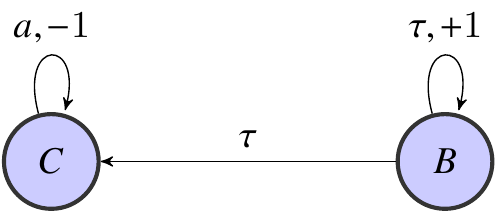}
    \end{center}
    We see that
    $A\WSIM{}{n}Cn$ and $A\notWSIM{}{n+1}Cn$
    hold for every $n\in\N$.
    Moreover, there is a weak step $B0\wstep{a}{}{}Cn$ for every $n\in\N$
    and therefore $A\WSIM{}{\omega}B0$.
    Still, it holds that $A\notWSIM{}{\omega+1}B0$
    because there is no
    weak $a$-step from $B0$ to a
    process $\alpha$ that satisfies $A\WSIM{}{\omega}\alpha$.
    It follows that
    $\WSIM{}{\omega}\;\neq\;\WSIM{}{\omega+1}$.
    We will later show (as \fullref{thm:wsim:approximants:omegasquare})
    that convergence of \textWSIM\ approximants on \OCN\
    can only be guaranteed at level $\omega^2$.
\end{example}
We resolve this problem in two steps.
First, the weak simulation problem between \OCNs\ is reduced to a \emph{strong}
simulation problem between a \OCN\ and a slightly generalized model that we call
$\omega$-nets, and that symbolically captures infinite branching.
In $\omega$-nets, there exist dedicated transitions with symbolic effect $\omega$, which allow to
arbitrarily increase the counter in a single step.
Secondly, this new strong simulation problem is solved using a
novel kind of approximant sequence, that is derived from the representation of
\V's system.
It is shown that this sequence converges at a finite index
and that individual approximant relations are effectively computable
semilinear sets.
In particular, knowing the representation of the approximant at level $k$,
one can characterize the approximant at the next level $k+1$
in terms of strong simulation over suitably modified \OCNs,
which is an effectively computable semilinear set by
\cref{thm:ssim-pspace}. 
A description of the \textWSIM\ preorder $\wsim$ can therefore
be computed by successively computing
the approximant relations and stopping once convergence is detected.
This procedure is effective
because the semilinear approximants are guaranteed to converge to
$\wsim$
at some finite level and equality is decidable for semilinear sets.

The remainder of the section is organized as follows.
In \cref{sec:w-nets}, $\omega$-nets are introduced and the reduction theorem
(\cref{thm:wsim:reduction}) is proved.
\Cref{sec:wsim-flavour} presents the key idea behind
the inductive procedure that follows.
In \cref{sec:wsim-msv} we recover a technical detail about the computability
of certain belts in strong simulation games between \OCNs.
In \cref{sec:wsim-approximants} we introduce and discuss approximants
and show that they converge 
to weak simulation at some finite level.
The main iterative construction to compute representations of approximants
is described in \cref{sec:wsim-construction}.
Finally, in \cref{sec:wsim-complexity},
we make a closer analysis of the complexity of this procedure.

\subsection{$\omega$-Nets}
  \label{sec:w-nets}
    \begin{definition}
        An \emph{$\omega$-net} $\NN{N}=(Q,\Act,\delta)$ is given by a finite set of
        control states $Q$, a finite set of actions $\Act$ and transitions
        $\delta\subseteq Q\x \Act\x\{-1,0,1,\omega\}\x Q$. It induces a transition system over
        $Q\x\N$
        that allows a step $pm\step{a}qn$ if either $(p,a,d,q)\in\delta$ and $n=m+d\in\N$ or if
        $(p,a,\omega,q)\in\delta$ and $n > m$.

A \emph{path} in $\NN{N}$ is a sequence
$\pi=p_0t_1p_1t_2\dots t_kp_k$ such that
$t_i=(p_{i-1},a_i,d_i,p_i)$
for every $1\le i\le k$.
%
%
We write
$\lambda(\pi)=a_1a_2\dots a_k \in\Act^*$ for the sequence of actions
it induces and $|\pi|=k$ for its length.
The \emph{effect}
$\effect{}(\pi)$
of such a path
is 
the minimum of $\omega$ and $\sum_{1\le i \le k}d_i$.
Its \emph{guard} is 
$\guard{}(\pi) = - \min\{\effect{}(^i\pi)\mid i\le k\}$, where ${}^i\pi$ denotes the prefix
$p_0t_1p_1t_2\dots t_ip_i$ of $\pi$ of length $i$. We call a path empty if $k=0$.
    \end{definition}
Observe that the effect of a path is $\omega$ 
iff it contains at least one $\omega$-transition.
Otherwise the effect is bounded by the length of the path.
Moreover, the guard of a path $\pi$
equals the guard of its longest
prefix without $\omega$-transitions
and therefore satisfies $0\le\guard{}(\pi)\le |\pi|$.


    Every \textOCN\ is an $\omega$-net without $\omega$-transitions.
    Unlike \textOCNs, $\omega$-nets can yield infinitely branching
    transition systems, since each $\omega$-transition $(p,a,\omega,q)$ induces steps 
    $pm\step{a}qn$ for any two naturals $n > m$.
    We observe that, just like \textOCNs, $\omega$-nets are monotone
    in the sense of
    \cref{prop:ocn-monotonicity} (page~\pageref{prop:ocn-monotonicity}):
    \begin{equation}
        pm\step{a}qn \text{ implies } p(m+d)\step{a}q(n+d) \text{ for all }d \in\N.
    \end{equation}
    This means that $pm\ssim p'm'$ implies $pn\ssim p'n'$ for $n\le m$, $n'\ge m'$.
%

    \medskip
    The following theorem justifies our focus on strong simulation games where \V\ plays
    on an $\omega$-net process.
    It shows that checking \textWSIM\ between two \OCN\ processes can be reduced to
    checking strong simulation between a \textOCN\ process and an
    $\omega$-net process.

    \begin{theorem}
        \label{thm:wsim:reduction}
        For two \OCNs\ $\NN{N}$ and $\NN{N'}$ with sets of
        control states $Q$ and $Q'$, respectively, one
        can in polynomial time construct a \OCN\ $\NN{M}$ with control states $M\supseteq Q$
        and an $\omega$-net $\NN{M'}$ with control states $M'\supseteq Q'$ such that
        \begin{equation}
            qn\ \wsim \ q'n' \mbox{ w.r.t. }\NN{N},\NN{N'} \mbox{ iff }
            qn\ssim q'n'\mbox{ w.r.t. }\NN{M},\NN{M'}
            \label{thm:wsim:reduction:main}
        \end{equation}
        holds for 
        each pair $q,q'\in Q\x Q'$ of original control states
        and all $n,n'\in\N$.
    \end{theorem}
    \todo[inline]{SL: better to say: $M$ and $M'$ can be computed in polynomial time}
    \todo[inline]{PT:check out new statement. it also removes the claim about
        approximants}

    The idea of the proof is to look for counter-increasing cyclic paths 
    via $\tau$-labeled transitions in the control graph of $\NN{N}'$
    and to introduce $\omega$-transitions accordingly.
    For any path that reads a single visible action and visits a `generator' state that is
    part of a silent cycle with positive effect, we add an $\omega$-transition.
    For all of the finitely many non-cyclic paths that read a single visible action
    we introduce direct transitions.
    
The remainder of this section is devoted to proving \cref{thm:wsim:reduction}. All further definitions in
this section are only relevant locally.
Formally, the proof of \cref{thm:wsim:reduction} will be done in two steps.
First (\cref{lem:app_reduction:GON}), we reduce weak simulation for \textOCNs\ to strong simulation
between a \textOCN\ and yet another auxiliary model called \emph{guarded $\omega$-nets}.
These differ from $\omega$-nets in that each transition may change the counter by more than one and is
explicitly guarded by an integer, i.e., it can only be applied if the current counter value exceeds the \emph{guard} attached to it.
In the second step (\cref{lem:app_reduction:normalize}) we normalize the effects of all transitions
to $\{-1,0,1,\omega\}$ and eliminate all integer guards and thereby construct an ordinary $\omega$-net
for \V.

\begin{definition}
    A \emph{guarded $\omega$-net} $\NN{N}=(Q,\Act,\delta)$ is given by finite sets $Q$
    and $\Act$ of control states and actions resp., and a transition relation
    $\delta\subseteq Q\x \Act\x\N\x(\Z\cup\{\omega\})\x Q$. 
    It defines a labeled transition system over $Q\x\N$ where $pm\step{a}qn$ iff
    there is a transition $(p,a,g,d,q)\in \delta$ with
    \begin{enumerate}
        \item $m\ge g$ and
        \item $n=m+d\in\N$ or $d=\omega$ and $n>m$.
    \end{enumerate}

\end{definition}

%

Specifically, $\NN{N}$ is an \emph{$\omega$-net} if $g=0$ and $d\in\{-1,0,1,\omega\}$
for all transitions $(p,a,g,d,q)\in\delta$.
The next construction establishes the connection between weak similarity of \textOCNs\
and strong similarity between \OCN\ and guarded $\omega$-net processes. 
In order to avoid confusion we write
$\Step{}{}{\NN{N}}$ and $\WStep{}{}{\NN{N}}$ for (weak) steps in the system
$\NN{N}$.

\begin{lemma}
    \label{L1}
    For any \OCN\ $\NN{N}=(Q,\Act,\delta)$ one can effectively construct a guarded $\omega$-net 
    $\NN{G}=(Q,\Act,\gamma)$ such that for all $a\in \Act$,
    \begin{enumerate}
        \item whenever $pm\WStep{a}{}{\NN{N}}qn$, there is some $l\ge n$ such that
            $pm\Step{a}{}{\NN{G}}ql$
        \item whenever $pm\Step{a}{}{\NN{G}}qn$, there is some $l\ge n$ such that
            $pm\WStep{a}{}{\NN{N}}ql$.
    \end{enumerate}
\end{lemma}

\begin{proof}
    The idea of the proof is to introduce direct transitions
    from one state to another for any path between them
    that reads exactly one visible action and does not contain silent cycles.

    For two states $s,t$ of $\NN{N}$, let $D(s,t)$ be the set of \emph{direct}
    (i.e., acyclic) paths
    from $s$ to $t$ 
    and let $SD(s,t)$ denote the subset of \emph{silent direct} paths
    $SD(s,t) = \{\pi\in D(s,t)\;|\: \lambda(\pi)\in \{\tau\}^*\}$ from $s$ to
    $t$.
    Every path in $D(s,t)$ has acyclic prefixes only and is therefore bounded in length by $|Q|$.
    Hence $D(s,t)$ and $SD(s,t)$ are finite and effectively computable for all pairs $(s,t)$.
    
    Using this notation, we define the transitions in $\NN{G}$ as follows.
    $\NN{G}$ contains a transition $(p,a,\Gamma(\pi),\Delta(\pi),q)$ for each path
    $\pi=\pi_1(s,a,d,s')\pi_2$
    where $\pi_1\in SD(p,s)$ and $\pi_2\in SD(s',q)$.
    This carries over all transitions of $\NN{N}$,
    including the ones with label $a=\tau\in\Act$, because the empty path
    is in $SD(s,s)$ for all states $s$.
    Moreover, introduce $\omega$-transitions in case $\NN{N}$ allows paths $\pi_1,\pi_2$
    as above to contain direct cycles with positive effect on the counter: If there is a path
    $\pi=\pi_1'\pi_1''\pi_1'''(s,a,d,s')\pi_2$ with
    \begin{enumerate}
        \item $\pi_1'\in SD(p,t)$, $\pi_1''\in SD(t,t)$ and $\pi_1'''\in SD(t,s)$
        \item $\Delta(\pi_1'')>0$
    \end{enumerate}
    for some $t\in Q$, then $\NN{G}$ contains a transition $(p,a,\Gamma(\pi_1'\pi_1''),\omega,q)$.
    Similarly, if for some $t\in Q$, there is a path $\pi=\pi_1(s,a,d,s')\pi_2'\pi_2''\pi_2'''$ that satisfies
    \begin{enumerate}
        \item $\pi_1\in SD(p,s)$,$\pi_2'\in SD(s',t)$, $\pi_2''\in SD(t,t)$ and $\pi_2'''\in SD(t,q)$
        \item $\Delta(\pi_2'')>0$
    \end{enumerate}
    add a transition $(p,a,g,\omega,q)$ with guard $g=\Gamma(\pi_1(s,a,d,s')\pi_2'\pi_2'')$.
    If there is an $a$-labeled path from $p$ to $q$ that contains a silent and direct cycle with positive effect,
    $G$ has an a-labeled $\omega$-transition from $p$ to $q$ with the guard derived from that path.

    To prove the first part of the claim, assume $pm\WStep{a}{}{\NN{N}}qn$. By definition of weak steps, there must be a path
    $\pi=\pi_1(s,a,d,s')\pi_2$ with $\lambda(\pi_1),\lambda(\pi_2)\in\{\tau\}^*$.
    Suppose both $\pi_1$ and
    $\pi_2$ do not contain cycles with positive effect. Then there must be paths $\pi_1'\in SD(p,s), \pi_2'\in SD(s',q)$ with
    $\Gamma(\pi_i')\le\Gamma(\pi_i)$ and $\Delta(\pi_i')\ge\Delta(\pi_i)$ for $i\in\{1,2\}$ that can be obtained from
    $\pi_1$ and $\pi_2$ by removing all simple cycles with effects less or equal $0$. So $\NN{G}$ contains a transition
    $(p,a,g',d',q)$ for some $g'\le m$ and $d'\ge n-m$ and hence
    $pm\Step{a}{}{\NN{G}}qn'$ for $n'=m+d'\ge n$. Alternatively, either $\pi_1$ or $\pi_2$ contains a
    cycle with positive effect. Note that for any such path, another path with lower or equal guard exists that connects the
    same states and contains only one such counter-increasing simple cycle:
    If $\pi_1$ contains a simple cycle with positive effect, there is a path $\bar{\pi_1}=\pi_1'\pi_1''\pi_1'''$ from
    $p$ to $s$, where $\pi_1',\pi''$ and $\pi_1'''$ are direct and $\Delta(\pi_1'')>0$ for the cycle
    $\pi_1''\in SD(t,t)$
    for some state $t$. In this case, $\NN{G}$ contains an $\omega$-transition $(p,a,g,\omega,q)$ with
    $g=\Gamma(\pi_1'\pi_1'')$. Similarly, if $\pi_2$ contains a counter-increasing cycle, there is a
    $\bar{\pi_2}=\pi_2'\pi_2''\pi_2'''$, with $\pi_2'\in SD(s',t), \pi_2''\in SD(t,t), \pi_2'''\in SD(t,q)$ and
    $\Delta(\pi_2'')>0$. This means there is a transition $(p,a,g,\omega,q)$ in $\NN{G}$ with
    $g=\Gamma(\pi_1(s,a,d,s')\pi_2'\pi_2'')$. In both cases, $g\le\Gamma(\pi)\le m$ and therefore
    $pm\Step{a}{}{\NN{G}}qi$
    for all $i\ge m$.

    For the second part of the claim, assume $pm\Step{a}{}{\NN{G}}qn$. This
    must be the result of a transition $(p,a,g,d,q)$ in $\NN{G}$ for some $g\le m$. 
    In case $d\neq\omega$, there is a path $\pi$ from $p$ to $q$ with
    $\Delta(\pi)=n-m$, $\lambda(\pi)\in \{\tau\}^*\{a\}\{\tau\}^*$ and $\Gamma(\pi)=g$ that witnesses
    the weak step $pm\WStep{a}{}{\NN{N}}qn$ in $\NN{N}$. Otherwise, if $d=\omega$, there must
    be a path $\pi=\pi_{11}\pi_{12}\pi_{13}(s,a,d,s')\pi_{21}\pi_{22}\pi_{23}$ from $p$ to $q$ in
    $\NN{N}$
    where $\Gamma(\pi)\le m$, all $\pi_{ij}$ are silent and direct and one of $\pi_{12}$ and $\pi_{22}$ 
    is a cycle with strictly positive effect.
    This implies that one can ``pump'' the value of the counter higher than any given value.
    Specifically, there are naturals $k$ and $j$ such that the path $\pi'=\pi_{11}\pi_{12}^k\pi_{13}(s,a,d,s')
    \pi_{21}\pi_{22}^j\pi_{23}$ from $p$ to $q$
    satisfies $\Gamma(\pi')\le\Gamma(\pi)\le m$ and $\Delta(\pi')\ge m-n$. Now $\pi'$
    witnesses the weak step $pm\WStep{a}{}{\NN{N}}qn'$ in $\NN{N}$ for an $l\ge n$.
\end{proof}

\begin{remark}
    \label{rem:w-nets-bounds}
    Observe that no transition of the net $\NN{G}$ as constructed above has a
    guard larger than $3|Q|+1$, nor any finite effect is larger than $2|Q|+1$.
\end{remark}
\begin{lemma}
    \label{lem:app_reduction:GON}
    For a \textOCN\ $\NN{N}'$ one can effectively construct a
    guarded $\omega$-net $\NN{G}'$ over the same set of control states, such that\ for any \OCN\ $\NN{N}$ and
    any two configurations $pm,p'm'$ of $\NN{N}$ and $\NN{N}'$ resp.,
    \begin{equation}
        pm\ \wsim \ p'm'\text{ w.r.t.\ }\NN{N},\NN{N}' \iff pm\ssim p'm'\text{ w.r.t.\ }\NN{N},\NN{G}'.
        \label{lem:app_reduction:GON:claim}
    \end{equation}
\end{lemma}
\begin{proof}
    Consider the construction from the proof of \cref{L1}.
    Let $\WSIM{}{\NN{N},\NN{N}'}$ be the largest weak simulation w.r.t.\ $\NN{N},\NN{N}'$ and
    $\SIM{}{\NN{N},\NN{G}'}$ be the largest strong simulation w.r.t.\ $\NN{N},\NN{G}'$.
    
    For the ``if'' direction we show that $\SIM{}{\NN{N},\NN{G}'}$ is a weak simulation w.r.t.\
    $\NN{N},\NN{N}'$. Assume
    $pm\SIM{}{\NN{N},\NN{G}'}p'm'$ and $pm\Step{a}{}{\NN{N}} qn$.
    That means there is a step $p'm'\Step{a}{}{\NN{G}'} q'n'$ for some
    $n'\in \N$ so that $qn\SIM{}{\NN{N},\NN{G}'} q'n'$.
    By \cref{L1} part 2, $p'm'\WStep{a}{}{\NN{N}} q'l$ for some $l\ge n'$.
    Since simulation is monotonic (point 2 of \cref{prop:ocn-monotonicity}),
    we know that also $qn\SIM{}{\NN{N,}\NN{G}'} q'l$. Similarly, for the ``only if''
    direction, one can use the first claim of \cref{L1} to check that
    $\WSIM{}{\NN{N,}\NN{N}'}$ is a strong simulation w.r.t.\ $\NN{N,}\NN{G}'$.
\end{proof}

\begin{lemma}
    \label{lem:app_reduction:normalize}
    For a \textOCN\ $\NN{N}$
    and a guarded $\omega$-net $\NN{G}'$
    with sets of control states $Q$ and $Q'$
    one can effectively construct a \textOCN\ $\NN{M}$
    and an $\omega$-net $\NN{M}'$
    with sets of control states $M\supseteq Q$ and $M'\supseteq Q'$ respectively,
    such that
    for any two configurations $qn,q'n'$ of $\NN{N}$ and $\NN{G}'$,
    \begin{equation}
        qn\ssim q'n'\text{ w.r.t.\ }\NN{N},\NN{G}' \iff qn\ssim q'n'\text{ w.r.t.\ }\NN{M},\NN{M}'.
        \label{lem:app_reduction:normalize:claim}
    \end{equation}
\end{lemma}
\begin{proof}
    We first observe (see also \cref{rem:w-nets-bounds})
    that for any transition of the guarded $\omega$-net $\NN{G}'$,
    the values of its guard is bounded by some constant. The same holds for
    all finite effects. Let $\Gamma(\NN{G}')$ be the maximal guard and $\Delta(\NN{G}')$
    be the maximal absolute finite effect of any transition of $\NN{G}'$.

    The idea of this construction is to simulate one round of the game $\NN{N}$ vs.\ $\NN{G}'$
    in $k=2\Gamma(\NN{G}')+\Delta(\NN{G}')+1$ rounds of a \sgame\ $\NN{M}$ vs.\ $\NN{M}'$.
    We will replace original steps of both players by sequences of $k$ steps in the new game,
    which is long enough to verify if the guard of \V's move is satisfied and
    adjust the counter using transitions with effects in $\{-1,0,+1,\omega\}$ only.

    We use one fresh symbol $b\notin\Act$ and let the new alphabet be
    $\widehat{\Act}=\Act\cup\{b\}$.
    We transform the net $\NN{N}=(Q,\Act,\delta)$ to the \textOCN\
    $\NN{M}=(M,\widehat{\Act},\mu)$ as follows:
    \begin{align}
         M =\ &Q\cup \{p_i\;|\;1\le i<k, p\in Q\}\\
     \mu =\ &\{p\step{a,d}q_k\;|\; p\step{a,d}q\in\delta\}\\
                  &\cup \{p_i\step{b,0}p_{i-1}\;|\;1<i<k\}\\
                  &\cup \{p_1\step{b,0}q\}.
    \end{align}
    We see that
    \begin{equation}
        pm\Step{a}{}{\NN{N}}qn \iff pm\Step{a}{}{\NN{M}}q_{k-1}n\Step{b^{k-2}}{}{\NN{M}}q_1n\Step{b}{}{\NN{M}}qn.
        \label{spoilernet}
    \end{equation}
    
    Now we transform the guarded $\omega$-net $\NN{G}'=(Q',\Act,\delta')$ to the $\omega$-net
    $\NN{M}'=(M',\widehat{\Act},\mu')$.
    Every original transition will be replaced by a sequence of $k$ transitions
    that test if the current counter value exceeds the guard
    and adjust the counter accordingly.
    The new net $\NN{M}'$ has all states of $\NN{G}'$ plus a chain of $k$ new states for each original
    transition.
    \begin{equation}
         M' = Q' \cup \{t_i\;|\;0\le i<k, t\in \delta'\}.
    \end{equation}
    For every transition $t=(p,a,g,d,q)$ in $\NN{G}'$,
    we add the following transitions to $\NN{M}'$. First, to test the guard:
    \begin{align}
         p\step{a,0}t_{k-1},\\
         t_i\step{b,-1}t_{i-1}, &\text{ for } k-g<i<k\\
         t_i\step{b,+1}t_{i-1}, &\text{ for } k-2g<i<k-g.
    \end{align}
    Now we add transitions to adjust the counter according to $d\in\N\cup\{\omega\}$.
    In case $0\le d<\omega$ we add
    \begin{align}
         t_i\step{b,+1}t_{i-1}, &\text{ for } k-2g-|d|<i<k-2g\\
         t_i\step{b,0}t_{i-1}, &\text{ for } 0\le i<k-2g-d.
    \end{align}
    In case $d<0$ we add
    \begin{align}
         t_i\step{b,-1}t_{i-1}, &\text{ for } k-2g-|d|<i<k-2g\\
         t_i\step{b,0}t_{i-1}, &\text{ for } 0\le i<k-2g+d.
    \end{align}
    In case $d=\omega$ we add
    \begin{align}
         t_i\step{b,\omega}t_{i-1}, &\text{ for } i=k-2g\\
         t_i\step{b,0}t_{i-1}, &\text{ for } 0\le i<k-2g.
    \end{align}
    Finally, we allow a move to the new state:
    \begin{equation}
         t_0\step{b,0}q.
    \end{equation}
    Observe that every transition in the constructed net $\NN{M}'$
    has effect in $\{-1,0,+1,\omega\}$. $\NN{M}'$ is therefore an ordinary $\omega$-net.
    It is straightforward to see that
    \begin{equation}
        pm\Step{a}{}{\NN{G}'}qn \iff pm\Step{ab^{k-1}}{}{\NN{M}'}qn.\label{dupnet}
    \end{equation}
    The equation \eqref{lem:app_reduction:normalize:claim} now follows from \cref{spoilernet,dupnet}.
\end{proof}

    \Cref{thm:wsim:reduction}
    now follows from \cref{lem:app_reduction:GON,lem:app_reduction:normalize}.

    \todo[inline]{PT: i have no idea where the paragraph below comes from and
        propose to remove it.
        COMMENTED OUT
    }
    \todo[inline]{PT: i don't like the next subsection title}
%
    

\subsection{Outline of the Construction}
  \label{sec:wsim-flavour}
It remains to show how to solve a strong \sgame\
between \R, playing on a \textOCN\ $\NN{N}$
and \V, playing on an $\omega$-net $\NN{N}'$.
Let us consider the following situation to get a flavor of the reasoning in the remaining part of
the \cref{sec:wsim}:
assume that the structure of $\NN{N}$ and $\NN{N}'$
guarantees that in any play of a \sgame, at most one
\emph{$\omega$-step} i.e. 
a step induced by
an $\omega$-transition, can be used. 
Consider a prefix of a play until the $\omega$-step, and let us assume that
after this prefix \R's configuration is $q n$ and that the $\omega$-step ends in the
configuration with the state $q'$. Observe that \R\ wins only if $n$ is big enough such that $q n$ is
not simulated by $q' n'$ for any $n'\in \N$. Otherwise, \V\ would choose a value
$n'$ big enough to simulate $q n$.
Moreover, observe that in order to find the minimal $n$ with which \R\ can win we need to investigate
only the simulation preorder between two \textOCNs, since after the $\omega$-step there
are no further $\omega$-steps allowed (by our assumption above). Namely, these nets are $\NN{N}$ and $\NN{N}'$
with all $\omega$-transitions removed.
We ask about the belt for $q,q'$.
\R\ wins the remaining play iff
\begin{enumerate}
  \item this belt is vertical (some $n_0$ exists with $qn_0\notSIM{}{}{}q'n'$ for all $n'$) and
  \item $n$ is larger than the width of this belt ($n$ is already sufficient).
\end{enumerate}
Assuming
that we have calculated $n_0$, we can design a gadget which will be
substituted instead of the $\omega$-transition in $\NN{N}'$ and which allow to test
if the \R's counter value is greater than $n_0$. 

This lets us transform the pair of \textOCN\ and $\omega$-net
into a pair of \textOCNs, in such way that preserves the outcome of all those
plays in which at most one $\omega$-step is used.
The overall approach is to iterate this procedure,
constructing a sequence of \textOCNs\ that approximate the behavior of the
original nets.
In \cref{sec:wsim-approximants} we define the notion of simulation approximants
and show that they stabilize at some finite level.
In \cref{sec:wsim-construction,sec:wsim-complexity}
we explain how to represent these approximant relations using the idea above,
and how efficient this representation is, i.e., how many iterations are necessary.
In the next section we briefly go back to
strong simulation between \textOCNs, and show how to check the two conditions
1) and 2) above.

\subsection{Computing Minimal Sufficient Values}
  \label{sec:wsim-msv}
In this section we present that computing the exact width of \emph{vertical} belts can be done in polynomial space.
It will be used in \cref{sec:wsim-complexity}.

Let us write $\suff{q,q'}$ for the least value $n\in\N$
such that $qn\notSIM{}{}q'n'$ for every $n'\in\N$ and $\omega$ if no such value
$n$ exists.
In terms of the \sgame, this is the minimal initial counter value
that is sufficient for \R\ to win against any initial value for  \V\ if we
fix the initial states to $q$ and $q'$.
Observe that $\suff{q,q'}=\omega$ iff the belt for the plane
$(q,q')$ is \emph{not} vertical.

The following is an easy consequence of \cref{thm:ssim-pspace},
because one can check the simulation problem
for selected positions.

\begin{lemma}\label{lem:compute-vbelts}
    Given \OCNs\ $\NN{N}$ and $\NN{N'}$ in normal form with sets of control
    states $Q$ and $Q'$,
    for any given pair $(q,q')\in Q\x Q'$ of control states,
    the value $\suff{q,q'}$ can be computed in \PSPACE.
    Moreover, if $\suff{q,q'}\neq\omega$, then it is bounded by $\Cacyc$,
    the maximal length of an acyclic path in the product of $\NN{N}$ and
    $\NN{N'}$.
\end{lemma}
\begin{proof}
  By \cref{thm:belt-theorem-bounds}, we can bound the coefficients of the slopes
  of all belts polynomially. In particular, we know that if $(\rho,\rho')$
  is the slope of some belt then $\rho$ and $\rho'$ are both non-negative
  and no bigger than $\Cacyc\le |Q\x Q'|$.
  The steepest possible such slope that is not vertical (i.e., with $\rho>0$) is thus
  given by the vector $(\alpha,\alpha') = (1,\Cacyc)$.
  
  To check if $\suff{q,q'}=\omega$ 
  we can pick a point $(n,n')$
  that is both $\Cacyc$-above $(\alpha,\alpha')$ and $\Cacyc$-below the vertical vector $(0, 1)$
  and check if $qn\SIM{}{}q'n'$ holds.
  For instance, $n = \Cacyc+1$ and $n' = 2(\Cacyc+1)^2$
  is surely such a point.
  If $\suff{q,q'}\neq\omega$, then the belt for $(q,q')$ is vertical and by
  \cref{thm:belt-theorem-bounds}, point 2, we have $qn\notSIM{}{}q'n'$.
  Otherwise, the belt is not vertical and has slope $(\rho,\rho')\lesssteep
  (\alpha,\alpha')$. Then by point 1 of
  \cref{thm:belt-theorem-bounds}, we must have $qn\SIM{}{}q'n'$.

  To compute $\suff{q,q'}\in\N$ for a vertical belt 
  recall that by point 1 of \cref{thm:belt-theorem-bounds},
  $qn\notSIM{}{}q'n'$ for all points with $n>\Cacyc$. Clearly, this means that
  $\suff{q,q'}$ is bounded by $\Cacyc\le |Q\x Q'|$.
  By \cref{lem:ssim:uperiodic},
  the coloring on this belt
  must be repetitive from some exponentially bounded level $n'_0$ onwards.
  By monotonicity, this means that the coloring of the belt must
  have stabilized at this level already, so that for all $n'\ge n'_0$, we have
  $qn\SIM{}{}q'n'$ iff $n< \suff{q,q'}$.

  We can now iteratively check the color of the point
  $(n,n'_0)$ for decreasing values $n\in\N$, starting with $\Cacyc$.
  By \cref{thm:ssim-pspace}, this can surely be done in polynomial space.
  The value $\suff{q,q'}$ must be the largest value $n<\Cacyc$, such that
  $qn\notSIM{}{}q'n'_0$ holds.
\end{proof}

\subsection{Approximants}
  \label{sec:wsim-approximants}

\newcommand{\CA}{\mathit{C\!A}}
\newcommand{\CB}{\mathit{C\!B}}
The basic idea of our procedure for
checking \textSSIM\ between a \OCN\ and an $\omega$-net,
and therefore \textWSIM\ between two \OCN,
is to stepwise compute
semilinear over-approximations
$\SIM{}{i}\;\supseteq\;\ssim$. For such a procedure to be effective,
it is crucial that these approximants converge to $\ssim$
at some finite level, i.e., $\SIM{}{k}\;=\;\SIM{}{k+1}\;=\;\ssim$
for some $k<\omega$.
Unfortunately, the usual \textSSIM\ approximants (see \cref{def:normal_approximants})
do not have this property, as
\cref{ex:weaksim:nonconvergence}
(page~\pageref{ex:weaksim:nonconvergence})
shows.

%
We overcome this difficulty by generalizing
the notion of $\SIM{}{\alpha}$ simulation approximants in 
the case of simulation between one-counter and $\omega$-net processes.
This yields approximants that indeed converge at a finite level for any pair of nets.

First we define approximants $\SIM{\beta}{\alpha}$ in two (ordinal) dimensions.
From the game perspective the subscript $\alpha$
indicates the number of rounds \V\ can survive
and the superscript $\beta$ denotes the number of \emph{$\omega$-steps} \R\ may allow before she loses, where $\omega$-step is a step induced by a $\omega$-transition.
For example, $qn\SIM{2}{5}q'n'$ 
holds if
\V\ survives round 5 of the \sgame\ or makes his second $\omega$-move until then.
If not stated otherwise we assume that $\NN{N}=(Q,\Act, \delta)$ is a \textOCN\
and $\NN{N'}=(Q',\Act,\delta')$ is an $\omega$-net.

\begin{definition}\label{def:approximants}
    For ordinals $\alpha$ and $\beta$, the
    \emph{approximant} $\SIM{\beta}{\alpha}$ 
    is inductively defined as follows.
    Let $\SIM{0}{\alpha}\;=\;\SIM{\beta}{0}\;=\;Q\x\N\x Q'\x\N$, the full relation.
    For successor ordinals $\alpha+1,\beta+1$ let $pm \SIM{\beta+1}{\alpha+1} p'm'$ iff for all
    $pm\step{a}qn$ there is a step $p'm'\step{a}q'n'$ such that either
    \begin{enumerate}
        \item $(p',a,\omega, q')\in\delta'$,
            $m'<n'$
            and $qn\SIM{\beta}{\alpha} q'n'$, or
        \item
            $(p',a,d, q')\in\delta'$,
            $n'=m'+d \in \N$
          and $qn\SIM{\beta+1}{\alpha} q'n'$.
    \end{enumerate}
    For limit ordinals $\lambda$ we define
    $\SIM{\lambda}{\alpha} \;=\; \bigcap_{\beta<\lambda}\SIM{\beta}{\alpha}$ and
    $\SIM{\beta}{\lambda} \;=\; \bigcap_{\alpha<\lambda}\SIM{\beta}{\alpha}$.
    Finally,
    \begin{align}\label{def:big_approximants}
        \SIM{\beta}{} \;=\; \bigcap_{\alpha\in\Ord} \SIM{\beta}{\alpha}
        &&\SIM{}{\alpha} \;=\; \bigcap_{\beta\in\Ord} \SIM{\beta}{\alpha}.
    \end{align}
\end{definition}\smallskip

\noindent Notice that the approximant $\SIM{\beta+1}{\alpha+1}$ above
is defined in terms of both 
$\SIM{\beta+1}{\alpha}$ and
$\SIM{\beta}{\alpha}$.
The first condition in its definition
asks that if a response is via a $\omega$-step
then the resulting pair of processes need to be related by the approximant
with reduced superscript $\beta$.
The second condition is for the case where a response is via a step induced 
by an ordinary transition.

    The approximants $\SIM{}{\alpha}$ correspond to the usual notion of simulation approximants
    defined on page~\pageref{def:normal_approximants}
    and $\SIM{\beta}{}$ is a special notion derived from the syntactic peculiarity of
    $\omega$-transitions present in the game on one-counter vs.\ $\omega$-nets.

\begin{example}
    \label{ex:wsim:a-b-convergence}
    Consider the net that consists of a single $a$-labeled loop in state $A$ and the
    $\omega$-net with transitions $B\step{a,\omega}C\step{a,-1}C$ only.
    This is a variant of the system of \fullref{ex:weaksim:nonconvergence},
    but now we are interested in \emph{strong} \textSSIM.
    We see that for any $m,n\in\N$,
    $Am\SIM{}{n}Cn \notSIMBY{}{n+1}Am$. Moreover,
    $Am\SIM{}{\omega}Bn$ but $Am\notSIM{}{\omega+1}Bn$ and
    $Am\SIM{1}{}Bn$ but $Am\notSIM{2}{\omega+1}Bn$ and therefore $Am\notSIM{2}{}Bn$.
\end{example}

We will further use a game characterization of these approximants, similar to
the \sgames\ that characterize strong simulation.

Intuitively, $\SIM{i}{}$ is given by a \emph{parameterized simulation game} that keeps track of
how often \V\ uses $\omega$-steps and in which
\V\ immediately wins if he plays such a step the $i$th time.
It is easy to see that this game favors \V\ due to the additional winning condition.
Hence, $\forall i\in\N,\;\SIM{i}{}\;\supseteq\;\SIM{i+1}{}$.
With growing index $i$, this advantage becomes less important and the game
increasingly resembles a standard simulation game.

\begin{definition}\label{def:approxgame}
    An \emph{\agame} is played in rounds between \R\ and \V.
    Game positions are quadruples $(pm,p'm',\alpha,\beta)$ where $pm, p'm'$ are configurations
    of $\NN{N}$ and $\NN{N}'$ respectively, and $\alpha,\beta$ are ordinals called step- and $\omega$-counter.
    In each round that starts in $(pm,p'm',\alpha,\beta)$:
    \begin{itemize}
      \item \R\ chooses two ordinals $\hat{\alpha}<\alpha$ and $\hat{\beta}<\beta$,
      \item \R\ makes a step $pm\step{a}qn$,
      \item \V\ responds by making a step $p'm'\step{a}q'n'$ induced by a transition $t$.
    \end{itemize}
    If $t$ was an $\omega$-transition then the game continues from position $(qn,q'n',\hat{\alpha},\hat{\beta})$.
    Otherwise the next round starts at $(qn,q'n',\hat{\alpha},\beta)$
    (in this case \R's choice of $\hat{\beta}$ becomes irrelevant).
    If a player cannot move then the other player wins and if $\alpha$ or $\beta$ becomes $0$, \V\ wins.
\end{definition}

\begin{lemma}
    \label{approximants:game:monotonicity}
    If \V\ wins the \agame\ from $(pm,p'm',\alpha,\beta)$
    then he also wins the game from $(pm,p'm',\hat{\alpha},\hat{\beta})$ for any $\hat{\alpha}\le \alpha$
    and $\hat{\beta}\le\beta$.
\end{lemma}
\begin{proof}
If \V\ has a winning strategy in the game from $(pm,p'm',\alpha,\beta)$ then he can
use the same strategy in the game from
$(pm,p'm',\hat{\alpha},\hat{\beta})$ and maintain the invariant that the pair of
ordinals in the game configuration is pointwise smaller than the pair in the
original game. Thus \V\ wins from $(pm,p'm',\hat{\alpha},\hat{\beta})$.
\end{proof}

\begin{lemma}[Game Characterization]
    \label{lem:game_interpretation}
    \V\ has a strategy to win the \agame\
    that starts in $(pm,p'm',\alpha,\beta)$
    iff $pm\SIM{\beta}{\alpha}p'm'$.
\end{lemma}
\begin{proof}
    We say a pair $(\alpha,\beta)\in \Ord^2$ of ordinals \emph{dominates}
    another such pair $(\alpha',\beta')$ iff
    $\alpha'\le \alpha$, $\beta'\le\beta$ and $(\alpha',\beta')\neq(\alpha, \beta)$.
    Both directions of the claim are now shown by well-founded induction on
    pairs of ordinals:
    If the claim holds for all pairs $(\alpha',\beta')$ that are
    dominated by
    $(\alpha,\beta)$ then it also holds for $(\alpha,\beta)$.

    For the ``if'' direction we assume $pm\SIM{\beta}{\alpha}p'm'$ and show that \V\ wins the game
    from $(pm,p'm',\alpha,\beta)$.
    In the base case of $\alpha=0$ or $\beta=0$ \V\ directly wins by definition.
    By induction hypothesis we assume that the claim is true for all pairs
    dominated by $(\alpha,\beta)$.
    \R\ starts a round by picking ordinals $\hat{\alpha}<\alpha$ and $\hat{\beta}<\beta$
    and moves $pm\step{a}qn$.
    We distinguish two cases, depending on whether $\beta$ is a limit or
    successor ordinal.

    \case{1} $\beta$ is a successor ordinal.
    By \cref{approximants:game:monotonicity}, we can safely assume that
    $\hat{\beta}=\beta-1$.
    By our assumption $pm\SIM{\beta}{\alpha}p'm'$ and \cref{def:approximants},
    there must be a response $p'm'\step{a}q'n'$ that is either due to an $\omega$-transition
    and then $qn\SIM{\hat{\beta}}{\hat{\alpha}}q'n'$ or due to an ordinary transition, in which case we have
    $qn\SIM{\beta}{\hat{\alpha}}q'n'$. In both cases, we know by the induction hypothesis
    that \V\ wins from this next position and thus also from the initial position.

    \case{2} $\beta$ is a limit ordinal.
    By $pm\SIM{\beta}{\alpha}p'm'$ and \cref{def:approximants},
    we obtain $pm\SIM{\gamma}{\alpha} p'm'\text{ for all }\gamma <\beta$.
    If $\alpha$ is a successor ordinal then, by \cref{approximants:game:monotonicity},
    we can safely assume that $\hat{\alpha}= \alpha -1$.
    Otherwise, if $\alpha$ is a limit ordinal, then, by
    \cref{def:approximants}, we have
    $pm\SIM{\gamma}{\ol{\alpha}} p'm'\text{ for all }\ol{\alpha} <\alpha$ and in
    particular $pm\SIM{\gamma}{\hat{\alpha}+1} p'm'$. So in either case we obtain  
    \begin{equation}
        \label{eq:game1}
        pm\SIM{\gamma}{\hat{\alpha}+1} p'm'\text{ for all }\gamma <\beta.
    \end{equation}
    If there is some $\omega$-step that allows a response
    $p'm'\Step{a}{}{\omega} q'n'$ that satisfies $qn\SIM{\hat{\beta}}{\hat{\alpha}}q'n'$,
    then \V\ picks this response and we can use the induction hypothesis
    to conclude that he wins the game from the next position.
    Otherwise, if no such $\omega$-step exists,
    \cref{eq:game1} implies that for every $\gamma <\beta$
    there is a response to some $q'n'$ via a step induced by a non-$\omega$-transition $t(\gamma)$
    and that satisfies $qn\SIM{\gamma}{\hat{\alpha}}q'n'$.
    Since $\beta$ is a limit ordinal, there exist infinitely many $\gamma < \beta$.
    By the pigeonhole principle, there must be one transition that occurs as
    $t(\gamma)$ for infinitely many $\gamma$,
    because there are only finitely many transitions in the net.
    Therefore, a response via a step induced by this particular transition
    satisfies $qn\SIM{\beta}{\hat{\alpha}}q'n'$.
    If \V\ uses this response, the game continues from position $(qn,q'n',\hat{\alpha},\beta)$
    and he wins by induction hypothesis.

    For the ``only if'' direction we show that $pm\notSIM{\beta}{\alpha}p'm'$ implies that \R\
    has a winning strategy in the \agame\ from $(pm,p'm',\alpha,\beta)$.
    In the base case of $\alpha=0$ or $\beta=0$ the implication holds
    trivially since the premise is false.
    By induction hypothesis, we assume that the implication is true
    for all pairs dominated by $(\alpha,\beta)$.
    Observe that if $\alpha$ or $\beta$ are limit ordinals then
    (by \cref{def:approximants}) there are successor ordinals
    $\hat{\beta}\leq \beta$ and $\hat{\alpha}\leq \alpha$ such that\
    $pm\notSIM{\hat{\beta}}{\hat{\alpha}}p'm'$. So without loss of
    generality we can assume that $\alpha$ and $\beta$ are successor ordinals.
    By the definition of approximants there must be a move $pm\step{a}qn$ such that
    \begin{itemize}
        \item for every response $p'm'\Step{a}{}{\omega}q'n'$ that uses
            some $\omega$-step we have $qn\notSIM{\beta-1}{\alpha-1}q'n'$,
        \item for every response $p'm'\step{a} q'n'$ via
            some step induced by not $\omega$-transition it holds that $qn\notSIM{\beta}{\alpha -1}q'n'$.
    \end{itemize}
    So if \R\ chooses $\hat{\alpha}=\alpha-1$, $\hat{\beta}=\beta-1$ and moves $pm\step{a}qn$
    then any possible response by \V\ will take the game to a
    position $(qn,q'n',\gamma,\hat{\alpha})$ for some $\gamma\le \beta$.
    By induction hypothesis \R\ wins the game.
\end{proof}
\begin{lemma}
        \label{lem:wsim:appr:props}
        For all ordinals $\alpha,\beta$ the following properties hold.
        \begin{enumerate}
            \item $pm\SIM{\beta}{\alpha}p'm'$ implies $pn\SIM{\beta}{\alpha}p'n'$
                for all $n\le m$ and $n'\ge m'$
                \label{lem:wsim:appr:props:monotonicity}
            \item If $\hat{\alpha}\geq \alpha$ and $\hat{\beta}\ge \beta$ then
                $\SIM{\hat{\beta}}{\hat{\alpha}}\;\subseteq\;\SIM{\beta}{\alpha}$\label{lem:wsim:appr:props:alpha-inc}.
            \item There are ordinals $\CA,\CB$ such that 
                $\SIM{}{\CA}\;=\;\SIM{}{\CA+1}$ and $\SIM{\CB}{}\;=\;\SIM{\CB+1}{}$.
                \label{lem:wsim:appr:props:convergence_ordinals}
            \item $\SIM{}{}\;=\;\bigcap_{\alpha}\SIM{}{\alpha} \;=\;\bigcap_\beta\SIM{\beta}{}$
                \label{lem:wsim:appr:props:convergence}
        \end{enumerate}
    \end{lemma}

    The first point states that individual approximants are monotonic with respect to the counter
    values.
    Points~\ref{lem:wsim:appr:props:alpha-inc}-\ref{lem:wsim:appr:props:convergence}
    imply that both $\SIM{}{\alpha}$ and $\SIM{\beta}{}$
    yield non-increasing sequences of approximants that converge towards simulation.
    As \fullref{ex:wsim:a-b-convergence} shows,
    the approximants $\SIM{}{\alpha}$ do not
    converge at finite levels, and not even at level $\omega$, i.e., $\CA>\omega$ in
    general. We will later show (in \cref{thm:wsim:approximants:convergence})
    that the approximants $\SIM{\beta}{}$
    converge at a finite level, i.e., $\CB$ is \emph{strictly} below $\omega$
    for any pair of nets,
    and further we bound $\CB$ in \cref{sec:wsim-complexity} to obtain an exact complexity
    upper bound.

    \begin{proof}\hfill
\begin{enumerate}
        \item By \cref{lem:game_interpretation}, it suffices
     to observe that \V\ can reuse a winning strategy in the \agame\
     from $(pm,p'm',\alpha,\beta)$ to win the game from $(pn,p'n',\alpha,\beta)$
     for naturals $n\le m$ and $n'\ge m'$.

     \item If $pm\SIM{\hat{\beta}}{\hat{\alpha}}p'm'$ then, by \cref{lem:game_interpretation}, \V\ wins
     the \agame\ from position $(pm,p'm', \hat{\beta}, \hat{\alpha})$.
     By \cref{approximants:game:monotonicity} he can also win the \agame\
     from $(pm,p'm', \beta,\alpha)$. Thus $pm\SIM{\beta}{\alpha}p'm'$ by \cref{lem:game_interpretation}.

    \item By point 2) we see that with increasing ordinal index $\alpha$ the approximant
     relations $\SIM{}{\alpha}$ form a decreasing sequence of relations,
     thus they stabilize for some ordinal $\CA$.
     The existence of a convergence ordinal for $\SIM{\CB}{}$ follows analogously.

     \item First we observe that
     $\bigcap_{\alpha}\SIM{}{\alpha}\;=\;\bigcap_{\alpha} \bigcap_\beta\SIM{\beta}{\alpha}\;=\;
     \bigcap_{\beta} \bigcap_\alpha \SIM{\beta}{\alpha}\;=\;
     \bigcap_\beta\SIM{\beta}{}$.
     It remains to show that $\SIM{}{}\;=\;\bigcap_{\alpha} \SIM{}{\alpha}$.
     In order to show $\SIM{}{}\;\supseteq\bigcap_{\alpha} \SIM{}{\alpha}$,
     we use $\CA$ from point \ref{lem:wsim:appr:props:convergence_ordinals}) and
     rewrite the right side of the inclusion to
     $\bigcap_{\alpha} \SIM{}{\alpha} \;=\; \SIM{}{\CA} \;=\; \SIM{}{\CA+1}$.
     From \cref{def:approximants} we get that $\SIM{}{\alpha}\;=\;\SIM{\gamma}{\alpha}$
     for $\gamma\ge \alpha$ and therefore
     $\SIM{\CA+1}{\CA+1}\:=\:\SIM{}{\CA+1}\:=\:\SIM{}{\CA}\:=\:\SIM{\CA}{\CA}$.
     We see that for every \R's move to a configuration $pn$ from a configuration in $\SIM{\CA}{\CA}$ 
there is \V's response to a configuration $p'n'$ such that $pn \SIM{\CA}{\CA}p'n'$.
     This means $\SIM{\CA}{\CA}\;=\bigcap_{\alpha}\SIM{}{\alpha}$ is a simulation
     relation and hence a subset of $\ssim$.

     To show $\SIM{}{}\;\subseteq \bigcap_{\alpha} \SIM{}{\alpha}$,
     we prove by induction that $\SIM{}{}\;\subseteq\;\SIM{}{\alpha}$ holds for all ordinals $\alpha$.
     The base case $\alpha=0$ is trivial.
     For the induction step we prove the equivalent property
     $\notSIM{}{\alpha}\:\subseteq\:\notSIM{}{}$.
     There are two cases.

     In the first case, $\alpha=\gamma+1$ is a successor ordinal.
     If $pm\notSIM{}{\gamma+1}p'm'$ then
     $pm\notSIM{\gamma+1}{\gamma+1}p'm'$ and therefore,
     by \cref{lem:game_interpretation},
     \R\ wins the \agame\ from
     $(pm,p'm',\gamma+1,\gamma+1)$.
     Let $pm\step{a}qn$ be an optimal initial move by \R.
     Now either there is no valid response and thus \R\ immediately wins
     in the \sgame;
     or for every \V\ response $p'm'\step{a}q'n'$
     that uses an $\omega$-step, we have
     $qn\notSIM{\gamma}{\gamma}q'n'$
     and for every response that does not use an $\omega$-move, we have
     $qn\notSIM{\gamma+1}{\gamma}q'n'$.
     Either way, we get
     $qn\notSIM{}{\gamma}q'n'$ and by induction
     hypothesis, $qn\notSIM{}{}q'n'$.
     By \cref{lem:game_interpretation}, we obtain that
     \R\ wins the \sgame\ from $(qn,q'n')$ and thus from $(pm,p'm')$.
     Therefore $pm\notSIM{}{}p'm'$, as required.

     In the second case, $\alpha$ is a limit ordinal. Then $pm\notSIM{}{\alpha}p'm'$ implies
     $pm\notSIM{}{\gamma}p'm'$ for some $\gamma< \alpha$
     and therefore $pm\notSIM{}{}p'm'$ by induction
     hypothesis.\qedhere
\end{enumerate}
     \end{proof}

\noindent The following lemma shows a certain uniformity property of the \sgame.
Beyond some fixed bound, an increased counter value of Spoiler can be
neutralized by an increased counter value of Duplicator, thus enabling Duplicator
to survive at least as many rounds in the game as before.
This lemma is necessary for the proof of \cref{thm:wsim:approximants:convergence},
which guarantees the existence of a finite bound for the convergence level $\CB$.

    \begin{lemma}\label{thm:approximants:colouring}
        For any \textOCN\ $\NN{N}=(Q,\Act,\delta)$ and $\omega$-net $\NN{N'}=(Q',\Act,\delta')$
        there is a fixed bound $c\in\N$ such that~for all states $(p,p')\in Q\x Q'$,
        naturals $n>m>c$
        and ordinals $\alpha$:
         \begin{equation}
             \label{thm1:eq}
             \forall m'.\ ( pm\SIM{}{\alpha} p'm' \implies \exists n'.\ pn\SIM{}{\alpha} p'n')
         \end{equation}
    \end{lemma}
    \begin{proof}
        It suffices to show the existence of a local bound $c$
        that satisfies \eqref{thm1:eq}
        for any given pair of states, since
        we can simply take the global $c$ to be the maximal such bound
        over all finitely many pairs.
        Let $\CA$ be the convergence ordinal provided
        by \cref{lem:wsim:appr:props}, point~\ref{lem:wsim:appr:props:convergence_ordinals}
        and consider a fixed pair $(p,p')\in (Q\x Q')$ of states.
        For $m,m'\in\N$, we define the following (sequences of) ordinals.
        \begin{align*}
            I(m,m') = &\ \text{the largest ordinal $\alpha$ with } pm\SIM{}{\alpha} p'm'
            \text{ or }\CA\\ &\ \text{if no such $\alpha$ exists},\\
            I(m) = &\ \text{the increasing sequence of ordinals $I(m,m')_{m'\ge 0}$},\\
            S(m) = &\ \sup\{I(m)\}.
        \end{align*}

\noindent Observe that $I(m,m')$ can be presented as an infinite matrix where $I(m)$ is a column and
        $S(m)$ is the limit of the sequence of elements of column $I(m)$ looking upwards.
        Informally, $S(m)=lim_{i\rightarrow \infty} I(m,i).$

\begin{center}
 \includegraphics[angle=0]{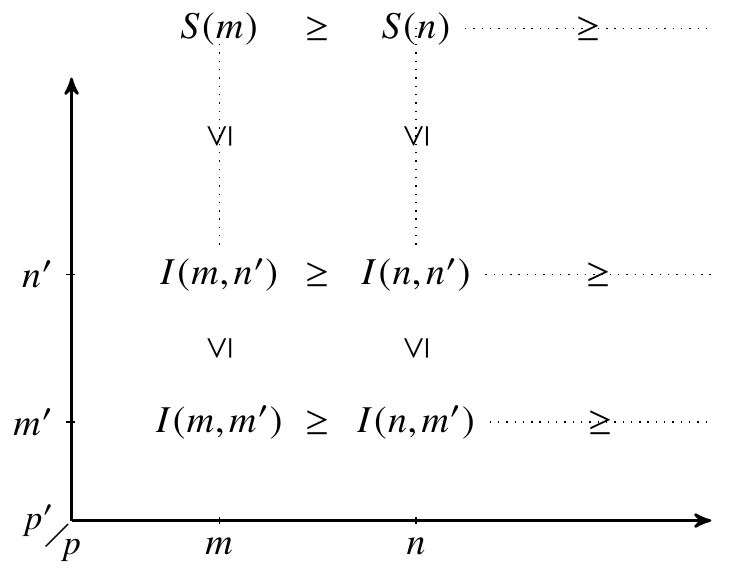}
\end{center}

        By \cref{lem:wsim:appr:props} (point \ref{lem:wsim:appr:props:monotonicity}), we derive
        that for any $m\le n\in\N$ and $m'\le n'\in\N$
        \begin{equation}
            I(m,n')\ge I(m,m') \ge I(n,m')\label{eq:approximants:thm1:mon1}
        \end{equation}
        and because of two inequalities, also that $S(m)\ge S(n)$.
        So the ordinal sequence $S(m)_{m\ge0}$ of suprema must be non-increasing and by the
        well-ordering of the ordinals there is a smallest index $k \in \N$ at
        which this sequence stabilizes:
        \begin{equation}
            \forall l>k.\ S(l) = S(k).
        \end{equation}  
        We split the remainder of this proof into three cases depending on whether $I(k)$ and $I(l)$
        for some $l>k$ have maximal
        elements. In each case we show the existence of a bound $c$ that satisfies requirement \eqref{thm1:eq}.

        \case{1} For all $l\ge k$ and $m'\in\N$ it holds that $I(l,m')<S(l)$,
        i.e., no $I(l)$ has a maximal element.
        In this case $c := k$ satisfies the requirement \eqref{thm1:eq}.
        To see this, take $n>m>c=k$ and $pm\SIM{}{\alpha} p'm'$.
        Then, by our assumption, $\alpha<S(m)$ and $S(m)=S(n)=S(k)$.
        Therefore $\alpha < S(n)$, which means that there must exist an $n' \in \N$ such that
        $pn\SIM{}{\alpha} p'n'$, as required.

        \case{2} For all $l\ge k$ there is a $n'_l \in \N$ such that $I(l,n'_l) = S(l)$,
        i.e., all $I(l)$ have maximal element $S(l)=S(k)$.
        Again $c := k$ satisfies the requirement \eqref{thm1:eq}.
        Given $n>m>c=k$ and $pm\SIM{}{\alpha}p'm'$ we let $n' := n'_{n}$ and obtain
        $I(n,n') = S(n) = S(k) \ge\alpha$ and thus $pn\SIM{}{\alpha} p'n'$, as required.

        \case{3} If none of the two cases above holds then there must exist
        some $l > k$ such that~the sequences $I(k), \dots, I(l-1)$ each have a maximal element and for
        $i >l$ the sequence $I(i)$ has no maximal element.
        To see this, consider sequences $I(m)$ and $I(n)$ with $n > m \ge k$.
        If $I(n)$ has a maximal element then so must $I(m)$, by
        \cref{eq:approximants:thm1:mon1} and $S(m)=S(n)=S(k)$.
        Given this, we repeat the argument for the first case, with $c:=l$ and
        again satisfy the requirement \eqref{thm1:eq}.
    \end{proof}

\begin{lemma}\label{thm:wsim:approximants:convergence}
    Consider strong simulation $\SIM{}{}$ between a \OCN\
    $\NN{N}=(Q,\Act,\delta)$ and an $\omega$-net $\NN{N'}=(Q',\Act,\delta')$.
    There exists a constant $\CB\in\N$ such that $\SIM{}{}\:=\:\SIM{\CB}{}$.
\end{lemma}
\begin{proof}
We assume the contrary and derive a contradiction. By
\cref{lem:wsim:appr:props}, part~\ref{lem:wsim:appr:props:convergence},
the inclusion
$\SIM{}{}\;\subseteq\;\SIM{\beta}{}$ always holds for every ordinal $\beta$.
Thus, if there is no $\CB\in\N$ with $\SIM{}{}\;=\;\SIM{\CB}{}$,
then for every finite
$\beta \in \N$ there are processes $p_0m_0$ and $p'_0m'_0$
such that $p_0m_0\SIM{\beta}{}p'_0m'_0$ but $p_0m_0 \notSIM{}{} p'_0m'_0$.
In particular, this holds for the special case of $\beta = |Q\x Q'|(c +1)$,
where $c$ is the constant given by
\cref{thm:approximants:colouring}, which we consider in the
rest of this proof.

Since $p'_0m'_0$ does not simulate $p_0m_0$, we can assume a winning strategy for \R\ in the
\sgame\ which is optimal in the sense that it guarantees that
the simulation level $\alpha_i$ -- the largest ordinal with
$p_im_i\SIM{}{\alpha_i}p'_im'_i$ -- strictly decreases
along rounds of any play.
By monotonicity (\cref{lem:wsim:appr:props},
part~\ref{lem:wsim:appr:props:monotonicity}),
we can thus infer that whenever a pair of control states repeats along a play, then \V's
counter must have decreased or \R's counter must have increased:
Along any partial play
\begin{equation}
(p_0m_0,p'_0m'_0)(t_0,t'_0)(p_1m_1,p'_1m'_1)(t_t,t'_1)\dots(p_km_k,p'_km'_k)
\end{equation}
with
$(p_i,p'_i)=(p_k,p'_k)$ for some $i<k$, we have $m_i<m_k$ or
$m'_i > m'_k$.
By a similar argument we can assume that \V\ also plays
optimally, in the sense that he uses $\omega$-steps to increase
his counter to higher values than in previous situations with the same
pair of control states. By combining this with the previously stated
property that the sequence of $\alpha_i$ strictly decreases we obtain
the following:
\begin{equation}
    \label{eq:approximants:thm2:loop}
    \text{If } (p_i,p'_i) = (p_k, p'_k) \text{ and }
    t'_{i-1},t'_{k-1}\in\delta_\omega' \text{ then } m_i < m_k.
\end{equation}
Here $\delta_\omega'$ denotes the set of transitions with symbolic effect $\omega$ in \V's
net $\NN{N}'$.

Although \V\ loses the \sgame\ between $p_0m_0$ and $p'_0m'_0$, 
our assumption $p_0m_0\SIM{\beta}{}p'_0m'_0$
with $\beta = |Q\x Q'|(c +1)$ means that \V\ can ensure that
no play with fewer than $\beta$ $\omega$-steps is losing for him,
regardless of \R's strategy.
So we can safely assume that there is a play in \R's supposed
optimal winning strategy
along which \V\ makes use of
$\omega$-steps $\beta$ times.
Let $\pi = (p_0m_0,p'_0m'_0)(t_0,t'_0)(p_1m_1,p'_1m'_1)(t_t,t'_1)$ $\dots(p_km_k,p'_km'_k)$
be such a play.

Our choice of $\beta = |Q\x Q'|(c +1)$ guarantees that some pair
$(p,p')$ of control states repeats at least $c+1$
times directly after \V\ making an $\omega$-step.
Thus there are indices $i(1)<i(2)<\dots <i(c+1)<k$ such that~for all
$1\le j\le c+1$ we have $(p_{i(j)},p'_{i(j)})=(p,p')$ and
$t'_{i(j)}\in\delta_\omega$.
By observation~\eqref{eq:approximants:thm2:loop} and $m_0 \ge 0$ we
obtain that $m_{i(x)} \ge x$ for all $x$ with $0 \le x \le c+1$.
In particular, $c \le m_{i(c)}<m_{i(c+1)}$, that is, both of \R's counter values after the last two such repetitions must lie above $c$.
This allows us to apply \cref{thm:approximants:colouring} to derive a contradiction.

Let $\alpha$ be the simulation level before this repetition: $\alpha$ is the
largest ordinal that satisfies $pm_{i(c)}\SIM{}{\alpha}p'm'_{i(c)}$.
Since $m_{i(c+1)} > m_{i(c)} > c$,
\cref{thm:approximants:colouring} ensures the existence of a natural $n'$
such that~$pm_{i(c +1)}\SIM{}{\alpha}p'n'$.
Because \V\ used an $\omega$-step in his last response leading to the repetition of
states there must be a partial play $\pi'$ in which both players make the same moves as
in $\pi$ except that \V\ chooses $m'_{i(c +1)}$ to be $n'$.
Now in this play we observe that the simulation level did in fact not strictly decrease as this
last repetition of control states shows: We have
$pm_{i(c)}\SIM{}{\alpha}p'm'_{i(c)}\not\nsucceq_{\alpha+1} pm_{i(c)}$ and
$pm_{i(c +1)}\SIM{}{\alpha}p'm'_{i(c +1)}$, which contradicts the
assumed optimality of \R's strategy.
\end{proof}

To conclude this section on approximants, we show that
ordinary \textWSIM\ approximants $\WSIM{}{\alpha}$ indeed converge at level $\omega^2$
for any pair of \OCNs.
For this, let us observe a property of the nets constructed in the reduction
\cref{thm:wsim:reduction}.

    \begin{lemma}
        \label{rem:translation-approximants}
    Let $\NN{N},\NN{N}'$ be two \textOCNs\ and 
    $\NN{M},\NN{M}'$ the pair of \OCN\ and $\omega$-net constructed in the
    proofs of \cref{lem:app_reduction:GON,lem:app_reduction:normalize}.  Then,
    \begin{equation}
      \mbox{If }qn\WSIM{}{\alpha} q'n'\mbox{ w.r.t. } \NN{N},\NN{N'}
      \mbox{ then }qn\SIM{}{\alpha}q'n'\mbox{ w.r.t. } \NN{M},\NN{M'}.
      \label{thm:wsim:reduction:ordinals}
    \end{equation}

    \todo[inline]{PT: this is the second claim of \cref{thm:wsim:reduction}}
    \end{lemma}
    \begin{proof}
     It suffices to observe that
    the construction of $\NN{M},\NN{M}'$, presented in Lemma \ref{lem:app_reduction:normalize}, is such that
    one round of a \sgame\ w.r.t.~$\NN{N}$ and the guarded $\omega$-net $\NN{G}'$ is simulated by $k$ rounds of a
    \sgame\ w.r.t.\ $\NN{M},\NN{M}'$.
    On the other hand the construction of $\NN{N},\NN{G}'$, presented in Lemma \ref{lem:app_reduction:GON}, guaranties that 
    one round of a \wsgame\ w.r.t.~$\NN{N},\NN{N}'$ is simulated by $1$ round of a
    \sgame\ w.r.t.\ $\NN{N},\NN{G}'$.
    Thus, if \R\ has a strategy to win the \sgame\ relative to $\NN{M},\NN{M}'$ in $\alpha$
    rounds then she can derive strategies to win the games relative to $\NN{N},\NN{G}'$
    and to $\NN{N},\NN{N'}$ in no more than $\alpha$ rounds.
    \end{proof}

\begin{lemma}\label{lem:wsim:approximants:alpha-beta-connection}
For relations between a \OCN\ and an $\omega$-net,
we have $\SIM{}{\omega i}\;\subseteq\;\SIM{i}{}$ for every $i \in \N$.
\end{lemma}
\begin{proof}
    By induction on $i$.
The base case of $i=0$ is trivial, since $\SIM{0}{}$ is the full
relation.
We prove the inductive step by assuming the contrary and deriving a
contradiction.
Let $pm\SIM{}{\omega i} p'm'$ and
$pm \notSIM{i}{} p'm'$ for some $i>0$.
Then there exists some ordinal $\alpha$ such that $pm \notSIM{i}{\alpha} p'm'$.
Without restriction let $\alpha$ be the least ordinal
satisfying this condition.
If $\alpha \le \omega i$ then we trivially have a
contradiction. Now we consider the case $\alpha > \omega i$.
By $pm \notSIM{i}{\alpha} p'm'$ and \cref{lem:game_interpretation},
\R\ has a winning strategy in the \agame\ from position $(pm,p'm',\alpha,i)$.
Without restriction we assume that \R\ plays optimally, i.e., wins as quickly as
possible.
Thus this game must reach some position
$(qn,q'n',\alpha'+1,i)$ where $\alpha' \ge \omega i$ is a limit ordinal,
such that \R\ can win from $(qn,q'n',\alpha'+1,i)$ but not from
$(qn,q'n',\alpha',i)$.
I.e., $qn \notSIM{i}{\alpha'+1} q'n'$, but $qn \SIM{i}{\alpha'} q'n'$.
Consider \R's move $qn\step{a} rl$ according to her optimal winning
strategy in the game from position $(qn,q'n',\alpha'+1,i)$.
Since $qn\SIM{i}{\alpha'} q'n'$ and $\alpha'$ is a limit ordinal,
for every ordinal
$\gamma_k < \alpha'$, \V\ must have some countermove
$q'n'\step{a} r'_kl'_k$ such that~$rl \SIM{j}{\gamma_k} r'_kl'_k$,
where $j=i-1$ if the move was via an $\omega$-step and $j=i$
otherwise. In particular, $\sup_k \{\gamma_k\} = \alpha'$.
However, since \R's move $qn\step{a} rl$ was according to an optimal winning
strategy from position $(qn,q'n',\alpha'+1,i)$, we have
that $rl\notSIM{j}{\alpha'} r'_kl'_k$.
Therefore, there must be infinitely many different responses
$q'n'\step{a} r'_kl'_k$. Infinitely many of these countermoves must be via
 $\omega$-steps, because apart from these the system is finitely
branching.
Thus for every ordinal $\gamma < \alpha'$ there is some
\V\ countermove $q'n'\step{a} r'_kl'_k$ which is via an $\omega$-step
such that~$rl \SIM{i-1}{\gamma_k} r'_kl'_k$ where $\gamma_k \ge \gamma$ (note the
$i-1$ index due to the $\omega$-step).
In particular, we can choose $\gamma=\omega (i-1)$, because $i>0$ and
$\alpha'\ge \omega i$.
Then we have $rl\SIM{i-1}{\omega (i-1)} r'_kl'_k$, but
$rl \notSIM{i-1}{\alpha'} r'_kl'_k$.
However, from $rl \SIM{i-1}{\omega (i-1)} r'_kl'_k$ and the induction hypothesis,
we obtain $rl \SIM{i-1}{} r'_kl'_k$
and in particular
$rl \SIM{i-1}{\alpha'} r'_kl'_k$. Contradiction.
\end{proof}

\begin{theorem}\label{thm:wsim:approximants:omegasquare}
Weak simulation approximants on \OCN\ converge at level $\omega^2$,
but not earlier in general.
\end{theorem}
\begin{proof}
First we show that $\WSIM{}{\omega^2}$ is contained in $\WSIM{}{}$ for \OCN.
Let $pm$ and $p'm'$ be processes of \OCN\ $\NN{N}$ and $\NN{N'}$, respectively,
and let $\NN{M}$ and $\NN{M'}$ be the derived \OCN\ and $\omega$-net from
\cref{thm:wsim:reduction}
(page~\pageref{thm:wsim:reduction}).
Assume $pm \WSIM{}{\omega^2} p'm'$ w.r.t.~$\NN{N},\NN{N'}$. 
By \cref{rem:translation-approximants}
\todo{change}
we conclude that $pm \SIM{}{\omega^2} p'm'$ w.r.t.~$\NN{M},\NN{M'}$.
In particular we have $pm \SIM{}{\omega\cdot\CB} p'm'$ w.r.t.~$\NN{M},\NN{M'}$,
for the level $\CB \in \N$ from \cref{thm:wsim:approximants:convergence}.
From \cref{lem:wsim:approximants:alpha-beta-connection} we obtain
$pm \SIM{\CB}{} p'm'$ w.r.t.~$\NN{M},\NN{M'}$.
\Cref{thm:wsim:approximants:convergence} then yields 
$pm \SIM{}{} p'm'$ w.r.t.~$\NN{M},\NN{M'}$.
Finally, by \cref{thm:wsim:reduction}, we obtain that
$pm \WSIM{}{} p'm'$ w.r.t.~$\NN{N},\NN{N'}$.  

To see that $\omega^2$ is needed in general,
consider the following class of examples,
that are the result of
extending the net from \fullref{ex:weaksim:nonconvergence}.
%
%
Let $\NN{N}$ be the simple \OCN\ that consists only of
the self-loop $A \step{a,0} A$.
For every $i\le k \in \N$
the \OCN\ $\NN{N}'_k$ has
transitions
$(C_i, a, -1, C_i)$, 
$(B_{i}, \tau, 0, C_{i})$
$(B_{i}, \tau, +1, B_{i})$,
and
$(C_{i+1}, \tau, 0, B_{i})$   (see \cref{fig:pic_weak_ladder} below for $k=3$).
We see that $A \WSIM{}{\omega\cdot k} B_k0$,
but $A \notWSIM{}{} B_k0$ w.r.t.~$\NN{N},\NN{N}'_k$.
So, for every $k\in\N$
there are \OCNs\ for which
$\WSIM{}{\omega\cdot k}\neq\; \WSIM{}{}$.
\end{proof}

\begin{figure}[h]
\begin{center}
    \includegraphics[scale=0.7]{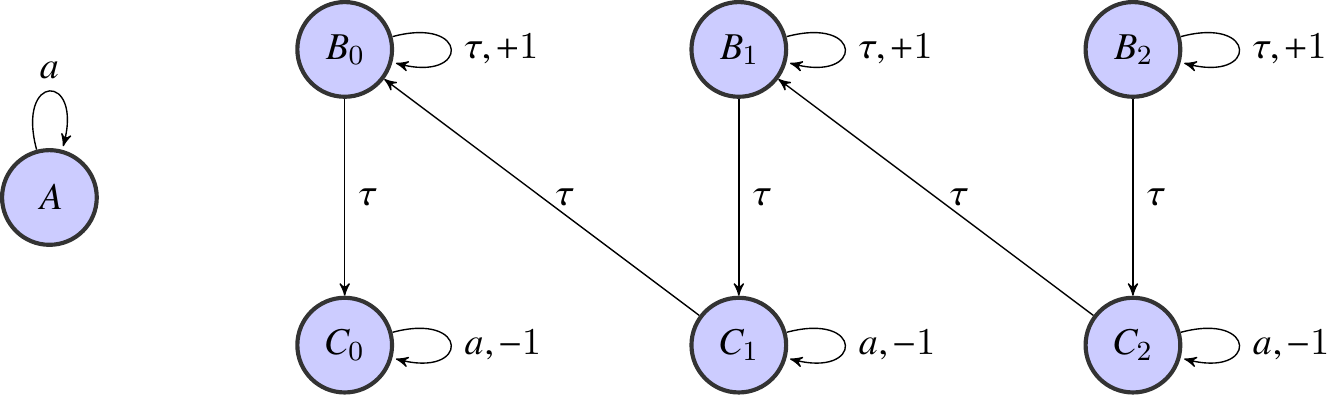}
\end{center}
\caption{The nets $\NN{N}$ and $\NN{N}'_3$ as constructed in
    \cref{thm:wsim:approximants:omegasquare}. We see that
    $An
    \;{{\protect\raisebox{+0.3ex}{$\preceq$}
            \protect\kern -.7em
            \protect\raisebox{-0.5ex}{\protect\rule{.59em}{.4pt}}
    }}{}_{\omega\cdot k}\;
    B_k
    $
    but also
    $An
    \ {{\protect\raisebox{+0.3ex}{$\not\preceq$}
            \protect\kern -.7em
            \protect\raisebox{-0.5ex}{\protect\rule{.59em}{.4pt}}
        }}{}_{\omega\cdot k+1}\
    B_k$.
}
\label{fig:pic_weak_ladder}
\end{figure}
    \todo{PT:caption new}

\subsection{Characterizing Weak Simulation Preorder}
  \label{sec:wsim-construction}
In order to show the decidability of simulation between \textOCNs\ and
$\omega$-nets we prove a stronger claim,
namely that the largest simulation relation is a semilinear set
and one can effectively compute its description.
To prove this claim for a fixed pair of nets,
we consider approximants $\SIM{k}{}$
and show (by repeated reduction to strong simulation over \OCN\
and using \cref{thm:ssim-pspace})
that in fact $\SIM{k}{}$ is effectively semilinear for every
level $k\in\N$.
To be precise, we show the following lemma.

\begin{lemma}\label{lem:main:approximants}
    For any \textOCN\ $\NN{M}$ 
    and $\omega$-net $\NN{M'}$ 
    with sets of control states $Q$ and $Q'$ respectively,
    there is an effectively computable sequence
    $(\NN{S}_k,\NN{S}'_k)_{k\in\N}$ of pairs of \OCN\
    with sets of control states $S_k\supseteq Q$ and $S_k\supseteq Q'$ respectively,
    such that
    for all $k,m,m'\in\N$ and states $p\in Q, p'\in Q'$,
    \begin{equation}
        pm\SIM{k}{}p'm'\text{ w.r.t. } \NN{M},\NN{M'} \iff
        pm\SIM{}{}p'm'\text{ w.r.t. } \NN{S}_k,\NN{S}_k'.  
    \end{equation}
\end{lemma}

\noindent A direct consequence of this is the effective semilinearity,
and thus decidability, of \textWSIM\ $\wsim$ over any fixed
pair of \textOCNs.
\begin{theorem}
    Let $\NN{N},\NN{N'}$ be two \textOCNs.
    The largest \textWSIM\ relation $\wsim$ with respect to $\NN{N},\NN{N'}$
    is a semilinear set and its representation is effectively computable.
\end{theorem}
\begin{proof}
    By \cref{thm:wsim:reduction}, it suffices to show the claim for the
    largest strong simulation
    $\ssim$ between a \OCN\ $\NN{M}$ and an $\omega$-net $\NN{M'}$.
    By \cref{lem:main:approximants}, one can iteratively
    compute the sequence $(\NN{S}_k,\NN{S}_k')_{k\in\N}$
    of nets characterizing $\SIM{k}{}$ for growing $k$.
    Because $\NN{S}_k$ and $\NN{S}_k'$ are \textOCNs,
    we can apply \cref{thm:ssim-pspace} and derive that strong simulation
    w.r.t.~$\NN{S}_k,\NN{S'}_k$, and hence the approximant
    $\SIM{k}{}$ w.r.t.~$\NN{M},\NN{M'}$ are effectively semilinear sets.
    Recall that for $k\in\N$, $\SIM{k+1}{}\:\subseteq\:\SIM{k}{}$.
    Because equality of semilinear sets is decidable,
    we can check after each iteration if $\SIM{k+1}{}\: \supseteq\:\SIM{k}{}$
    holds, in which case we stop with the description of $\SIM{k}{}\:=\:\SIM{}{}$.
    Termination of this procedure is guaranteed by
    \cref{thm:wsim:approximants:convergence}.
\end{proof}

Before we prove \cref{lem:main:approximants}, we introduce
two important ingredients for the construction
of the nets $\NN{S}_k,\NN{S}_k'$.
The first is a class of simple gadgets called
\emph{test chains} that will form part of these nets
and allow us to check, by means of
a continued \sgame, if the counter value of \R\ is $\ge i$
for some hard-wired constant $i\in\N$.
A \emph{test chain} for $i\in\N$,
is a pair $\NN{T}_i,\NN{T}_i'$
of \OCNs\ with initial states $t_i$ and $t'_i$
over actions $\Act=\{e,f\}$.
We let $t_i$ be the starting point of a counter-decreasing
chain of $e$-steps of length $i$ where the last state of the chain
can make an $f$-step, whereas $t'_i$ is a simple $e$-loop
(see \cref{fig:test-chains}).
Then we observe that for all $m,n\in\N$,
\begin{equation}
    t_im\notSIM{}{} t'_in \iff m\ge i \label{eq:main:gadget}.
\end{equation}

\begin{figure}[h]
  \centering
  \includegraphics{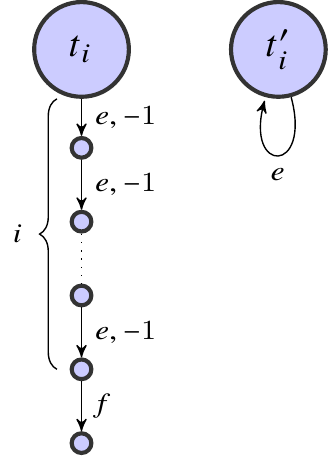}
  \caption{A test chain for value $i\in\N$.}
  \label{fig:test-chains}
\end{figure}

\noindent The \emph{test chain} for $\omega$ is the pair $\NN{T}_\omega,\NN{T}_\omega'$
of nets, consisting of simple $e$-loops
$t_{\omega}\step{e}t_{\omega}$ and
$t_{\omega}'\step{e}t_{\omega}'$, respectively.
Trivially, for all $m,n\in\N$ it holds that
\begin{equation}
    t_\omega m\SIM{}{} t'_\omega n
\end{equation}

  
\noindent The second ingredient for our construction is the 
notion of \emph{minimal sufficient values}.
    Consider the approximant $\SIM{k}{}$ for some parameter $k$,
    and let $(q,q')\in (Q\x Q')$ be a pair of states.
    By monotonicity (\cref{lem:wsim:appr:props}, point~\ref{lem:wsim:appr:props:monotonicity}),
    there is a minimal value $\suff{q,q',k}\in\N\cup\{\omega\}$ satisfying
    \begin{equation}\label{sufficient_values}
      \forall n'\in\N.\ q(\suff{q,q',k}) \notSIM{k}{} q'n'.
    \end{equation}
    Let $\suff{q,q',k}$ be $\omega$ if no finite value satisfies this condition.
The following properties are immediate from the definitions.
\begin{lemma}\label{lem:suff}\
For all $q\in Q, q'\in Q'$ and $k\in\N$,
\begin{enumerate}
  \item \label{lem:suff:omega} $\suff{q,q',0}=\omega$, and
  \item \label{lem:suff:decrease} $\suff{q,q',k}\geq \suff{q,q',k+1}$.
\end{enumerate}
\end{lemma}

We are now ready to present the construction of
successive pairs of nets $\NN{S}_k,\NN{S}_k'$,
that satisfy the claim of \cref{lem:main:approximants}.
The idea behind the construction of nets for parameter $k+1$ is as follows.
Assuming we have already constructed a semilinear representation of $\SIM{k}{}$
in the form of two \OCN\ $\NN{S}_k$ and $\NN{S}_k'$,
we can compute
the values $\suff{q,q',k}$ for every pair $(q,q')$.

The nets $\NN{S}_{k+1}$ and $\NN{S}_{k+1}'$ are constructed so that
a \sgame\ played on nets $\NN{S}_{k+1},\NN{S}_{k+1}'$ mimics
the \agame\ played on $\NN{M},\NN{M}'$ with $\omega$-parameter $(k+1)$
until \V\ responses via an $\omega$-step,
leading to some game position $qn$ vs.\ $q'n'$.
Afterwards, the \agame\ would continue with the next lower parameter $k$.
In the \sgame\ on $\NN{S}_{k+1}$ and $\NN{S}_{k+1}'$,
\V\ cannot make the $\omega$-step but can instead enforce the play to continue in
some subgame (a test chain) that he wins iff \R's counter is smaller than the hard-wired value $\suff{q,q',k}$.
This ``forcing'' of the play can be implemented for \OCN\ simulation using a standard technique
called \emph{defender's forcing} (see e.g.~\cite{KJ2006}),
that essentially allows \V\ to reach a universal process (and thus win) in the
next round unless his opponent moves in some specific way.

The nets $\NN{S}_{k+1}$ and $\NN{S}_{k+1}'$ thus consist of the original nets $\NN{M},\NN{M}'$ where
all $\omega$-transitions in \V's net $\NN{M}'$ are replaced by a small constant
defenders-forcing script, leading to the corresponding testing gadget.
The only difference between two such pairs of nets for different parameters $k$
is the lengths of the test chains.

\begin{definition}[The construction of $\NN{S}_k$ and $\NN{S}_k'$]
    \label{def:SiS'i}
Fix a \OCN\ $\NN{M}=(Q,\Act,\delta)$, an $\omega$-net $\NN{M}'=(Q',\Act,\delta)$
and a constant $k\ge 1$.
We construct the \textOCNs\ $\NN{S}_k$ and $\NN{S}_k'$ that characterize
the approximant $\SIM{k}{}$.

For any pair $(p,p')\in Q\x Q'$ of states, let
$\NN{T}_{p,p'}$ and $\NN{T'}_{p,p'}$ be
the nets that describe the test chain for $\suff{p,p',k-1}$.
Let $\NN{T}_{p,p'}=(T_{p,p'},\{e,f\},\delta_{p,p'})$ and
$\NN{T}_{p,p'}'=(T_{p,p'}',\{e,f\},\delta'_{p,p'})$ and
let $t_{p,p'}$ and $t'_{p,p'}$ be the initial states of $\NN{T}_{p,p'}$ and $\NN{T}_{p,p'}'$
respectively.
W.l.o.g.\ we can assume that $e,f\notin \Act$ are new letters.
We define the \textOCNs\ $\NN{S}_k$ and
$\NN{S}_k'$ over the new alphabet $\ol{\Act}$ as follows.
$\ol{\Act}$ 
contains all letters of the original alphabet,
two (new) actions $e,f$ used in test gadgets
and a new action $(p,p')$ for every pair of original states.
\begin{equation}
    \ol{\Act} = \Act\cup \{f,e\}\cup (Q\x Q').
\end{equation}
The net $\NN{S}_k=(S_k,\ol{\Act},\delta_k)$ has all original states
of $\NN{M}$, plus those of all test chains:
\begin{equation}
    S_k = Q \cup \bigcup_{p\in Q, p'\in Q'} T_{p,p'}
\end{equation}
Its transitions $\delta_k \supseteq \delta \cup \bigcup_{q\in Q, q'\in Q'} \delta_{q,q'}$
are those of $\NN{M}$, all test chains, and the following for all $q\in Q, q'\in Q'$:
\begin{equation}
    q\step{(q,q'),0}t_{q,q'} \label{eq:sisi:spoilmove}
\end{equation}
The net $\NN{S}_k'=(S'_k,\ol{\Act},\delta'_k)$ has states
\begin{equation}
    S'_k = Q'
    \cup(\bigcup_{q\in Q, q'\in Q'} T'_{q,q'})
    \cup \{W\}.
\end{equation}
So it contains all original states of $\NN{M}'$, those of all test chains
and a new ``win'' state $W$.
Its set of transitions is
$\delta'_k \supseteq \{q\step{a,x}q' \in \delta'\ |\ x \neq \omega\} \cup \bigcup_{q\in Q, q'\in Q'} \delta_{q,q'}$.
It contains those transitions in $\NN{M}'$
which are not labeled by $\omega$, the transitions of the test chains
plus the following, that allow \V\ to force the game into a test chain:
\begin{align}
    p'\step{a,0}t'_{p,q'}\quad
      &\text{for all $p\in Q$ and $q',p'\in Q'$ if $p'\step{a,\omega}q' \in \delta'$},
      \label{eq:sisi:omegamove} \\
    p'\step{(q,q'),0} W\quad
      &\text{for all $q\in Q$ and $q',p'\in Q'$},
      \label{eq:sisi:punish_wrong_spoilmove}\\
    t'_{q,q'}\step{(q,q'),0}t'_{q,q'}\quad
      &\text{for all $q\in Q$ and $q'\in Q'$},
      & \label{eq:stay}\\
    t'_{q,q'}\step{(q,p'),0}W\quad
      &\text{for all $q\in Q$ and $q',p'\in Q'$ if $p'\neq q'$},
      \label{eq:sisi:punish1}\\
    t'_{q,q'}\step{a,0}W\quad
      &\text{for all $q\in Q$ and $q'\in Q'$ and $a \in \Act$},
      \label{eq:sisi:punish2}\\
    W\step{a,0}W\quad
      &\text{for all $a \in \Act'$}.
\end{align}
\Cref{fig:weak_forcing} illustrates the forcing mechanism
due to these new transitions.
\end{definition}
\begin{figure}[h]
    \centering
    \includegraphics{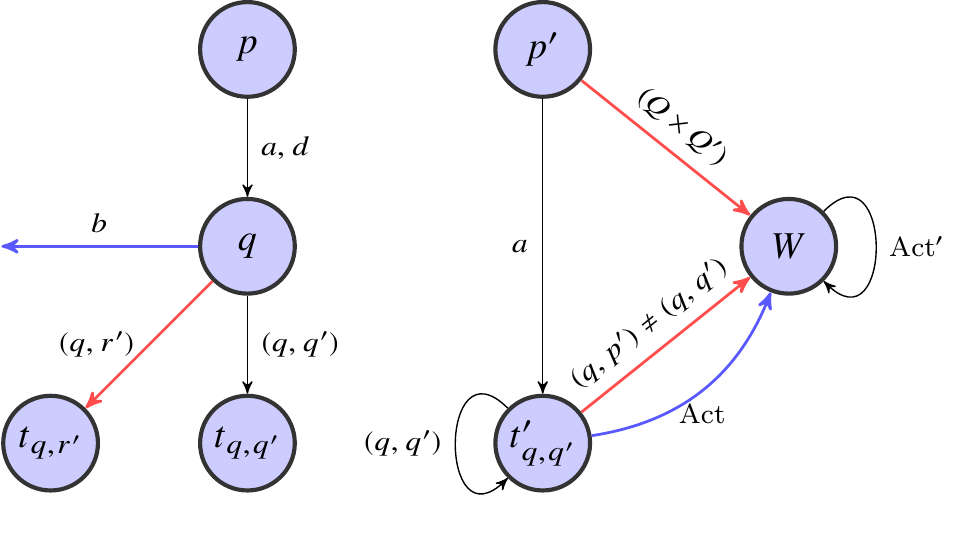}
    \caption{The forcing mechanism
        that replaces a \V\ transition $p'\step{a,\omega}q'$.
        Counter effects are omitted, individual transitions are grouped
        and punishing moves are colored.
        For instance, the red arrow from $p'$ to $W$ depicts all transitions due
        to \cref{eq:sisi:punish_wrong_spoilmove} that prevent
        \R\ from using any actions of the form $(p,p')\in\Act$ unless \V\ already moved to
        some state $t'_{q,q'}$.
        Note that \R\ must prevent \V\ from reaching the universal state $W$
        and that once the players are at states $t_{q,q'}$ and $t'_{q,q'}$,
        she has no other option but to play
        the test chain that starts here.
    }
    \label{fig:weak_forcing}
\end{figure}

\noindent Observe that the definition of the nets $\NN{S}_k,\NN{S}_k'$ above
depends on the sufficient values $\suff{p,p',k-1}$ for all original control states
$p,p'$.
It is therefore crucial to know these values for this construction to be effective.
The following two lemmas state the correctness of the construction.

\begin{lemma}\label{lem:main:appr:P1}
 For all control states $(p,p') \in Q \x Q'$ and naturals $k,m,n \in \N$:
 \begin{equation}
     pm\notSIM{}{}p'm'\text{ w.r.t.~$\NN{S}_k,\NN{S}_k'$}
     \impliedby
     pm\notSIM{k}{}p'm'\text{ w.r.t.~$\NN{M},\NN{M'}$}
 \end{equation}
\end{lemma}

\begin{proof}
  Note that by definition of approximants,
  $pm\notSIM{k}{}p'm'$ implies $pm\notSIM{k}{\alpha}p'm'$ for some ordinal $\alpha$.
  By the game interpretation (\cref{lem:game_interpretation})
  it is thus sufficient to show that
  for all ordinals $\alpha$,
  if \R\ has a winning strategy in the \agame\
  from position $(pm, p'm', \alpha, k)$ then she also
  has a winning strategy in the \sgame\ between $\NN{S}_k,\NN{S}_k'$ from
  position $(pm, p'm')$.
  
  We proceed by ordinal induction on $\alpha$.
  The base case trivially holds since \R\ loses from a position $(pm,p'm',0,k)$
  by definition of the \agame\ (\cref{def:approxgame}).
  
  For the induction step let \R\ play the same move $pm \step{a} qn$
  for some $a \in \Act$ in both games according to her assumed winning strategy in the
  \agame.
  Now \V\ makes his response move in the new game between $\NN{S}_k,\NN{S}_k'$,
  which yields two cases.
  In the first case, \V\ does not use a step induced by a transition from
  \cref{eq:sisi:omegamove}. Then
  his move induces a corresponding move in the \agame\ which
  leads to a new configuration $(qn, q'n', \gamma, k)$ where
  $qn\notSIM{k}{\gamma}q'n'$ for some ordinal $\gamma<\alpha$.
  By the induction hypothesis, \R\ now has a winning
  strategy to continue the \sgame\ from position $(qn,q'n')$.
  
  In the second case, \V's response is via a step induced by a transition from
  \cref{eq:sisi:omegamove}, which leads to a new configuration
  $(qn, t'_{r,q'}n')$ for some $r\in Q$.
  Thus in the \agame\ there will exist \V\ moves to
  positions $(qn, q'n', \gamma, k-1)$ where $\gamma<\alpha$ and
  $n' \in \N$ can be arbitrarily high.
  We can safely assume that \V\ chooses $r=q$, since otherwise \R\ can afterwards
  win in one round by a $(q,q')$ labeled step from $qn$.
  Now in the next round \R\ can play $qn\step{(q,q')}t_{q,q'}n$
  by \cref{eq:sisi:spoilmove} and \V's only option
  is to stay in his current state by \cref{eq:stay}.
  The \sgame\ thus continues from
  $(t_{q,q'}n, t'_{q,q'}n')$, which is the beginning of the testing gadget
  for states $q,q'$.
  To show that \R\ wins the rest of the \sgame, we show that
  indeed, $n$ must be at least be $\suff{k-1,q,q'}$.
  By our initial assumption,
  \R\ wins the \agame\ from the position $(pm,p'm', \alpha, k)$.
  Thus there is some ordinal $\gamma<\alpha$ such that
  \R\ also wins the \agame\ from position $(qn, q'n', \gamma, k-1)$
  for every $n'\in \N$.
  Thus, by \cref{lem:game_interpretation} and \cref{def:approximants},
  we have $qn \notSIM{k-1}{\gamma} q'n'$ and by
  \cref{lem:wsim:appr:props} (item~\ref{lem:wsim:appr:props:alpha-inc})
  $qn \notSIM{k-1}{} q'n'$ for all $n' \in \N$.
  By the definition 
  of sufficient values, we obtain $n\geq \suff{q,q',k-1}$.
  By the construction of the gadgets and \cref{eq:main:gadget}
  we get $t_{q,q'}n\notSIM{}{} t'_{q,q'}n'$,
  which concludes our proof.
\end{proof}
\begin{lemma}\label{lem:main:appr:P2}
 For all control states $(p,p') \in Q \x Q'$ and naturals $k,m,n \in \N$:
 \begin{equation}
     pm\notSIM{}{}p'm'\text{ w.r.t.~$\NN{S}_k,\NN{S}_k'$}
     \implies
     pm\notSIM{k}{}p'm'\text{ w.r.t.~$\NN{M},\NN{M}'$}
 \end{equation}
\end{lemma}
\begin{proof}
    Assume $pm \notSIM{}{} p'm'$ w.r.t.~$\NN{S}_k$ and $\NN{S}_k'$.
    Since both $\NN{S}_k,\NN{S}_k'$ are just \textOCNs,
    non-simulation manifests itself at some finite approximant
    $\alpha \in \N$, i.e., $pm\notSIM{}{\alpha} p'm'$.
    By definition of $\SIM{k}{}$ it suffices to show that
    some ordinal $\gamma$ exists such that 
    $pm \notSIM{k}{\gamma} p'm'$ w.r.t.~$\NN{M},\NN{M'}$.
    By the game characterization of approximants (\cref{lem:game_interpretation})
    this amounts to showing a winning strategy for \R\
    in the \agame\ from position $(pm, p'm', \gamma, k)$.
    
    We proceed by induction on $\alpha$.
    The claim is trivial for the base case $\alpha=0$.
    For the induction step we consider a
    move $pm \step{a} qn$ for some $a \in \Act$ by \R\ in both games
    according to \R's assumed winning strategy in the \sgame\ between
    $\NN{S}_k,\NN{S}_k'$.
    It cannot be a \R\ move
    $p\step{(p,q'),0}t_{p,q'}$ by \cref{eq:sisi:spoilmove}, because
    this would allow \V\ to reply by moving to the universal state $W$
    by \cref{eq:sisi:punish_wrong_spoilmove}.
    Now we consider all (possibly infinitely many) replies by \V\ in the
    \agame\ between $\NN{M},\NN{M}'$ from a position $(pm,p'm',\gamma, k)$
    for some yet to be determined ordinal $\gamma$.
    These replies fall into two classes.

    In the first class, \V's move $p'm' \step{a} q'n'$ is \emph{not} via
    an $\omega$-step and thus also a possible move in the
    \sgame\ between $\NN{S}_k,\NN{S}_k'$.
    From our assumption that \R\ wins the
    \sgame\ from position $(pm, p'm')$ in at
    most $\alpha\in\N$ steps, it follows that \R\ wins
    from $(qn,q'n')$ in at most $\alpha-1$ steps.
    By induction hypothesis, there is an ordinal $\beta$
    such that\ \R\ has a winning strategy in the
    \agame\ for $\SIM{k}{\beta}$ between $\NN{M},\NN{M}'$
    from position $(qn, q'n')$.
    There are only finitely many such replies.
    Let $\gamma^0$ be the maximal such $\beta$.

    In the second class, \V's move $p'm' \step{a} q'n'$ uses an
    $\omega$-step, which does not exist in $\NN{S}_k'$.
    Instead, \V\ can move
    $p'm' \step{a,0}t'_{r,q'}m'$ by a step induced by a transition
 due to \cref{eq:sisi:omegamove}.
    From our assumption that \R\ wins the
    \sgame\ from position $(pm, p'm')$ in at
    most $\alpha\in\N$ steps, it follows that \R\ wins from
    $(qn,t'_{r,q'}m')$ in at most $\alpha-1$ steps.
    If $r \neq q$ then this is trivially true by a move due to
    \cref{eq:sisi:spoilmove}. Otherwise, if $r=q$,
    then this can only be achieved by a \R\ move
    $qn \step{(q,q'),0} t_{q,q'}n$ in the next round,
    because for any other \R\ move \V\ has a winning countermove
    by Equations~\eqref{eq:sisi:punish1} or \eqref{eq:sisi:punish2}.
    In this case \V\ can only reply with a move
    $t'_{q,q'}m' \step{(q,q'),0}t'_{q,q'}m'$ due to \cref{eq:stay},
    and we must have that \R\ can win in at most $\alpha-2$ steps
    from position $(t_{q,q'}n, t'_{q,q'}m')$,
    which is the beginning of the testing gadget for states $(q,q')$.
    By construction of $\NN{S}_k,\NN{S}_k'$, in particular by
    definition of the gadgets and \cref{eq:main:gadget},
    this implies that $n \ge \suff{q,q',k-1}$.
    By the definition 
    of sufficient values we obtain
    $\forall n' \in\N.\, qn\notSIM{k-1}{} q'n'$.
    Therefore, for every $n' \in \N$ there exists some ordinal
    $\beta$ such that\ $qn\notSIM{k-1}{\beta} q'n'$.
    Let $\gamma$ be the least ordinal greater or equal all those $\beta$.
    Each of the finitely many distinct $\omega$-transitions yields such
    a $\gamma$. Let $\gamma^1$ be the maximum of them.

    Finally, we set $\gamma := \max(\gamma^0,\gamma^1)+1$.
    Then every reply to \R's initial move $pm\step{a}qn$ in the \agame\
    from $(pm,p'm',\gamma,k)$ leads to a position that is winning for \R.
    It follows that \R\ has a winning strategy in the
    \agame\ from $(pm,p'm',\gamma,k)$.
\end{proof}

The proof of \cref{lem:main:approximants} is now a formality.

\begin{proof}[Proof of \cref{lem:main:approximants}]
    Let $\NN{M}= (Q, \Act, \delta)$ and $\NN{M}' = (Q',\Act,\delta')$.
    We iteratively construct nets $(\NN{S}_k,\NN{S'}_k)$
    that characterize $\SIM{k}{}$ for growing $k\in\N$.

    For the base case $k=0$, we observe that
    $\SIM{0}{}\:=\:Q\x \N \x Q' \x \N$ is the full relation.
    The claim therefore trivially holds for the
    pair $\NN{S}_0,\NN{S}_0'$ of nets that contain no transitions at all.
    Also, by \cref{lem:suff}, point \ref{lem:suff:omega}, the minimal sufficient value
    $\suff{q,q',0}$ equals $\omega$ for every pair of states $(q,q')\in Q\x Q'$.
    
    For the induction step, consider $k>0$.
    By assumption, we have already constructed the pair
    $(\NN{S}_{k-1},\NN{S}_{k-1}')$ of nets correctly characterizing $\SIM{k-1}{}$.
    By \cref{thm:ssim-pspace} (page~\pageref{thm:ssim-pspace}) we know that the
    simulation preorder 
    w.r.t.~$\NN{S}_{k-1},\NN{S}_{k-1}'$ is effectively semilinear.
    Since semilinear sets are effectively closed under projections
    and complements, we can compute the semilinear representation of
    the approximant $\SIM{k-1}{}$ and its complement
    and therefore also the values
    $\suff{q,q',k-1}$ for all $(q,q') \in Q\x Q'$.
    Knowing these values, we can
    construct the next pair
    $(\NN{S}_{k},\NN{S}_{k}')$ of nets according to \cref{def:SiS'i}.
    The correctness of this new pair follows from
    \cref{lem:main:appr:P1,lem:main:appr:P2}.
\end{proof}
Note that in the proof above, we construct a description of 
the previous approximants only to compute the values $\suff{p,p',k-1}$.
We will now show that these values are in fact polynomially bounded and
can also be computed in polynomial space.

\subsection{Complexity Analysis}
  \label{sec:wsim-complexity}

We show that the bounds on the coefficients of the Belt Theorem,
as derived in \cref{sec:belts/proof}, imply that the construction
shown in the previous section for checking \emph{weak} simulation
actually uses only polynomial space.

To obtain an upper bound for the complexity of this procedure, we will
show that the sizes of all nets $(\NN{S}_k,\NN{S}_k')$, as constructed in
\cref{def:SiS'i}, are polynomial in the sizes of $\NN{M}$ and $\NN{M}'$.
We start with some observations about the shape
of the nets $\NN{S}_k$ and $\NN{S}_k'$.

\begin{lemma}\label{lem:shape}\

\begin{enumerate}
  \item The net $\NN{S}_k'$ remains constant from index $k=1$ on.
      \label{lem:shape:constant-dup}
  \item \label{lem:shape:number}
      Every net $\NN{S}_k$ for $k>0$ contains precisely $|Q\x Q'|$
      many disjoint testing chains,
      one for each pair of states in $\NN{M}$ and $\NN{M}'$.
  \item
      \label{lem:shape:test-length}
      If $\suff{q,q',k-1}\neq\omega$, then
      the length of the test chain for states $q,q'$ in net $\NN{S}_k$
      is exactly $\suff{q,q',k-1}$.
      Otherwise, it is a simple $e$-labeled loop.
\end{enumerate}
\end{lemma}

\noindent Using these properties above and \cref{lem:suff},
point~\ref{lem:suff:decrease}, we derive that at some $k\in\N$,
the sequence $(\Net{S}_i,\Net{S}_i')_{i\in\N}$ of nets stabilizes
to $(\NN{S}_{k},\Net{S}_{k}') = (\NN{S}_{k},\Net{S}_{1}')$. This observation is actually an alternative proof of \cref{thm:wsim:approximants:convergence}.
 Indeed above claim holds because for any pair $(q,q')$ there can only be one index $i$
such that the respective sufficient value jumps from $\suff{q,q',i}=\omega$ to
$\suff{q,q',i+1}\in\N$.
Because these nets characterize
approximants $\SIM{k}{}$ and $\SIM{k+1}{}$ w.r.t.\ $\NN{M},\NN{M}'$
(by \cref{lem:main:appr:P1,lem:main:appr:P2})
we obtain that $\SIM{k}{}\;=\;\SIM{k+1}{}\;=\;\SIM{}{}$.

\begin{lemma}
    \label{lem:sisi-poly}
    Consider the sequence $(\NN{S}_k,\NN{S}_k')_{k\in\N}$ as constructed
    in \cref{def:SiS'i} for the \OCN\ $\NN{M}$ and $\omega$-net $\NN{M}'$.
    For any index $k\in\N$, the nets $\NN{S}_k,\NN{S}_k'$
    are of polynomial size, and can be constructed
    in polynomial space with respect to the sizes of
    the original nets $\NN{M}$ and $\NN{M}'$.
\end{lemma}
\begin{proof}
    For $k=0$, these nets are defined to be just copies of $\NN{M}$ and $\ND{M}$
    with no transitions. The claim is therefore trivial for $k=0$.
    For all higher indices $k+1$, we consider nets
    $\NS{S}_{k+1}$ and $\NN{S}_{k+1}'$ individually.

    By \cref{lem:shape}, point~\ref{lem:shape:constant-dup},
    $\NN{S}_{k+1}'$ is the same as $\NN{S}_1'$,
    which can easily be seen to be of polynomial size
    in the sizes of $\NS{M}$ and $\ND{M}$ (cf. \cref{def:SiS'i}).
    The net $\NS{S}_{k+1}$ is completely determined by
    the original pair of nets and the length of the test chains,
    which in turn are derived only from the minimal sufficient values
    $\suff{q,q',k}$ for level $k$.
    By construction, the size of the net $\NS{S}_{k+1}$ is polynomial (actually linear)
    in the sizes of $\NS{M},\ND{M}$ and the maximal length of a test chain
    in the net $\NN{S}_k$.
    By \cref{lem:shape}, point~\ref{lem:shape:test-length}, it is therefore
    enough show that one can compute the values $\suff{q,q',k}$ for all states
    $q\in Q$ and $q'\in Q'$ in polynomial space and bound them polynomially
    w.r.t.~$\NS{M},\ND{M}$ in case they are finite.

    Recall that $\suff{q,q',k}$ is defined in terms of the approximant
    $\SIM{k}{}$, which is characterized as the strong simulation $\SIM{}{}$
    relative to the nets $\NN{S}_k,\NN{S}_k'$ by \cref{lem:main:approximants}.

    Let $C_k$ be larger than the maximal length an acyclic path in the product
    of nets $\NS{S}_k$ and $\NN{S}_k'$.
%
    By \cref{thm:belt-theorem-bounds}, $C_k$ is sufficient
    for the claim of the Belt Theorem applied to the nets $\NN{S}_{k}$
    and $\NN{S}_{k}'$. In particular, by \cref{lem:compute-vbelts},
    it bounds the width of all vertical belts and therefore all
    finite values $\suff{q,q',k}$:
    \begin{equation}
        \suff{q,q',k}\in\N \implies \suff{q,q',k}\le C_k.
    \end{equation}

\noindent The form of the nets (\cref{lem:shape}, points~\ref{lem:shape:number},\ref{lem:shape:test-length})
    means that the longest acyclic path
    in the product of $\NS{S}_{k}$ and $\ND{S}_{k}$,
    must actually start within the part described by the original nets,
    and eventually go through one of the test chains.
    We can therefore bound $C_k$ by
    \begin{equation}
        \label{eq:Ckbound}
        C_k\le C_1 + C_{k-1}.
    \end{equation}

\noindent    We fix a pair $(q,q')$ of states and consider the length of the test chain
    for this pair in the net $\NS{S}_i$ for growing indices $i$.
    By \cref{lem:suff} and \cref{lem:shape}, point~\ref{lem:shape:test-length},
    we see that there can only be one index $i$ such that
    the length of the chain increases, namely 
    if $\suff{q,q',i}=\omega > \suff{q,q', i+1}\in\N$.
    Because there are always exactly $K=|Q\x Q'|$ many test chains,
    this means that there can be at most $K$ indices $i$ such that
    $C_{i+1} \ge C_i$.
    Together with \cref{eq:Ckbound} we can therefore globally bound every $C_k$ by
    \begin{equation}
      C_k\leq K \cdot C_1.
    \end{equation}
    
\noindent     We conclude that the sizes of all $\NN{S}_{k},\NN{S}_{k}'$ are polynomial in the
    sizes of $\NN{M}$ and $\NN{M}'$.
    By \cref{lem:compute-vbelts}, we can thus compute the exact
    values of $\suff{q,q',k}$ and construct $\NN{S}_{k+1},\NN{S}_{k+1}'$
    using polynomial space w.r.t.~$\NN{M}$ and $\NN{M}'$ as required.
\end{proof}
\begin{theorem}\label{thm:wsim-poly}
    For any pair $\NN{N},\NN{N}'$ of \textOCNs\ one can construct, 
    in polynomial space,
    two polynomially bigger \OCNs\ $\NN{S}$ and $\NN{S}'$ that contain the original
    states of $\NN{N}$ and $\NN{N}'$ respectively, such that
    weak simulation $\wsim$ w.r.t.\ $\NN{N},\NN{N}'$ is the projection
    of strong simulation w.r.t.\ $\NN{S},\NN{S}'$.
\end{theorem}
\begin{proof}
    The claim follows from
    \cref{thm:wsim:reduction,lem:main:approximants,lem:sisi-poly}.
    Indeed, due to \cref{thm:wsim:reduction} we can construct in polynomial time two $\omega$-nets $\NN{M},\NN{M'}$ such that
    weak simulation $\wsim$ w.r.t.\ $\NN{N},\NN{N}'$ is the projection
    of strong simulation w.r.t.\ $\NN{M},\NN{M'}$. 
    By \cref{lem:main:approximants}, there is a sequence of pairs of
    nets $(\NN{S}_k,\NN{S}_k')_{k \in \N}$, such that  
    for all $k,m,m'\in\N$ and states $p\in Q, p'\in Q'$,
    $
        pm\SIM{k}{}p'm'\text{ w.r.t. } \NN{M},\NN{M'}$  iff
        $pm\SIM{}{}p'm'\text{ w.r.t. } \NN{S}_k,\NN{S}_k'.  
    $
    Finally, by \cref{lem:sisi-poly} elements of this sequence can be constructed in polynomial space, and
    for some $l < k$ it must hold 
    $(\NN{S}_k, \NN{S}_k') = (\NN{S}_l, \NN{S}_l')$.
    Thus such a pair $(\NN{S}_k, \NN{S}_k')$ can be computed polynomial space; and
    \begin{equation*}
        pm\SIM{}{}p'm'\text{ w.r.t. } \NN{M},\NN{M'} \quad \text{ iff } \quad
        pm\SIM{}{}p'm'\text{ w.r.t. } \NN{S}_k,\NN{S}_k',
    \end{equation*}
    as required.
\end{proof}
The main result of this section is now a direct consequence
of \cref{thm:wsim-poly,thm:ssim-pspace}.
Recall that a \PSPACE\ lower bound already holds for strong simulation.
\begin{theorem}\label{thm:wsim-pspace}
  Checking weak simulation preorder between two \OCNs\ is \PSPACE-complete.
  Moreover, the largest weak simulation relation is semilinear
  and can be explicitly represented in space exponential in the sizes of the input nets.
\end{theorem}

\section{Conclusion}\label{sec:conclusion}
In this paper we showed that both strong and weak simulation
for \textOCNs\ are \PSPACE-complete.
A \PSPACE\ lower bound, as well as decidability of strong simulation were known
before \cite{AC1998,JM1999,Srb2009}.

Our first contribution is a new constructive proof of the \emph{Belt Theorem}
(see \cref{sec:strongSim}), based on a bounded abstraction of the simulation game.
A consequence of this construction is that
the simulation relation for fixed nets
is a semilinear relation of a very specific form that can be represented explicitly
in space exponential in the size of the input nets.
Due to the locality of the simulation condition, this representation can
be stepwise guessed and verified, which leads to a \PSPACE\ procedure
to check whether simulation holds between two given configurations.
The complexity of this procedure depends only on the size of the
input nets, not on the size of the given configurations.

Our second main contribution is an iterative reduction from \emph{weak} to strong
simulation over \textOCNs.
The main difficulty is to deal with unbounded branching (i.e., unrestricted
counter increases) of \V\ during a weak simulation game.
Our argument uses a suitable sequence of over-approximations, 
based on the number of times \V\ uses unbounded increases
during a play.
Using the results for the strong case, we show that this sequence necessarily
converges at a polynomially bounded level, and that each approximant relation
can in fact be represented as the maximal strong simulation over a pair of
polynomially enlarged \textOCNs.
This allows to conclude that our results for the strong case,
namely the effectiveness of an \EXPSPACE-representation as well as
a \PSPACE-decision procedure, carry over to the more general weak simulation as well.

Interesting open problems concern ``asymmetric'' generalizations, where one of
the input systems allows zero-tests, i.e., is a \textOCA.
In \cite{AAHMNT2014} we showed that strong simulation between $\OCA$ and $\OCN$
is semilinear and thus decidable.
However, the proof of semilinearity is not effective, so computability of the
relation as well as the complexity of its membership problem remains open.
Apart from the obvious \PSPACE\ lower bounds, not much is known about simulation
between \OCN\ and \OCA, as well as for the weak simulation problems in either
way.
It is worth mentioning that further generalizations
(\PDA\ vs.\ \OCN,
\OCA\ vs.\ \OCA, as well as 
\OCN\ vs.\ 2-dimensional \VASS)
are already undecidable \cite{AAHMNT2014}.

Another direction for further research is to establish the exact complexity
of strong/weak simulation for \OCN\ with binary encoded
increments and decrements on the counter. Trivially,
the \PSPACE-lower bound applies for this model and an \EXPSPACE\ upper bound
follows from the results of this paper with the observation that
these more expressive nets can be unfolded into ordinary \OCN\ with
an exponential blow-up.

\bibliographystyle{plain}
\bibliography{references}

\end{document}